\documentclass[sigconf, nonacm, pdfa]{acmart}
\title[Fast and Space-Efficient Parallel Algorithms for Influence Maximization]{Fast and Space-Efficient Parallel Algorithms \\ for Influence Maximization}

\newcommand\vldbdoi{10.14778/3632093.3632104}
\newcommand\vldbpages{400 - 413}
\newcommand\vldbvolume{17}
\newcommand\vldbissue{3}
\newcommand\vldbyear{2023}
\newcommand\vldbauthors{\authors}
\newcommand\vldbtitle{\shorttitle}
\newcommand\vldbavailabilityurl{https://github.com/ucrparlay/Influence-Maximization}
\newcommand\vldbpagestyle{empty}

\setcopyright{none}
\usepackage{booktabs}   %% For formal tables:
\usepackage{subcaption} %% For complex figures with subfigures/subcaptions
\usepackage[font=small,labelfont=bf]{caption}
\usepackage{colortbl}

\usepackage{style}
\usepackage{bigstrut}
\usepackage{microtype}
\usepackage{tabularx}
\usepackage{multirow}
\usepackage{multicol}
\usepackage{rotating}
\usepackage{lipsum}
\usepackage{balance}
\usepackage{mfirstuc}%\captalisewords
\usepackage{titlecaps}%\titlecap
\usepackage{listings}
\usepackage{float}
\usepackage[rightcaption]{sidecap}
\usepackage{wrapfig}
\usepackage{tcolorbox}
\usepackage{soul}
\usepackage{blindtext}
\usepackage{titlesec}
\usepackage[a-2b]{pdfx}
\hide{
 \titlespacing*{\section}
 {0pt}{.3em}{0.3em}
 \titlespacing*{\subsection}
 {0pt}{.5em}{0.3em}
}

\titleformat{\section}
  {\normalfont\Large\bfseries}{\thesection}{1em}{\MakeUppercase}

\usepackage{mathtools}
\usepackage{cleveref}

\usepackage{etoolbox}

\makeatletter
\patchcmd{\@algocf@start}% <cmd>
  {-1.5em}% <search>
  {0pt}% <replace>
  {}{}% <success><failure>
\makeatother

\def\fullversion{}
\newcommand{\codeskip}{{\vspace{.05in}}}

\definecolor{mypurple}{RGB}{142,38,9}
\newcommand{\revision}[1]{#1}
\definecolor{forestgreen}{rgb}{0.13, 0.55, 0.13}
\definecolor{mordantred19}{rgb}{0.68, 0.05, 0.0}
\definecolor{green(ryb)}{rgb}{0.4, 0.69, 0.2}
\definecolor{green(html/cssgreen)}{rgb}{0.0, 0.5, 0.0}
\newcommand{\knn}{$k$-NN\xspace}

\newcommand{\rate}{\alpha}

\newcommand{\implementationname}[1]{\textit{#1}}
\newcommand{\staticgreedy}{\implementationname{StaticGreedy}\xspace}
\newcommand{\simplegreedy}{\implementationname{GeneralGreedy}\xspace}
\newcommand{\mixedgreedy}{\implementationname{MixGreedy}\xspace}
\newcommand{\infusersimple}{\implementationname{Infuser}\xspace}
\newcommand{\infuser}{\implementationname{InfuserMG}\xspace}
\newcommand{\ripples}{\implementationname{Ripples}\xspace}
\newcommand{\NoSingles}{\implementationname{No-Singles}\xspace}

\newcommand{\BST}{\implementationname{BST}\xspace}
\newcommand{\wintree}{\implementationname{Win-Tree}\xspace}
\newcommand{\WinTree}{\wintree}
\newcommand{\ourtree}{\wintree}
\newcommand{\ptree}{\implementationname{P-tree}\xspace}
\newcommand{\oursystem}{\implementationname{PaC-IM}}
\newcommand{\ours}[1]{{Ours}$_{#1}$}

\newcommand{\sketch}[1]{\Phi_{#1}}

\newcommand{\marginal}{\textsc{Marginal}}
\newcommand{\ccsize}{\textit{size}}
\newcommand{\lbl}{\textit{label}}
\newcommand{\compactcc}{\textsc{CompactSketch}}
\newcommand{\getcenter}{\textsc{GetCenter}}
\newcommand{\sketchfunc}{\textsc{Sketch}}

\newcommand{\numcenters}{\rho}
\newcommand{\CC}[1]{\mathcal{C}_{{#1}}}
\newcommand{\influence}{\sigma}
\newcommand{\truegain}[1]{\Delta(#1)}
\newcommand{\gain}[1]{\Delta(#1)}
\newcommand{\lazygain}[1]{\bar{\Delta}[#1]}
\newcommand{\lazygaini}[2]{\bar{\Delta}_{#1}[#2]}

\newcommand{\lazygainc}[1]{\bar{\Delta}_{CELF}[#1]}
\newcommand{\lazygaint}[1]{\bar{\Delta}_{BST}[#1]}
\newcommand{\vvv}[1]{v_{#1}}
\newcommand{\uuu}[1]{u_{#1}}

\newcommand{\deltas}{\Delta^*}

\newcommand{\gainsimple}{\Delta}
\newcommand{\nextseed}{\textsc{NextSeed}}

\newcommand{\numseeds}{k}

\newcommand{\modelop}[1]{\texttt{#1}}
\newcommand{\forkins}{\modelop{fork}}

\newcommand{\thread}{thread}

\newcommand{\WriteMax}{\mf{WriteMax}\xspace}
\newcommand{\writemax}{\WriteMax}

\mathchardef\sdash="2D
\newcommand{\polylog}{\text{polylog}}
\newcommand{\true}{\emph{true}}
\newcommand{\false}{\emph{false}}

\newcommand{\False}{\textsc{False}\xspace}

\newcommand{\lleft}{\mathit{left}}
\newcommand{\rright}{\mathit{right}}
\newcommand{\renew}{\mathit{renew}}
\newcommand{\stale}{\mathit{stale}}
\newcommand{\parent}{\mathit{parent}}
\newcommand{\rt}{\mathit{root}}
\newcommand{\globalmax}{\deltas}

  \newcommand{\ifconference}[1]{{{\ifx\fullversion\undefined{#1}\fi}}}
  \newcommand{\iffullversion}[1]{{{\ifx\conference\undefined{#1}\fi}}}

\hide{
}

\makeatletter
\def\dfnt@space@setup{%
\dfnt@preskip=\parskip
  \dfnt@postskip=0pt}
\makeatother

\newtheoremstyle{exampstyle}
{.05in} % Space above
{.05in} % Space below
{} % Body font
{.5em} % Indent amount
{\sc \bfseries} % Theorem head font
{.} % Punctuation after theorem head
{.5em} % Space after theorem head
{} % Theorem head spec (can be left empty, meaning `normal')
\theoremstyle{exampstyle} 
\theoremstyle{exampstyle} 
\theoremstyle{exampstyle} 

\makeatletter
\renewenvironment{proof}[1][\proofname]{\par
\vspace{-\topsep}% remove the space after the theorem
\pushQED{\qed}%
\normalfont
\topsep0pt \partopsep0pt % no space before
\trivlist
\item[\hskip\labelsep
      \itshape
  #1\@addpunct{.}]\ignorespaces
}{%
\popQED\endtrivlist\@endpefalse
\addvspace{3pt plus 3pt} % some space after
}

 \crefname{section}{Sec.}{Sec.}
 \crefname{theorem}{Thm.}{Thm.}
 \crefname{lemma}{Lem.}{Lem.}
 \crefname{corollary}{Col.}{Col.}
 \crefname{table}{Tab.}{Tab.}
 \crefname{algorithm}{Alg.}{Alg.}
 \crefname{figure}{Fig.}{Fig.}
 \crefname{fact}{Fact}{Fact}
\Crefname{table}{Tab.}{Tab.}
\crefname{problem}{Problem}{Problem}

\definecolor{commentgreen}{RGB}{0,128,0}

\SetCommentSty{mycommfont}

\begin{document}
\fancyhead{}
% \ifconference{
% \input{cover_page}
% }

% \iffullversion{
  \balance
% }


\author{Letong Wang}
\affiliation{%
  \institution{UC Riverside}
  % \city{Riverside}
}
\email{lwang323@ucr.edu}

\author{Xiangyun Ding}
\affiliation{%
  \institution{UC Riverside}
  % \city{Riverside}
}
\email{xding047@ucr.edu}

\author{Yan Gu}
\affiliation{%
  \institution{UC Riverside}
  \country{}
}
\email{ygu@cs.ucr.edu}

\author{Yihan Sun}
\affiliation{%
  \institution{UC Riverside}
  \country{}
}
\email{yihans@cs.ucr.edu}



\begin{abstract}
Influence Maximization (IM) is a crucial problem in data science.
The goal is to find a fixed-size set of highly influential \emph{seed} vertices on a network
to maximize the influence spread along the edges.
%Since IM was formalized in 2013 (xx), numerous paper has been focusing on this topic.
While IM is NP-hard on commonly used diffusion models, a greedy algorithm can achieve $(1-1/e)$-approximation by repeatedly selecting the vertex with the highest \emph{marginal gain} in influence as the seed. 
%Due to theoretical guarantees, rich literature focuses on improving the performance of the greedy algorithm. 
However, we observe two performance issues in the existing work that prevent them from scaling to today's large-scale graphs: space-inefficient memorization to estimate marginal gain, and time-inefficient seed selection process due to a lack of parallelism. 
%To estimate the marginal gain, existing work either runs Monte Carlo (MC) simulations of influence spread or pre-stores hundreds of \emph{sketches} (usually per-vertex information). However, these approaches can be inefficient in time (MC simulation) or space (storing sketches), preventing the ideas from scaling to today's large-scale graphs.
%Due to theoretical guarantees, rich literature focuses on improving the performance of the greedy algorithm. 

This paper significantly improves the scalability of IM using two key techniques. The first is a \emph{sketch-compression} technique for the independent cascading model on undirected graphs. It allows combining the simulation and sketching approaches to achieve a time-space tradeoff. The second technique includes new data structures for parallel seed selection. Using our new approaches, we implemented \oursystem{}: \underline{Pa}rallel and \underline{C}ompressed IM.

We compare \oursystem{} with state-of-the-art parallel IM systems on a 96-core machine with 1.5TB memory.
\oursystem{} can process the ClueWeb graph with 978M vertices and 75B edges in about 2 hours.
%On tested graphs with more than 10 million edges,
On average, across all tested graphs, our uncompressed version is 5--18$\times$ faster and about 1.4$\times$ more space-efficient than existing parallel IM systems.
Using compression further saves 3.8$\times$ space with only 70\% overhead in time on average.
%1.7$\times$ with 3.8$\times$ improve 3.2---32$\times$ faster while using 3.3--10$\times$ less space.
%The gap is up to 997$\times$ against \ripples{} although \ripples{} also supports directed graphs.
%With compression, , while no other systems can finish within 5h. On all tested graphs, including several billion-scale graphs, \oursystem{} is always better than all baselines in both time and memory.
%\vspace{-1.5em}
\end{abstract} 


\maketitle

\ifconference{
\pagestyle{\vldbpagestyle}
\begingroup\small\noindent\raggedright\textbf{PVLDB Reference Format:}\\
\vldbauthors. \vldbtitle. PVLDB, \vldbvolume(\vldbissue): \vldbpages, \vldbyear.\\
\href{https://doi.org/\vldbdoi}{doi:\vldbdoi}
\endgroup
\begingroup
\renewcommand\thefootnote{}\footnote{\noindent
This work is licensed under the Creative Commons BY-NC-ND 4.0 International License. Visit \url{https://creativecommons.org/licenses/by-nc-nd/4.0/} to view a copy of this license. For any use beyond those covered by this license, obtain permission by emailing \href{mailto:info@vldb.org}{info@vldb.org}. Copyright is held by the owner/author(s). Publication rights licensed to the VLDB Endowment. \\
\raggedright Proceedings of the VLDB Endowment, Vol. \vldbvolume, No. \vldbissue\ %
ISSN 2150-8097. \\
\href{https://doi.org/\vldbdoi}{doi:\vldbdoi} \\
}\addtocounter{footnote}{-1}\endgroup

\ifdefempty{\vldbavailabilityurl}{}{
\vspace{.3cm}
\begingroup\small\noindent\raggedright\textbf{PVLDB Artifact Availability:}\\
The source code, data, and/or other artifacts have been made available at \url{\vldbavailabilityurl}.
\endgroup
}
}

% \hide{
% \setlength{\abovedisplayskip}{1pt}
% \setlength{\belowdisplayskip}{1pt}
% \setlength{\abovedisplayshortskip}{1pt}
% \setlength{\belowdisplayshortskip}{1pt}

% \setlength\abovecaptionskip{0.2em}
% \setlength\belowcaptionskip{0.3em}
% \setlength{\floatsep}{0.5em}
% % this is to change the space below tables/figures
% \setlength{\textfloatsep}{0.5em}
% % this is to change the space above tables/figures
% \setlength{\intextsep}{0.5em}
% \setlength{\dbltextfloatsep}{1em}
% \setlength{\dblfloatsep}{.2em}
% }
\vspace{-0.2mm}
% \vspace{-.5in}
% \vspace{-.2mm}
\section{Introduction}\label{sec:intro}
% \hide{\letong{Give application examples to motivate computing IM on undirected graphs. Such as in AI, active learning is a semi-supervised learning method. How to select limited data points for human labeling to get better training data. Data points that have similar features are likely to have the same label, and we build k-NN graphs for data points. On such graphs, the influence is undirected.
% }
% }
%Given a graph G=(V,E)G=(V,E), the Influence Maximization (IM) problem
Influence Maximization (IM) is a crucial problem in data science.
The goal is to find a fixed-size set of highly influential \emph{seed} vertices on a network
to maximize the spread of influence along the edges.
%under stochastic cascade models of propagation.
%It formalizes the
For example, in viral marketing, the company may choose to send free samples to a small set of users in the hope of triggering a large cascade of further adoptions through the ``word-of-mouth'' effects.
Given a graph $G=(V,E)$ and a stochastic \emph{diffusion model} to specify how influence spreads along edges, we use $n=|V|$, $m=|E|$, and $\influence(S)$ to denote the expected influence spread on $G$ using the seed set $S\subseteq V$. %\yan{reword}
The IM problem aims to find a seed set $S$ with size $\numseeds$ to maximize $\influence(S)$.
Given its importance, IM is widely studied, and we refer the audience to a list of surveys~\cite{banerjee2020survey, AACGS10,zhou2021survey} that reviews the numerous applications and a few hundred papers on this topic.

Among various diffusion models, Independent Cascade (IC)~\cite{goldenberg2001talk} (defined in \cref{sec:prelim}) is one of the earliest and most widely used.
%This model assumes each edge between a pair of vertices propagates information independently.
%In the beginning, all seeds are \defn{activated}. In each timestamp, each vertex vv that is newly activated in the last round will activate its neighbors uu with a probability pvup_{vu}.
% In this model, all seeds are initially \defn{activated}.
In IC, only seeds are \defn{active} initially.
In each timestamp, each vertex~$v$ that is newly activated in the last timestamp will activate its neighbors $u$ with a probability $p_{vu}$.
%Kempe et al.~\cite{kempe2003maximizing} proved NP-hardness to obtain an optimal solution under IC.
%Fortunately, the problem exhibits the \emph{submodular} property (see details in \cref{sec:prelim}), which allows for a simple greedy algorithm with ($1-1/e$)-approximation~\cite{kempe2003maximizing}.
% Although IM is NP-hard on the IC model~\cite{kempe2003maximizing}, it exhibits the \emph{submodular} property (see \cref{sec:prelim}), which allows for a greedy algorithm with ($1-1/e$)-approximation~\cite{kempe2003maximizing}.
Although IM is NP-hard on IC~\cite{kempe2003maximizing}, the monotone and submodular properties of IC allow for a greedy algorithm with ($1-1/e$)-approximation~\cite{kempe2003maximizing}.
%Although obtaining an optimal solution under this model is proved to be NP-hard~\cite{kempe2003maximizing} by Kempe et al, in the same paper they provided a simple greedy algorithm that yields a (1−1/e1-1/e)-approximation solution due to the \emph{submodular} property of the problem.
Given the current seed set $S$, the greedy algorithm selects the next seed as the vertex with the highest \emph{marginal gain},
i.e., $\arg \max_{v\in V} \{\influence(S\cup \{v\})-\influence(S)\}$.
Due to the theoretical guarantee, this greedy strategy generally gives better solution quality than other heuristics~\cite{li2018influence}.
However, the challenge lies in estimating the influence $\influence(S)$ of a seed set $S$.
Early work uses Monte-Carlo (MC) experiments by averaging $R'$ rounds of influence diffusion simulation~\cite{kempe2003maximizing,leskovec2007cost},
%However, to guarantee accuracy, the value R′R' is usually large, leading to high running time.
% However, the solution quality relies on using a large value of $R'$, which can be expensive.
%However, the solution quality relies on the value of $R'$, which is usually around $10^4$.
but the solution quality relies on a high value of $R'$ (usually around $10^4$).
%Despite the provably good approximation, this approach is expensive due to the high simulation cost, which forbids the approach from being adopted
%to today's large-scale networks with hundreds of millions of vertices and billions of edges.
%There has also been a large body of proxy-based methods proposed ~\cite{chen2010scalable, jung2012irie, kim2013scalable, galhotra2016holistic} that replaces
%MC simulations with more efficient heuristic approaches, but their solution quality can be unstable due to the lack of theoretical guarantee~\cite{li2018influence}.
%These methods replace the MC simulations with heuristic influence estimation and are more scalable on graph size, but the quality of their solutions is unstable due to lack of theoretical guarantee~\cite{li2018influence}.
%To keep high solution quality by MC-simulations while accelerate the MC-simulation-based greedy algorithms while keeping the high quality of the solutions,
\hide{To avoid MC experiments, \defn{sketch-based} approaches have been widely used~\cite{chen2009efficient, cheng2013staticgreedy, ohsaka2014fast, borgs2014maximizing,cohen2014sketch,tang2014influence,tang2015influence}.
}
Later work uses \defn{sketch-based} approaches~\cite{chen2009efficient, cheng2013staticgreedy, ohsaka2014fast, borgs2014maximizing,cohen2014sketch,tang2014influence,tang2015influence} to avoid MC experiments.
%The high-level idea is to settle down a fixed set of RR sampled subgraphs as \defn{sketches}.
%The high-level idea is to pre-store the information of RR rounds of MC simulations in RR \defn{sketches},
%often corresponding to RR sampled graphs (e.g., each edge is chosen with the probability in which influence can spread along this edge).
%, each reflecting the information of an MC simulation on GG.
%In many cases, a sketch is a sampled graph of GG where each edge is selected with the probability that the influence is diffused through this edge~\cite{}.
%The high-level idea is to settle down a fixed set of RR sampled subgraphs of GG (i.e., each edge is chosen with the probability in which influence can spread along this edge), and store RR \defn{sketches} based on the sampled graphs. These sketches can either be the sampled graphs, or more pre-computed information of these sampled graphs.
Such algorithms pre-store $R$ \defn{sketches}. Each sketch is a sampled graph---an edge $(v,u)$ is chosen with probability $p_{vu}$.
%with the probability in which influence can spread along this edge.
%The sketches .
%When estimating the influence, the diffusion is only simulated on the sampled graphs, which act as the results of the MC experiments.
When estimating $\influence(S)$, the sampled graphs are used as the results of the MC experiments of influence diffusion.
% Many follow-up work continued in this direction using the same high-level idea, i.e., to record the influence spread on pre-computed sketches to avoid explicit simulation.
In an existing study~\cite{cheng2013staticgreedy}, using $R\approx 200$ sketches achieves a similar solution quality to $R'=10^4$ MC experiments, greatly improving efficiency. %rendering roughly two orders of magnitude performance gain.
%Empirical studies show that the simulation using a fixed set of $R$ sketches converges much faster with the number of sketches $R$ than that with the number of MC experiments~\cite{cheng2013staticgreedy}.
The sketches can either be the sampled graphs and/or \emph{memoizing} more information from the sampled graphs to accelerate influence computation, such as connectivity~\cite{chen2009efficient,gokturk2020boosting} or strong connectivity~\cite{ohsaka2014fast}.
%For example, on undirected graphs, a vertex's influence on a sketch is all vertices in the same connected component (CC).
%Therefore, many algorithms~\cite{mixgreedy,infuser}\yan{add} memoize the CC information as the sketch to accelerate influence computation.
%allows for much more rapid convergence of the MC simulations,
%greatly reducing the number of simulations needed~ (hundreds vs. tens of thousands on million-scale networks) without loss of accuracy. %\xiangyun{In fact, RIS and IMM are also sketch-based algorithm using precomputed samples}

While numerous sketch-based solutions have been developed, we observed great challenges in scaling them to today's large-scale graphs.
%Most of the sketch-based algorithms are not scalable to graphs with more than tens of millions of vertices.
\hide{
In a SIGMOD'17 benchmark paper~\cite{arora2017debunking},
the largest tested graph with (sequential) sketch-based algorithms on the IC model is the Twitter graph (42M vertices and 1.4B edges)~\cite{kwak2010twitter}.
Among nine tested algorithms, only one can process Twitter, but it takes 10 hours.
%Since then, some parallel algorithms have been recently designed~\cite{gokturk2020boosting,popova2018nosingles,minutoli2019fast, minutoli2020curipples}, but they still need more than 1.5 hours to process Twitter on a 96-core machine (See Table~\ref{tab:baselines}).
Even the recent parallel algorithms~\cite{gokturk2020boosting,popova2018nosingles,minutoli2019fast, minutoli2020curipples} need more than 10 minutes to process Twitter on a 96-core machine (See Table~\ref{tab:baselines}).
}
In a benchmark paper~\cite{arora2017debunking} on nine state-of-the-art (SOTA) sequential IM solutions,
none of them can process the Friendster (FT) graph~\cite{kwak2010twitter} with 65M vertices and 3.6B edges due to timeout (more than 40 hours) or out-of-memory.
%Since then, some parallel algorithms have been recently designed~\cite{gokturk2020boosting,popova2018nosingles,minutoli2019fast, minutoli2020curipples}, but they still need more than 1.5 hours to process Twitter on a 96-core machine (See Table~\ref{tab:baselines}).
Even the recent parallel algorithms~\cite{gokturk2020boosting,popova2018nosingles,minutoli2019fast, minutoli2020curipples} need more than half an hour to process FT on a 96-core machine (See Table~\ref{tab:baselines}).
%Note that Friendster is \emph{not} a very large graph---today's real-world networks can easily reach tens of billions of edges.
Two major challenges exist to scale sketch-based approaches to billion-scale graphs.
The first is the \emph{space}.
%Even though each sketch is a subgraph of $G$, we usually need to store some per-vertex information.
Storing each sketch usually needs per-vertex information.
This indicates $O(Rn)$ space, which is expensive on large graphs (empirically, $R$ is a few hundred).
%and therefore
%the total size of hundreds of sketches prohibits the algorithms from running on large-scale graphs.
%Since we still need to store hundreds of sketches for good solution quality, the total size prohibits the algorithms from running on large-scale graphs.
The second is \emph{insufficient parallelism}.
Many SOTA IM solutions use the \emph{CELF}~\cite{leskovec2007cost}
% \cite{leskovec2007cost, goyal2011celf++, chen2009efficient, gokturk2020boosting,minutoli2019fast}
optimization for seed selection (see details below), which is inherently sequential.
%In particular, the seed selection process in state-of-the-art IM solutions~\cite{chen2009efficient, gokturk2020boosting}
%all suffer from low parallelism (see discussions below), as shown in \cref{fig:xx}, especially.
%long processing time in both sketch construction and seed selection.
%Even with optimizations such as CELF~\cite{leskovec2007cost},
%the seed selection process can still be slow, due to the large number of candidate vertices to be re-evaluated in each round.
%the nature of greedy seed selection process is

This paper takes a significant step to \defn{improve the scalability of sketch-based IM solutions} and tests the algorithms on \defn{real-world billion-scale graphs}.
%Throughout the paper, we assume $G$ is undirected and consider the IC model unless otherwise specified.
 %on undirected graphs and the IC model
%Our algorithm substantially improves existing solutions in both space and time consumption.
We propose two techniques to improve both space and time.
% Our first solution is a \defn{sketch compression} technique for \emph{undirected graphs and the IC model},
The first is a \defn{sketch compression} technique for the \emph{IC model on undirected graphs},
which limits the auxiliary space by a user-defined capacity to reduce space usage.
%where each sketch size is sublinear to the number of vertices in the graph
%with a user-defined parameter \numcenters\numcenters to control the compression ratio.
Our second technique is \defn{parallel data structures} for seed selection to reduce running time, which works on general graphs and any diffusion model with submodularity.
%To improve the runtime, we proposed various \defn{parallel algorithms and data structures} for both computing the sketches and selecting seeds.
Combining the new ideas, we implemented \textbf{\oursystem{}: \underline{Pa}rallel and \underline{C}ompressed IM}.
We show a heatmap in \cref{fig:heatmap} to overview our results.
On the aforementioned FT graph, \oursystem{} only uses 128 seconds without compression (using 2.5$\times$ auxiliary space on top of the input graph), or 609 seconds when limiting auxiliary space in 0.45$\times$ input size, %$5\times$ smaller auxiliary space,
using a 96-core machine.
\oursystem{} is at least $15\times$ faster than existing parallel solutions while using much less space and achieving the same solution quality (see \cref{tab:baselines}).
%processes  in two minutes without any compression,
%and in six minutes with a compression that reduces the space usage to 50\%.
Below, we overview the key contributions of this paper. %two key techniques proposed in this paper.

\begin{figure}
  \centering
  % \vspace{-.5em}
  \includegraphics[width=0.9\columnwidth]{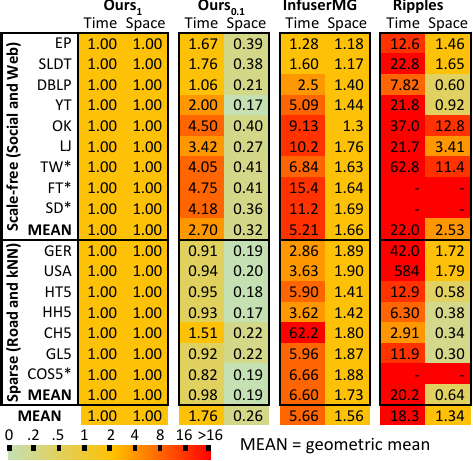}
  \caption{\textbf{Heatmap of relative running time and space usage, normalized to \ours{1}.}
  \ours{1}: \oursystem{} with no compression. \ours{0.1}: \oursystem{} with $10\times$ sketch compression.
  \infuser{}~\cite{gokturk2020boosting} and \ripples{}~\cite{minutoli2019fast}: existing parallel IM systems.
  Lower/green is better. The graph information is in \cref{tab:graph_info}.
  The running times are in \cref{tab:baselines}. $*$: graphs with more than a billion edges.
  %(Colors are adjusted to be distinguishable in gray-scale.)
  \label{fig:heatmap}}
  % \vspace{-.5em}
\end{figure}  

Our first contribution is a \textbf{compression scheme for sketches} on undirected graphs and the IC model, which allows for user-defined compression ratios (details in \cref{sec:compact_sketching}).
%Such compression allows for any user-defined compression ratio.
%Inspired by the observations from~\cite{chen2009efficient}, G\"okt\"urk and Kaya~\cite{gokturk2020boosting} showed that storing per-vertex connectivity information on all sketches can avoid the time-consuming MC simulation. %on the sketches (sampled graphs).
%Our idea is inspired by the following fact observed by previous papers~\cite{chen2009efficient,gokturk2020boosting}:
%all vertices activated by a vertex $v$ on a sketch must be the vertices in the same connected component (CC) of $v$ on this sketch. Therefore, we can pre-compute the connectivity information for each vertex on each sketch to accelerate influence estimation.
%\letong{NewGreedy and MixGreedy found that on undirected graphs, the spread of each vertex is the average number of vertices it can reach on sampled graphs.
%G\"okt\"urk and Kaya~\cite{gokturk2020boosting} then showed that instead of explicitly storing the sampled graphs, they could infuse the process of sampling graphs and simulation by xxx. }
%While this approach does not need to store sampled graphs, they still need to store the number of vertices each vertex can reach on each sample, which
%However, it requires $O(Rn)$ auxiliary storage where $n$ is the number of vertices and $R$ is at least a few hundred in practice---impractical when $n$ is large.
%Recall that to estimate the influence on a sketch, one can simulate the influence spread and count the number of activated
%full influence spread simulations when in influence estimation.
Similar to existing work~\cite{chen2009efficient,gokturk2020boosting}, \oursystem{} memoizes connected components (CC) of the sketches
%Our approach also uses the CC information of each sketch to avoid full simulations of spread when computing the influence,
but \emph{avoids the $O(Rn)$ space to store per-vertex information}.
%A straightforward implementation of this idea requires $O(Rn)$ auxiliary storage, which is expensive for large graphs.
%Considering $R$ is at least a few hundred, this can be impractical when $n$ is large.
%Our technique is an algorithmic framework that allows only storing \emph{partial} vertex information.
%Our new technique only stores \emph{partial} vertex information in each sketch.
Our idea is a combination (and thus a tradeoff) of memoization and simulation.
The idea is to memoize the CC information only for \emph{centers} $C\subseteq V$, where $|C|=\alpha n$, and $\rate\in[0,1]$ is a user-defined parameter. 
% The influence information of a non-center vertex will be retrieved by a local simulation.
A local simulation will retrieve the influence information of a non-center vertex.
%We design parallel algorithms to construct such compressed sketches.
\hide{
While our solution is sophisticated,
%the high-level idea is to retrieve the CC information during seed selection \emph{on-the-fly} from the partial CC information we maintained.
the high-level idea is to run a local simulation on the sketch during influence estimation until a center is seen, 
and compute the influence by the center's CC information stored in the sketch.
As such, we can achieve a tradeoff between time and space by controlling the value of $\numcenters$ (or $\rate$).
}
Theoretically, we show that we can limit the auxiliary space by a factor of $\rate$ by increasing the time by a factor of $O(1/\rate)$. %, and experimentally verify this in \cref{fig:compression}.
Experimentally, such tradeoff is studied in \cref{fig:compression}.
%, with some prior techniques reviewed in \cref{sec:prelim}

%The second contribution of this paper is various new parallel algorithms and data structures.
%In addition to the parallel construction of the compact sketches, we also highlight the parallel priority queues in the seed selection phase.
Our second contribution is \textbf{two new parallel data structures for seed selection}.
Recall that many SOTA IM solutions~\cite{chen2009efficient, gokturk2020boosting, minutoli2019fast, kim2013scalable, cheng2013staticgreedy} use the CELF optimization~\cite{leskovec2007cost} (details in \cref{sec:prelim}) for seed selection. 
%, utilizing the \emph{submodularity} of the problem.
In a nutshell, CELF is an iterative approach that lazily evaluates the marginal gain of vertices in seed selection, one at a time.
While laziness reduces the number of vertices to evaluate,
CELF is inherently sequential. 
%, which limits the parallelism of existing IM systems.
We proposed two novel solutions that achieve high parallelism for CELF. %without introducing much additional work.
%The key algorithmic challenge is to identify the ``correct'' vertices to re-evaluate in highly asynchronous settings.
% The key challenge is to identify more vertices to evaluate in parallel,
% while remaining lazy so that most of the ``unpromising'' vertices are untouched as in CELF.
The challenge is evaluating more vertices in parallel while avoiding "unpromising" vertices as in CELF.
Our first solution is a binary search tree (BST) called \ptree{}~\cite{sun2018pam,blelloch2016just,blelloch2022joinable} (\cref{sec:BST_section}). We highlight its \emph{theoretical efficiency} (\cref{lemma:connectness,lemma:efficiency,thm:ptreecost}).
Our second solution is referred to as \wintree{} (\cref{sec:WinTree}), which has lower space usage,
leading to slightly better overall performance.
The two solutions work on both directed and undirected graphs and any diffusion model with submodularity. They are potentially extendable to other optimization problems with submodular objective functions (see discussions in \cref{sec:conclusion}).
%(which is true for most diffusion models of interest, see details in \cref{sec:selecting})
% Furthermore, we believe they apply to general optimization problems with submodular objective functions. %with CELF optimization.
%For instance, the new data structures can be used in other submodular optimization problems such as \yan{cite}.
%We leave this as future work, and discuss it briefly in \cref{sec:conclusion}.

\hide{
 focusing on theoretical and practical efficiency, respectively.

Recall that the greedy algorithm always selects the vertex with the highest marginal gain as the next seed. Instead of trying all remaining vertices, many state-of-the-art IM solutions~\cite{chen2009efficient, gokturk2020boosting} use the CELF optimization~\cite{leskovec2007cost} (more details in \cref{sec:prelim}), utilizing the \emph{submodularity} of the problem.
%CELF keeps a priority queue $Q$ that maintains the marginal gain of all vertices with \emph{lazy} updates.
CELF evaluates vertices' marginal gain in decreasing order of their stale (lazily-evaluated) marginal gain,
%re-evaluate $v$'s marginal gain, %until certain rules are satisfied (see \cref{alg:CELF_framework}).
and checks an early stop condition (see \cref{alg:CELF_framework}) after each re-evaluation, such that not all vertices need to be re-evaluated.
%This process sometimes only iterates a few times, but in certain rounds (usually the earlier ones) it can repeat for $O(n)$ iterations.
%In many cases, CELF can greatly reduce the number of re-evaluated vertices.
Since this idea is inherently sequential, existing parallel solutions~\cite{chen2009efficient, gokturk2020boosting} left the CELF process sequential, which limits parallelism.
To tackle this challenge, we propose two new solutions focusing on theoretical and practical efficiency, respectively.
%At a high level, both of them can apply the lazy updates and in parallel without introducing much overhead.
Both of them carefully choose batches of vertices to evaluate in parallel without introducing much overhead.
Our first solution is based on a parallel binary search tree (BST) called \ptree{}~\cite{sun2018pam,blelloch2016just,blelloch2022joinable}.
We prove that the total number of re-evaluated vertices by our BST-based approach is asymptotically the same as the standard CELF, while enabling good parallelism.
Our second solution is based on simpler but potentially more practical winning trees proposed in this paper, called \wintree{}.
Our experiments show that both data structures achieve high parallelism,
while the winning tree performs slightly better due to its simpler structure.
%asynchronous algorithms and more efficient construction.
We present more details in \cref{sec:selecting}.
}
\hide{
We believe that both our compact skecth technique and the seed selection data structures are of independent interest within and beyond the IM problem. Our compact sketch algorithms provide a general interface to re-evaluate the new marginal gain of a vertex and mark a vertex as a seed, which can be combined with any seed selection algorithm. Our seed selection data structures can be combined with any influence estimate algorithm and find the next seed based on the greedy strategy. In addition, it is general to other submodule functions with CELF optimization.
}
\hide{
Although our sketch compression scheme focuses on IC model and undirected graphs, we believe that the algorithmic insights are of independent interest.
Our seed selection data structures work on both directed and undirected graphs, and any diffusion model with  submodularity. %, and is
Our new data structures are also potentially general to other submodular functions with CELF optimization.
%For instance, the new data structures can be used in other submodular optimization problems such as \yan{cite}.
We leave this as future work, and discuss it briefly in \cref{sec:conclusion}.
}

%Our final contribution is an in-depth experimental study of the proposed techniques, and the performance comparison among \oursystem{} and SOTA parallel IM systems.
We experimentally study the performance of \oursystem{} and compare it with SOTA parallel IM systems.
We tested 17 graphs, five of which have over a billion edges.
Besides social networks,
we also tested other real-world graphs, including web graphs, road graphs, and \knn{} graphs.
One can consider IM on such graphs as studying the influence diffusion among webpages and geologically or geometrically close objects.
%, and the largest graph is the ClueWeb graph with 978M vertices and 75B edges.
%We compare \oursystem{} with state-of-the-art parallel algorithms.
%Our implementation is a faithful implementation of the \textsc{StaticGreedy} algorithm, which provides better or comparable solution quality to other algorithms.
On all tested graphs, \oursystem{} achieves the best running time and space usage while guaranteeing comparable or better solution quality to all baselines.
Compared to the best baseline, \oursystem{} with no compression is $5.6\times$ faster and is 1.5$\times$ more space-efficient on average (geometric mean across tested graphs), and is $3.2\times$ faster and uses $6\times$ less space using compression with $\rate=0.1$.
Due to space- and time-efficiency, \oursystem{} is the only system to process the largest graph ClueWeb~\cite{webgraph} with 978M vertices and 75B edges.
%The best among the existing algorithms could only process xx out of xx graphs, and failed due to timeout and/or out-of-memory.
We believe \oursystem{} is the first IM solution that scales to tens of billions of edges and close to a billion vertices.

\iffullversion{We publish our code at \cite{wang2023fastcode}.}
\ifconference{For page limit, we provide the full version of this paper~\cite{wang2023fast} with full proofs and more experimental evaluations.}

\hide{
%In particular, in each sketch, we only store the influence information for a small number of \emph{center} vertices.
Our idea is based on an existing observation~\cite{} that on the sampled graph, a vertex will influence all vertices in the same connected component (CC),
and thus the CC size for each vertex $v$ can be stored as $v$'s influence on this sketch.
If any vertex in this CC is selected as a seed, the (marginal) influence of all vertices in this CC can be set as 0 on this sketch.
Different from previous work~\cite{}, we only store this information for $\numcenters$ \emph{center} vertices $C\subseteq V$, such that each sketch has $o(n)$ size.
Our idea of compressing sketches is inspired by recent theory work on sublinear CC representation~\cite{}.
To compute the influence of a vertex $v$ on a compressed sketch, we will start a simulation from $v$ on the fly using breadth-first-search (BFS), but stop BFS once a center $c\in C$ is encountered.
With randomly selected centers, a vertex connected with a center will likely encounter a center in a few hops, and $v$ has the same influence as $c$ on this sketch.
If $v$ is not connected with any centers, $v$ is likely in a small CC, and the BFS will terminate quickly. In either case, we can expect a reasonable cost to compute the influence for each vertex.
By controlling the parameter $\numcenters$, we can achieve trade-off between
time and space, i.e., smaller $\numcenters$ indicates lower space usage, but a higher cost when computing influence.
This approach is a combination of the sketch-based and the simulation-based algorithms - it stores a small set of influence information to avoid complete simulation to save time,
but allows for a ``partial'' simulation
such that only a subset of vertices' influence information is needed.
To compute the connectivity of each sketch, we adopt a highly-parallelized union-find algorithm from xxx, which provided high performance in our experiments.

Our second technique is proposing parallel priority queues for CELF seed selection (see more details below). % a new data structure \ourtree{}. \ourtree{} is to address the high cost of seed selection.
Recall that the greedy algorithm repeatedly adds the vertex with the highest marginal gain to the seed set.
However, in general, selecting each seed requires evaluating the latest marginal gain of all vertices, which can be expensive.
Most state-of-the-art implementations use the CELF optimization~\cite{} (see more details in \cref{sec:prelim}) based on the submodularity of the IM problem.
%We present more details about CELF in \cref{sec:prelim}.
In general, CELF uses a priority queue $Q$ to maintain the marginal gain of all vertices with lazy updates. %and only re-evaluate the top vertices in QQ.
To select a seed, the algorithm starts with popping the top vertex $v$ from $Q$ and re-evaluating its marginal gain $\gain{v}$. If $\gain{v}$ becomes lower than the marginal gain of the next candidate $u$ in $Q$, $v$ will be inserted back to $Q$, and the algorithm repeats to pop $u$ from $Q$ and re-evaluate $u$'s marginal gain.
Although CELF greatly decreases the number of candidates to evaluate for selecting each seed, this process is inherently sequential, and parallelizing it is highly non-trivial.
As a result, all existing parallel IM algorithms~\cite{} only parallelize sketch computing, but still perform seed selection sequentially, which becomes their performance bottleneck (as we will show in our experiments).
%the top vertices from QQ will be popped and re-evaluated.
%until the actual seed
%In particular, the algorithm will repeatedly pop the top vertex vv from QQ, and re-evaluate their marginal gain based on the current seed set.
%If vv's marginal gain is still the top after re-evaluation, vv is safe to be recognized as the next seed. In most cases, the number of vertices re-evaluated is much smaller than nn.
%In CELF, all vertices with their influence are maintained in the priority queue QQ.
%In each round, the top vertex in QQ will be popped and re-evaluated.
%Only when
%accelerate seed selection.
%We note that this is highly non-trivial based on the iterative
In this paper, we propose two parallel data structures to address this challenge, focusing on theoretical and practical efficiency, respectively.
Our first solution is based on parallel binary search trees (BST).
We proved that using BST, the total number of re-evaluated vertices is asymptotically the same as the sequential CELF, while achieving good parallelism.
Our second solution is based on simpler but potentially more I/O-friendly tournament trees.
Our experimental results show that both data structures achieve high parallelism, while the tournament tree performs slightly better due to its I/O-friendliness.
%After re-evaluating the top vertex vv in QQ, if vv's performance becomes lower, we will process the two branches in the winning tree in parallel.
%We maintain a global variable as the highest marginal gain so far, and use it to prune the search in the winning tree. More details are presented in \cref{sec:winningtree}.
}

\section{Preliminaries}\label{sec:prelim}

%\myparagraph{Basic Notations}.
For graph $G=(V,E)$, we use $n = |V|$ and $m = |E|$.
Since our sketch compression is designed for undirected graphs and IC model,
throughout the paper, we assume $G$ is undirected and consider the IC model unless otherwise specified.
%The vertex set VV represents agents, and edge set EE represents connections between agents. p:E→[0,1]p: E \to [0,1] is a function mapping each edge to a real number between zero and one. In practice, p(u,v)p(u,v) is the strength that uu can influence vv \cite{kempe2003maximizing}. If GG is undirected, p(u,v)p(u,v) is equal to p(v,u)p(v,u).
A \defn{connected component (CC)} is a maximal subset in $V$ s.t. every two vertices in it are connected by a path.
In a max-priority-queue, we use \defn{top} to refer to the element with the largest key, and use function \defn{pop} to find and remove the top element.
$\tilde{O}(f(n))$ denotes $O(f(n)\cdot \polylog(n))$.

\myparagraph{Computational Model}.
%We use the work-span (or work-depth) notions to analyze the asymptotic cost for parallel algorithms on the fork-join model with binary forking~\cite{CLRS,blelloch2020optimal}.
%which is recently used in many papers on parallel algorithms~\cite{agrawal2014batching,blelloch2010low,BCGRCK08,BG04,Blelloch1998,blelloch1999pipelining,BlellochFiGi11,BST12,BBFGGMS16,dinh2016extending,xu2022efficient,xu2020parallel,blelloch2018geometry,dhulipala2020semi,BBFGGMS18,blelloch2020randomized,gu2021parallel}.
We use the fork-join parallelism~\cite{CLRS,blelloch2020optimal}, and the work-span analysis~\cite{blumofe1999scheduling,gu2021parallel,gu2022analysis}.
We assume a set of \thread{}s that access a shared memory.
A thread can \forkins{} two child software \thread{s} to work in parallel.
When both children complete, the parent process continues.
A parallel for-loop can be simulated by recursive \forkins{s} in logarithmic levels.
The \defn{work} of an algorithm is the total number of instructions, and
the \defn{span} is the length of the longest sequence of dependent instructions.
We can execute the computation using a randomized work-stealing scheduler~\cite{blumofe1999scheduling,ABP01} in practice.
%We assume unit-cost atomic operation \cas{}(p,vold,vnew)(p,v_{\mathit{old}},v_{\mathit{new}}) (or \CAS{}),
%which atomically reads the memory location pointed to by pp, and write value vnewv_{\mathit{new}} to it if the current value is voldv_{\mathit{old}}.
%It returns \true\true{} if successful and \false\false{} otherwise.
% \iffullversion{
We use \emph{atomic} operation \textsc{WriteMax}$(t,v_{\mathit{new}})$ to write value $v_{\mathit{new}}$ at the memory location $t$ if $v_{\mathit{new}}$ is larger than the current value in $t$.
% which atomically reads the memory location pointed to by $t$, and write value $v_{\mathit{new}}$ to it if $v_{\mathit{new}}$ is larger than the current value in $t$.
We use compare-and-swap to implement \textsc{WriteMax}.
% }
%\yihan{Did we use CAS? I think we only used write-max.}
%\letong{Yes, in union find}

%\vspace{-\baselineskip}
% \vspace{-.2em}
\subsection*{The Influence Maximization (IM) Problem}

\begin{table}[t]
\small
\hide{
\begin{tabular}{ll}
   $\influence_{G,M}(S)$  & The influence spread (expected number of activated \\
   or  $\influence(S)$ & vertices) of seed set $S$ on graph $G$ and diffusion model $M$\\
   \hline
   $\gain{v\,|\,S}$  & $\gain{v\,|\,}=\influence(S\cup \{v\})-\influence(S)$. \\
   or $\gain{v}$&Marginal increase of $v$ on top of $S$\\
   \hline
   $\gainsimple[v]$ & The (Lazy-evaluated) marginal gain. It may be a stale value \\
   & and is an upper bound of the true marginal gain $\gain{v\,|\,S}$ \\
   &for the current seed set $S$
\end{tabular}
}

\caption{\textbf{Notations in the paper.}\label{tab:notations}}

\rule{\columnwidth}{.05em} % \vspace{-.05in}

% \vspace{-.05in}

\begin{description}[labelwidth=.15in,leftmargin=.15in]%[labelwidth=.3in, leftmargin=.3in, itemindent=.2in]
    \item[$\boldsymbol{G=(V,E)}$]: The input graph.\hfill $\boldsymbol{\numseeds}$~~: number of seeds.\quad
    %\hfill $R$: number of sketches
    \item[$	\boldsymbol{\influence(S)}$ \textbf{or} $\boldsymbol{{\influence_{G,M}(S)}}$]: The influence spread of seed set $S\subseteq V$ on graph $G$ and diffusion model $M$.
    \item[$\boldsymbol{\truegain{v\,|\,S}}$]: The \defn{true score} (marginal gain) of $v$ on top of $S$. $\truegain{v\,|\,S}=\influence(S\cup \{v\})-\influence(S)$. We omit $S$ and use $\truegain{v}$ with clear context.
    \item[$\boldsymbol{\lazygain{v}}$]: The \defn{stale score} (lazily-evaluated marginal gain) of $v$ stored in an array. It is an upper bound of $\truegain{v\,|\,S}$ for the current seed set $S$. % maintained in an array
    %\item[$\boldsymbol{\lazygaini{i}{v}}$]: The stale score of vertex $v$ before selecting the $i$-th seed.
    \item[$\boldsymbol{\sketch{1..R}}$]: the sketches. Formally defined in \cref{sec:compact_sketching}.
    \item[$\boldsymbol{\rho}$ and $\boldsymbol{\rate}$]: $\rho=\rate n$ is the number of centers. 
\end{description}
 \vspace{-.1in}
\rule{\columnwidth}{.05em} 

% \vspace{-.06in}

\begin{description}[labelwidth=.15in,leftmargin=.15in]%[labelwidth=1in, leftmargin=.3in]
%\vspace{-.1in}\rule{\columnwidth}{.05em}
\item [\textbf{Function names:}]
\item [\sketchfunc$(G,r)$]: Compute the $r$-th sketch from graph $G$
\item [\marginal{}$(S,v,\sketch{1..R})$]: The marginal gain of vertex $v$ given the current seed set $S$ estimated from $R$ sketches $\sketch{1..R}$
\item [\textsc{NextSeed}$(S,\sketch{1..R})$]: Greedily determine the next seed based on sketches $\sketch{1..R}$ given the current seed set $S$.
\end{description}
 \vspace{-.1in}
\rule{\columnwidth}{.05em}
% \vspace{-.2in}

\hide{
\begin{itemize}[leftmargin=*]
    \item $G=(V,E)$: The input graph.
    \item $S\subseteq V$: seed set $S$.
    \item $\boldsymbol{\influence(S)}$ or ${{\influence_{G,M}(S)}}$: The influence spread (expected number of activated vertices) of seed set $S$ on graph $G$ and diffusion model $M$.
    \item $\gain{v\,|\,S}$: Marginal gain of $v$ on top of $S$. $\gain{v\,|\,S}=\influence(S\cup \{v\})-\influence(S)$.
    \item $\gainsimple{[}v{]}$: The lazy-evaluated marginal gain of $v$. It is an upper bound of the true marginal gain $\gain{v\,|\,S}$ for the current seed set $S$.
    \item $\sketch{r}$: the $r$-th sketch. Formally defined in \cref{sec:cc}.
\end{itemize}
}

% \vspace{-.08in}
\end{table}

Given a graph $G=(V,E)$, an influence diffusion model $M$ specifies how influence spreads from a set of current \defn{active} vertices to \defn{activate} more vertices in $V$.
Given a \defn{seed} set $S\subseteq V$, we use $\influence_{G,M}(S)$ to denote the expected number of vertices that $S$ can activate (including $S$) under diffusion model $M$ on graph $G$.
The IM problem is to find $S^*\subseteq V$ with size $\numseeds$, s.t. $S^*$ maximizes the influence spread function $\influence_{G,M}$.
\hide{
Given a graph $G=(V,E)$, the edge weight function $p:E\mapsto [0,1]$ is defined as the probability that vertex

%and an integer kk, \defn{Influence Maximization (IM)} problem is to find a subset S⊆VS\subseteq V with size kk, s.t. SS maximizes the influence spread function \influenceG,M\influence_{G,M}.
If the influence spread from vertex $i$ to $j$, we say $i$ \defn{activate} $j$.
$p_{i,j}$ is the probability that vertex $i$ activate $j$ through edge $(i,j)$.
The \defn{Influence Spread Function $\influence_{G,E}(S)$} of a vertex subset $S\subseteq V$, is the expected number of vertices that $S$ can activate (including $S$) under the propagation model $M$ on graph $G$.
}
With clear context, we omit $M$ and $G$, and use $\influence(\cdot)$.
Several propagation models have been proposed, including the Independent Cascade (IC) model~\cite{goldenberg2001talk}, Linear Threshold (LT) model~\cite{granovetter1978threshold,schelling2006micromotives}, and more~\cite{kempe2003maximizing,liu2012time, chen2012time, rodriguez2011uncovering}.
Since our sketch compression focuses on the IC model, we briefly introduce it here.
%It is popular in analysis social networks
%\myparagraph{Independent Cascade (IC) Model}.
In the IC model, influence spreads in rounds.
Initially, only the seed vertices are active.
%all vertices are inactive except for those in the seed set $S\subseteq V$.
%In round $i$, each vertex $u$ that was newly activated in round $i-1$ attempts to activate each inactive neighbor $v$ independently, and succeeds with probability $p_{u,v}$.
In round $i$, each vertex $u$ that was newly activated in round $i-1$ attempts to spread the influence via all incident edges $e$, and activate the other endpoint $v$ with probability $p_e$.
%If $G$ is undirected, $p_{u,v}=p_{v,u}$.
%In each round, a vertex vv is activated if at least one of its neighbors uu (incoming neighbors if on directed graphs) successfully activated it with probability p(u,v)p(u,v) given that uu is activated in the last round.
%The cascade probabilities are \textit{independent} (from each vertices and previous activations).

\hide{
\myparagraph{Submodularity and the Greedy Algorithm.} Kempe et al.~\cite{kempe2003maximizing} proved that IM under the IC model is NP-hard, and that the influence spread function is \defn{submodular}, i.e., for every $X,Y\subseteq V$ where $X\subseteq Y$, and $v\in V\setminus Y$, we have:
$$\influence(X\cup\{v\})-\influence(X)\ge \influence(Y\cup\{v\})-\influence(Y)$$
Due to the submodularity, the following greedy algorithm (later referred to as \simplegreedy{}) can achieve a $(1-1/e)$-approximation.
The algorithm starts with $S=\emptyset$ and repeatedly adds the vertex with the highest \defn{marginal gain} to $S$, until $|S|=\numseeds$.
The marginal gain $\truegain{v\,|\,S}$ of a vertex $v$ given the current seed set $S$ is defined as:
\begin{align}
\truegain{v\,|\,S}=\influence(S\cup \{v\})-\influence(S)
\end{align}
}

%\myparagraph{The Greedy Algorithm.} 
Kempe et al.~\cite{kempe2003maximizing} proved that IM under the IC model is NP-hard, and that the influence spread function has the following properties: for every $X,Y\subseteq V$ where $X\subseteq Y$, and $v\in V\setminus Y$, we have:%\vspace{.2em}
\begin{align}
\text{\bf Monotonicity: }&  \quad\quad\quad\,\,\influence(Y\cup\{v\})\ge \influence(Y) \\
\text{\bf Submodularity: }&\influence(X\cup\{v\})-\influence(X)\ge \influence(Y\cup\{v\})-\influence(Y)
\label{eq:submodular}
%\vspace{.5em}
\end{align}
These two properties allow the following \textbf{greedy algorithm} (later referred to as \simplegreedy{}) to give a $(1-1/e)$-approximation.
The algorithm starts with $S=\emptyset$ and repeatedly adds the vertex with the highest \defn{marginal gain} to $S$, until $|S|=\numseeds$.
The marginal gain $\truegain{v\,|\,S}$ of a vertex $v$ given the current seed set $S$ is defined as:
\begin{align}
\truegain{v\,|\,S}=\influence(S\cup \{v\})-\influence(S)
\end{align}

%With clear context, we omit SS and only use \gainv\gain{v}.

%The challenge in \textsc{GeneralGreedy} is that evaluating \influence(⋅)\influence(\cdot) requires running RR Monte Carlo (MC) simulations to compute the expected number of activated, which can be expensive.
With clear context, we omit $S$ and use $\truegain{v}$, and also call it the \defn{score} or the \defn{true score} of $v$.
We call the process to compute the true score of a vertex an \defn{evaluation}.
To estimate $\influence(\cdot)$, %obtain the expected number of activated vertices $\influence(\cdot)$, %\textsc{GeneralGreedy}
early solutions average $R'$ rounds of Monte Carlo (MC) experiments of influence diffusion simulation.
However, on real-world graphs, this approach requires a large value of $R'$ to converge, which can be expensive.

% \vspace{-\baselineskip}
% \vspace{-0.2em}
\myparagraph{Sketch-Based Algorithms.}
% Sketch-based algorithms are proposed to accelerate influence spread simulation.
Sketch-based algorithms are proposed to accelerate $\truegain{v}$ evaluations.
Instead of running independent MC experiments for each evaluation,
sketch-based algorithms statically sample $R$ graphs to reflect MC experiments 
and consistently simulate the results on them.
%Namely, the MC simulation uses the same $R$ sampled graphs to select all $k$ seeds.
%An experimental study on the NetHEPT graph (15K vertices and 59K edges) shows that the solution qualities on $R\approx 200$ static graphs and $10^4$ rounds of na\"ive MC simulation only differ by 0.5\%~\cite{cheng2013staticgreedy}.
%A study on the NetHEPT graph (15K vertices and 59K edges) shows that using $R\approx 200$ sampled graphs roughly match the quality of using $10^4$ MC experiments~\cite{cheng2013staticgreedy}.
%for \textsc{StaticGreedy} to achieve less than 0.005 relative difference between the influence spread and the ground truth (to the result of \textsc{GeneralGreedy}), while \textsc{GeneralGreedy} requires R>104R>10^4.
Using sketches allows for much faster convergence, making the number of needed simulations (i.e., sketches) $R$ smaller than that in MC experiments.
Hence, sketch-based algorithms are widely studied~\cite{chen2009efficient, cheng2013staticgreedy, ohsaka2014fast, borgs2014maximizing,cohen2014sketch,tang2014influence,tang2015influence}.
We summarize sketch-based algorithms in two steps (see \cref{alg:sketch_framework}):
\emph{sketch construction} and \emph{seed selection}.
Next, we introduce both steps with optimizations in previous work.
We summarize some related work in \cref{tab:related}, and review more in \cref{sec:related}.
Some notations are given in \cref{tab:notations}.

\newcommand{\tmc}{T}
\newcommand{\nc}{n_c}
\begin{table*}[t]
    \centering
    \small\vspace{-1em}
    
    \caption{\textbf{Existing approaches and our new one.}
    \mixedgreedy{} uses different approaches to select the first seed and the other seeds, so we list them separately.
    ``\#vertices per seed'': number of vertices to visit in all re-evaluations involved to find a seed. 
    For \staticgreedy{}, we assume the \emph{fusion} optimization in~\cite{gokturk2020boosting} to avoid explicitly storing sampled graphs. 
    $n$: number of vertices.     $\nc$: number of re-evaluations needed in CELF.
    $\tmc$: the average number of reachable vertices in a simulation (or a sketch).
    Empirically, $\tmc$ is large, and $\nc \ll n$. To achieve similar quality, $R\ll R'$.%\yan{add citations}
    %Algorithms that require CC information also need a preprocessing time, which is omitted in the table.
    }
    
    \begin{tabular}{@{}c@{  }@{  }c@{  }@{  }c@{  }c@{  }@{  }c@{  }@{  }c@{  }@{  }c@{}}
    & \bf Randomization  & \bf Compute Influence & \bf Select Seed &\bf \#vertics per seed & \bf Space & \bf Parallel\\
    \hline
         \bf \simplegreedy{}~\cite{kempe2003maximizing} & $R'$ Monte Carlo experiments  &  Simulation  &Evaluate all & $O(nR'\tmc)$ & $O(n)$ & no\\
         \bf \mixedgreedy{}~\cite{chen2009efficient} 1st Seed&  Fixed $R$ sampled graphs & Memoization&Evaluate all & $O(nR)$ & $O(n)$& no\\
         \bf \mixedgreedy{}~\cite{chen2009efficient} Others&  $R'$ Monte Carlo experiments & Simulation& CELF& $O(\nc R'\tmc)$ & $O(n)$& no\\
         \bf \staticgreedy{}~\cite{cheng2013staticgreedy} &  Fixed $R$ sampled graphs & Simulation& CELF& $O(\nc R\tmc)$ & $O(n)$& no\\
         \bf \infuser{}~\cite{gokturk2020boosting} & Fixed $R$ sampled graphs & Memoization& CELF&$O(\nc R)$ &  $O(nR)$& yes\\
         \bf \oursystem{} (this work) & Fixed $R$ sampled graphs & Simulation + Memoization &CELF& $O(\nc R\cdot\min(T,1/\alpha))$ & $O((1+\alpha R)n)$& yes\\
    \end{tabular}
    
    \label{tab:related}\vspace{-.5em}
\end{table*}

%\begin{enumerate}[label=Step \arabic*.,topsep=.3em]
%    \item Sketch construction (\cref{line:sketch_begin}--\cref{line:sketch_end}). This step constructs some sketches (see details below) to facilitate later evaluations of the influence spread function.
%    \item Seed selection (\cref{line:select_begin}--\cref{line:select_end}). This step selects $k$ seeds greedily by the marginal gain of the vertices. The marginal gain is computed purely from the sketches.
%    %A widely-adopted optimization is to use CELF~\cite{}.
%\end{enumerate}
%Next, we briefly introduce the high-level idea of both steps, and useful techniques in previous work.

\myparagraph{Sketch Construction.}
%The sketch-based algorithms first compute $R$ graphs generated from $G$ (usually sampled graphs based on the diffusion model), called \defn{sketches} or \defn{snapshots}, to accelerate influence evaluation.
In the earliest sketch-based algorithm \staticgreedy{}, $R$ sampled graphs~\cite{cheng2013staticgreedy} are explicitly stored as sketches.
In the IC model, the $r$-th sketch corresponds to a sampled graph $G'_r=(V,E_r')$, where $E_r'\subseteq E$, such that each edge $e\in E$ is sampled with probability $p_{e}$,
meaning a successful activation.
%On undirected graphs, the reachable vertex set from $v$ becomes the connected component that $v$ belongs.
%To evaluate a vertex, \staticgreedy{} simulate the influence spread on all sketches and count the number of activated vertices.
An evaluation will average the number of reachable vertices on all sampled graphs from the seed set $S$.
%As mentioned in \cref{sec:intro}, if $G$ is undirected,
%on a certain sketch, a vertex $v$ can influence all vertices in the same CC as $v$.
A later paper \infusersimple{}~\cite{gokturk2020boosting} proposed the \defn{fusion} optimization, which uses hash functions to avoid explicitly storing the sampled graphs $G'_r$.
They sample an edge $e$ in a sketch $G'_r$ with a random number generated from seed $\langle e,r\rangle$,
%The randomization for each sampled graph is fixed by the sketch id $r$, such that $G'_r$ can be fully reconstructed on the fly when needed.
%In $G'_r$, an edge $e$ is sampled with a random number generated from seed $\langle e,r\rangle$.
such that whether an edge is selected in a certain sketch is always deterministic, and a sampled graph $G'_r$ can be fully reconstructed from the sketch id $r$.
We also use this idea in our sketch compression algorithm.

\newcommand{\nosemic}{\renewcommand{\@endalgocfline}{\relax}}% Drop semi-colon ;
\newcommand{\dosemic}{\renewcommand{\@endalgocfline}{\algocf@endline}}% Reinstate semi-colon ;
\newcommand{\popline}{\Indm\dosemic}% Undent
\newcommand{\pushline}{\Indp}% Indent
\setlength{\algomargin}{.5em}
\SetSideCommentLeft

\begin{algorithm}[t]
\SetNoFillComment
\small
\caption{Sketch-based IM algorithm}
\label{alg:sketch_framework}
\SetKwInOut{Global}{Notations}
\Global{$G=(V,E)$: the input graph.
	$\numseeds$: the number of seed vertices.
	$R$: the number of sampled graphs}
\KwOut{
	$S$: a set of $K$ seeds that maximizes influence on $G$}
\SetKwInOut{Maintains}{Maintains}
\SetKwInput{Notes}{Notes}
\SetKwProg{myfunc}{Function}{}{}
\SetKwFor{parForEach}{ParallelForEach}{do}{endfor}
\Notes{%\\
 %\pushline %\pushline
\noindent$\sketch{1\dots R}$: $R$ sketches computed from $R$ sampled graphs \\
%\pushline\noindent$\gain{v}$: total marginal gain of vertex $v$ on $R$ sampled graphs\\
}
\DontPrintSemicolon
\tcp{\underline{Step 1: Sketch construction}}
\parForEach {$r \gets 1 \dots R$  \label{line:sketch_begin}}{
	%$G'=(V,E') \gets $ \textsc{Sample}($G,r$)\\
	$\sketch{r}\gets$ \sketchfunc{}$(r)$\tcp*[f]{Compute the $r$-th sketch} \label{line:sketch_end}\\
}
\vspace{.05in}
\tcp{\underline{Step 2: Seed selection using CELF}}
\hide{
\tcp{{\upshape\marginal{}}$(S,v,\sketch{1..R})$: the marginal gain of vertex $v$'s influence on $R$ sketches $\sketch{1..R}$ given seed set $S$}
\parForEach {$v \in V$}{
		% $GAIN(G',v)= get\_size(sketch_i,v)$\\
		$\gain{v} \gets \marginal(\emptyset, v, \sketch{1..R})$
}
\tcp{A max-priority-queue on all vertices. The key for vertex $v$ is $\gain{v}$}
$Q \gets \textsc{PriorityQueue}(V, \gain{\cdot})$\\
% $aug\_data = \{sketch, G\}$\\
% $Q.construct(\sigma, aug\_data)$\\
}
$S \gets \emptyset$\\
\hide{
\While{$|S|<\numseeds$}{
	$u \gets Q.\textsc{Pop}()$\tcp*[f]{{\upshape\textsc{Pop}}: find and remove the top}\\
        \tcp{Compute the latest marginal gain of $u$ based on $S$}
        $\gain{u}\gets$ \marginal$(S,u,\sketch{1..R})$\\
        \tcp{the new influence $\gain{u}$ is still larger than the best candidate in $Q$}
	\If(\tcp*[f]{\textsc{\upshape Top}: find the top}){$\gain{u}>\gain{Q.\textsc{Top}()}$}{
            \textsc{MarkAsSeed}($u$)\\
            $S \gets S \cup \{u\}$
        } \lElse{
        $Q.\textsc{Insert}(u)$ \tcp*[f]{insert $u$ back with new $\gain{u}$}
        }
% 	$\sigma_{u} = 0$\\
    %\tcp{Recompute the $\sigma$ of $Q.top$ and update $Q$ until the value of the top is up to date}
	%$Q.\textsc{UpdatePriority}(u, sketch_{1\dots R}, G)$  \Comment{recompute the marginal influence after selecting $u$ if necessary}\\
}
}
\While{$|S|<\numseeds$ \label{line:select_begin}}{
%\tcp{The next seed is $s^*\gets \arg\max_{v\in V\backslash S} \gain{v\,|\, S}$ with $\gain{\cdot}$ computed from the sketches $\sketch{1..R}$}
%\tcp{The next seed $s^* = \arg\max_{v\in V\backslash S} \gain{v\,|\, S}$}
%\tcp{The $\gain{\cdot}$ values are estimated by the sketches $\sketch{1..R}$}
$s^*\gets $\nextseed$(S,\sketch{1..R})$\tcp*[f]{Find \upshape $\arg\max_{v\in V}$\marginal$(S,v,\sketch{1..R})$}\\
\textsc{MarkSeed}($s^*,\sketch{1..R},S$)\tcp*[f]{Mark $s^{*}$ as a seed in the sketches} \label{line:next_seed}\\
$S \gets S \cup \{s^*\}$ \label{line:select_end}
}
\Return{$S$}

%\tcp{return the vertex with the highest marginal gain on the sketches}
%\myfunc{\upshape\textsc{NextSeed}$(S,\sketch{1..R})$} {
%  \Return {\(\arg\max_{v\in V} \marginal(S,v,\sketch{1..R})\)}
%}

\end{algorithm}

\begin{algorithm}[t]
\SetNoFillComment
\small
\caption{Sequential Seed Selection with CELF}
\label{alg:CELF_framework}
\hide{
\KwIn{$G=(V,E)$: the input graph\\
\pushline
	$K$: the number of seed vertices\\
	$R$: the number of sampled graphs}
\KwOut{
	$S$: a seed set that maximizes influence on $G$}
 }
\SetKwInOut{Maintains}{Maintains}
\SetKwInput{Notes}{Notes}
\SetKwProg{myfunc}{Function}{}{}
\SetKwFor{parForEach}{ParallelForEach}{do}{endfor}
\SetKwFor{Justrepeat}{Repeat}{}{}
\Notes{%\\
 %\pushline %\pushline
\noindent
%$\sketch{1\dots R}$: $R$ sketches\\
%\noindent$\gain{v}$: total marginal gain of vertex $v$ on $R$ sketches. \\
$Q$: max-priority-queue on all vertices $v\in V$ with key $\lazygain{v}$\\
\pushline
Initially $\lazygain{v}=\marginal(\emptyset,v,\sketch{1..R})$\\
}
\DontPrintSemicolon
\myfunc(\tcp*[f]{$S$: current seed set; $\sketch{1..R}$: $R$ sketches}){\upshape{\textsc{NextSeed}}$(S,\sketch{1..R})$}{
%$s^* \gets Q.\textsc{Pop}()$\tcp*[f]{{\upshape\textsc{Pop}}: find and remove the top}\\
%Build a priority queue $Q$ on all vertices. The initial key is $\marginal(\emptyset,u,\sketch{1..R})$\\
\Justrepeat{}{
        $s^* \gets Q.\textsc{Pop}()$\tcp*[f]{{\upshape\textsc{Pop}}: find and remove the top} \label{line:Q_start}\\
        %\tcp{Compute the latest marginal gain of $s^*$ based on $S$}
        $\lazygain{s^*}\gets$ \marginal$(S,s^*,\sketch{1..R})$ \label{line:marginal}\\
        %\tcp{the new influence $\gain{u}$ is still larger than the best candidate in $Q$}
        %\tcp{\textsc{\upshape Max(\,)}: the max $\Delta$ value in $Q$}
	\lIf{$\lazygain{s^*} > \lazygain{          Q.\textsc{Top}()$} }{
            \Return {$s^*$}
        } \lElse{
        $Q.\textsc{Insert}(s^*)$ \tcp*[f]{insert $s^*$ back with new score} \label{line:Q_end}
        }
        %$u \gets Q.\textsc{Pop}()$\tcp*[f]{{\upshape\textsc{Pop}}: find and remove the top}\\
}
}

\end{algorithm}

\hide{

\myfunc{\sc{Sample}($G,r$)}{
$E' \gets \emptyset$\\
\For {$(u,v) \in E$}{
        \tcp{{Generate random number $p\in[0,1]$ from seed $u,v,r$}}
        $p\gets $ random$(u,v,r)$ \\
% 	Randomly choose $p \in_{u,v,r} [0,1]$ from a uniform dist\\
	\If {$p \leq w_{u,v}$}{
		$E' \gets E' \cup \{(u,v)\}$
	}
}
Construct $G' = (V,E')$\\
\Return $G'$
}
} 

Many existing algorithms also use \defn{memoization} to avoid influence spread simulation on sketches.
%A later approach \infuser{} accelerates \staticgreedy{} using the idea first proposed by the \mixedgreedy{}~\cite{} paper -
%On undirected graphs, this means memoizing the connected component (CC) size for each vertex, since
On undirected graphs and the IC model, the \mixedgreedy{} paper~\cite{chen2009efficient} first observed that a vertex $v$'s influence on a sketch is all vertices in the same connected component (CC) as $v$, but they only used this idea to select the first seed.
%Hence, later work on sketch construction boiled down to efficiently processing and storing connectivity information.
%that a evaluation can be done in $O(1)$ cost per sketch by reading the CC information.
%For instance, \infuser{}~\cite{gokturk2020boosting} stores the per-vertex connectivity information for each sampled graph as part of the sketches.
Later, \infuser{} adopts this idea to select all seeds and memoizes the CC information of each sampled graph as the sketch.
A vertex $v$'s score is then the average of the (inactivated) CC sizes on all $R$ sketches, which
%costs $O(R)$.
can be obtained in $O(R)$ cost.
This approach avoids simulation but leads to $O(Rn)$ space that is expensive for large graphs.
%On this sketch, a vertex $v$ can activate all vertices reachable from it on $G'_r$, which are all the vertices in the same CC as $v$ on $G'_r$ if the input graph is undirected~\cite{xx}\yan{add}.
%Based on the same observation, later work further stores \emph{additional information} as part of the sketches to improve performance.
%For instance, PMC~\cite{ohsaka2014fast} stores the strongly connected components on each graph, and \infuser{}~\cite{gokturk2020boosting} stores connected components.
%For example, \infuser{} stores the connected component sizes for each sampled graph in the sketch.
%This requires per-vertex information on each sketch, which leads to $O(Rn)$ space.
\cref{sec:compact_sketching} presents how our sketch compression approach reduces this high space usage.
%Since this paper focuses on undirected graphs, here we review how \infuser{} uses connectivity information to avoid MC simulation.

\hide{
The simplest sketches just stored the $R$ sampled graphs, as in the first sketch-based algorithm \textsc{StaticGreedy}~\cite{cheng2013staticgreedy}.
Graphs are sampled based on the diffusion model, and the MC simulation on $R$ sampled graphs gives the marginal influence.
%Later work further avoids the MC simulation on the sampled graphs by storing \emph{additional information} as part of the sketches.
Later work further stores \emph{additional information} as part of the sketches to improve performance.
For instance, PMC~\cite{ohsaka2014fast} stores the strongly connected components on each graph, and Infuser~\cite{gokturk2020boosting} stores connected components.
Since this paper focuses on undirected graphs, here we review how Infuser uses connectivity information to avoid MC simulation.

\yan{Review Infuser's technique here.}
\letong{
The \textsc{INFuseR-MG}~\cite{gokturk2020boosting} computes and stores all the connected component labels for each vertex on all sketches. The first seed vertex is indeed the one having the largest average component size on all sketches. Instead of resampling as in \textsc{NewGreedy} and \textsc{MixGreedy}, \textsc{INFuseR-MG} utilizes this information during the CELF selection state for computing marginal gains and finding the remaining $K-1$ seed vertices.
The marginal gain for a vertex $u$ can be found by computing the average size of connected components (over all the $R$ samples) that contain $u$ but do not contain any seed vertices.
}

The vertices activated by seed set $S$ on a sketch are all vertices reachable from $S$.
Namely, a vertex $v$ is activated if and only if there exists a path from $u$ to $v$ where $u\in S$.
The influence $\influence(S)$ for $S$ can be computed by averaging the number of activated vertices from the sampled graphs.
Existing experimental results show that, using pre-generated (static) snapshots for MC simulations, the influence spread function converges much faster with the number of simulations (or sketches) $R$.

Also, Infuser avoids storing the sampled graph by using hash functions.
\yan{Review the solution here.}
\letong{When an edge $(u,v)$ with a certain orientation is read from the memory, \textsc{INFuseR-MG} decides to skip it or not in the $r$-th sampling by comparing the propagation probability $p_{u,v}$ and the outcome of a direction-oblivious hash function that takes edge vertices and sketch number as the input $\rho(u,v)_r$.
}
}

\hide{
\begin{table}
    \centering
    \small
    \begin{tabular}{@{}l@{ }lll@{ }l@{}}
    & Randomization & Get Influence& Time & Space\\
    \hline
         SimpleGreedy & Independent & MC simulation & $O(R'T_{MC})$ & $O(n)$\\
         MixGreedy$^1$&  Fixed & CC info& $O(R)$ & $O(n)$\\
         MixGreedy$^*$&  Independent & MC simulation& $O(R'T_{MC})$ & $O(n)$\\
         StaticGreedy &  Fixed & MC simulation& $O(RT_{MC})$ & $O(n)$\\
         Infuser & Fixed & CC info& $O(R)$ & $O(Rn)$\\
         \oursystem{} & Fixed & partial MC + & $O(R/\alpha)$ & $O(\alpha Rn)$\\
         &&partial CC info
    \end{tabular}
    \caption{Summarizing methodology in existing work and our system.
    ``randomization'' the use of randomization to decide the activation of edges.
    ``Independent'': independent simulation every evaluation. ``Fixed'': only simulate on $R$ fixed sampled graphs.
    ``Get Influence'': the way to compute influence with given randomization.
    ``MC simulation'': simulate the spread process.
    ``CC info'': pre-compute the CC sizes as the influence.
    ``Time'': time of one influence estimation. $T_{MC}$: the time for one MC simulation.
    }
    \label{tab:related}
\end{table}
}

%\bigskip
\myparagraph{Seed Selection with CELF.}
%Another performance bottleneck of the greedy algorithm is to evaluate all $n$ vertices for every seed selection round.
A useful optimization for the greedy algorithm is CELF~\cite{leskovec2007cost}, which avoids evaluating all vertices in $\nextseed$.
CELF uses \emph{lazy} evaluation for {submodular} functions and
% \hide{(defined in
% % \eqref{eqn:properties}
% \eqref{eq:submodular}
% ).}
%i.e., for every $X,Y\subseteq V$ where $X\subseteq Y$, and $v\in V\setminus Y$, we have:
%$$\influence(X\cup\{v\})-\influence(X)\ge \influence(Y\cup\{v\})-\influence(Y)$$
evaluates a vertex only if it becomes a ``promising'' candidate for the next seed.
We show the CELF seed selection algorithm in \cref{alg:CELF_framework}.
CELF uses a priority queue $Q$ to maintain all vertices with their scores as the key.
Due to submodularity, $\gain{v}$ is non-increasing with the expansion of the seed set $S$. With lazy evaluation, the scores in $Q$ may be stale but are always upper bounds of the true scores.
%To distinguish it from the true score, 
We call this lazily-evaluated score the \defn{stale score} of $v$ and denote it as $\lazygain{v}$ stored in an array.
To select the next seed, CELF keeps popping the top element $v$ from $Q$, re-evaluating its true score $\truegain{v}$, storing $\truegain{v}$ to $\lazygain{v}$, and inserting it back unless $\truegain{v}$ is greater than the current largest value in $Q$ (\cref{line:Q_start}-\ref{line:Q_end}).
In this case, we can set $v$ as the next seed without more evaluations, as the true scores of other vertices can not exceed their values in $Q$.
%The last popped vertex $u$ must have its latest $\sigma_u$ value larger than
%Existing results showed that CELF can often greatly reduce the number of influence evaluations in certain rounds~\cite{}.
\hide{
The algorithm uses a priority queue $Q$ to maintain all vertices with their marginal gain as the key. Initially, each $\gain{v}$ is the influence of $v$ as a single seed.
%To select the next seed, CELF keeps popping the top element $u$ from $Q$, re-evaluating its marginal gain $\sigma_u$, and inserting it back until the new marginal influence of the popped vertex is greater than all other vertices in $Q$. %The last popped vertex is the new seed to add, and we can avoid evaluating further vertices in $Q$.
%Usually, this only takes a few rounds to select a new seed, which is much smaller compared to the total number of vertices, but it is inherently sequential.
To select the next seed (the while-loop in line xx), CELF pops the top element $u$ from $Q$. Note that $\gain{u}$ may be stale due to lazy evaluation, and thus we re-evaluate $\gain{u}$ using the sketches (Line xx).
Due to submodularity, $\gain{u}$ is non-increasing with the expansion of the seed set $S$.
As a result, $\gain{u}$ may get lower after re-evaluation.
If $\gain{u}$ is lower than the current top element in $Q$, then $\gain{u}$ may not be the largest marginal gain, and we insert $u$ back to $Q$ and repeat the process (Line xx).
Otherwise, we can directly mark $u$ as a seed and finish this round (Line xx). We do not need to evaluate further vertices in $Q$ since their true marginal gain can only be lower than the value stored in $Q$.
Existing results showed that CELF greatly reduces the number of influence evaluations.
}
%Due to submodularity, a stale $\sigma_v$ value in $Q$ is an upper bound of the true marginal gain for $v$.
% In many cases, CELF can reduce the number of re-evaluations, but it is essentially sequential and evaluates all vertices one by one.
Although CELF can reduce the number of evaluations, it is essentially sequential and evaluates vertices one by one.
In \cref{sec:selecting}, we present our new data structures that allow for parallel evaluations in CELF.

\hide{

The sketch-based algorithms pre-compute $R$ sketches, each storing some information of an MC simulation.
The sketch-based algorithms are derivatives of the \textsc{StaticGreedy} algorithm\yan{cite}, which samples graphs based on the diffusion model and computes marginal influence on them.
The sampled graphs are also called snapshots.
In the IC model, each snapshot $G'=(V,E')$ is generated by selecting each edge $(u,v)\in E$ with probability $p_{u,v}$.
In addition to the sampled graph, the sketch-based  algorithms~\cite{ohsaka2014fast,gokturk2020boosting} choose to store some other useful information (like strongly connected components in PMC~\cite{ohsaka2014fast}, and connected components in Infuser~\cite{gokturk2020boosting}) in the sketch to accelerate influence evaluation.
}

\section{Space-Efficient Sketches}\label{sec:compact_sketching}

\begin{figure*}
    \centering % \vspace{-.75em}
    \includegraphics[width=0.95\textwidth]{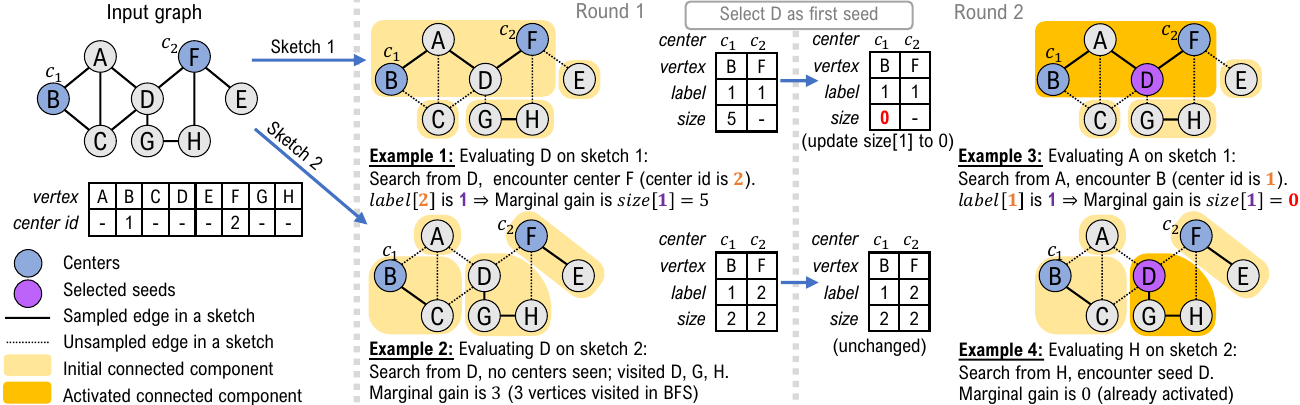}\vspace{.5em}
    \caption{
        \textbf{An example of our sketch compression} on a graph with 8 vertices (B and F as centers) and $R=2$ sampled graphs. %Vertices B and F are centers, whose center id is 1 and 2. %\vspace{-.75em}%In this case, vertices 1 and 6 are selected as centers.  Vertex 4 is selected as the first seed as it has the largest influence.
        \label{fig:sketch}
    }
    % \vspace{-.1in}
\end{figure*}

\begin{algorithm}[t]
\small
\caption{Our sketch algorithm with compression \label{alg:oursketch} }
\SetKwInput{Global}{Global Variables}
\SetKwInput{Notes}{Notes}
\Global{$G=(V,E)$: input graph; $R$: number of sketches\\
\DontPrintSemicolon
\pushline
$C=\{c_1,c_2,\dots,c_{\numcenters}\}\subseteq V$: randomly selected \emph{centers}. $\numcenters=|C|$
}
\Notes{A sketch $\sketch{r}$ is a triple $\langle r,\ccsize[1..\numcenters],\lbl[1..\numcenters]\rangle$, defined at the beginning of \cref{sec:compact_sketching}.\\
\hide{
\pushline
        $r$: the id of the sketch. \\
	$\ccsize[1..\numcenters]$: $\ccsize[i]$ is the connect component size of center $i$. \\
        $\lbl[1..\numcenters]$: $\lbl[i]$ is the connect component id of center $i$. }
}
\SetKwProg{myfunc}{Function}{}{}
\SetKwFor{parForEach}{ParallelForEach}{do}{endfor}
\vspace{.05in}
% Parallel sample each vertex as a center with probability $b$\\

\myfunc(\tcp*[f]{$r$: sketch id }){\upshape \textsc{Sketch}($r$)
	} {
\hide{
\tcp{$L[v]$ maintains a label for each vertex to track its connectivity information with the centers}
  \tcp{the initial label $L[v]$ is the id of any center connected with $v$}
\lparForEach{$v \in V$}{
  $ L[v]\gets \text{\getcenter{}}(r,v)$
}
% $L\gets UnionFind(E', L)$  \Comment{for each $(u,v)\in E'$, call union $(L(u),L(v))$ if $L[u]\neq L[v]$}\\
\tcp{parallel Union-Find}
\parForEach{$(u,v)\in E$ \label{line:union}} {
    \If {\textsc{Sample}$(u,v,r)$ and $L[u]> 0$ and $L[v]> 0$} {
    \lIf {$\textsc{Find}(L[u]) \neq \textsc{Find}(L[v])$}{
        %$L[u] \gets \textsc{Find}(L[v])$ \tcp*[f]{Union sets $L[u]$ and $L[v]$}\\
        \textsc{Union}$(L[u],L[v])$
    }}
}
\lparForEach{$v\in V$ s.t. $L[v]> 0$} {
$L[v]\gets$\textsc{Find}($L[v]$) \label{line:find}
}
%\tcp{$h[i]$: \#occurrences of $i$ in array $L[\cdot]$}
$h \gets$  \textsc{Histogram}$(L)$\tcp*[f]{$h[i]$: \#occurrences of $i$ in array $L[\cdot]$}\\
}
Compute the connected components of graph $G'=(V,E')$, where $E=\{(u,v)~|~(u,v) \in E, \textsc{Sample}(u,v,r)=\true\}$\\
\parForEach{$c_i \in C$}{
$\lbl[i] \gets \min_j\{c_j \text{~is in the same CC as~} c_i\}$\\
	%\lIf{$\lbl[i]=i$}{$\ccsize[i] \gets$ the CC size of center $c_i$}
        %\tcp*[f]{if $c_i$'s connected component is labeled by $c_j$, set $\lbl[i]=j$}
        %$\lbl[i] \gets L[c_i]$

}
\DontPrintSemicolon
\parForEach{$c_i \in C$}{
\lIf{$\lbl[i]=i$}{$\ccsize[i] \gets$ the CC size of center $c_i$}}
\Return $\langle r, \lbl[1..\numcenters],\ccsize[1..\numcenters]\rangle$
}

\codeskip
\tcp{sample an edge $e$ with probability $p_{e}$ for sketch $r$}
\myfunc(\tcp*[f]{$e$: edge identifier; $r$: sketch identifier}){\sc{Sample}($e,r$)}{
% 	Randomly choose $p \in_{u,v,r} [0,1]$ from a uniform dist\\
$p\gets\mathit{random}(e,r)$\tcp*[f]{Generate $p\in[0,1]$ from random seed $e,r$}\\
\Return {$p \leq p_{e}$}
}

\codeskip

\DontPrintSemicolon

%\tcp{For sketch $r$, compute the \textbf{center id} for vertex $v$. If no such center found, return a negative value $-1\times$(\#vertices connected with $v$)}
\tcp{$\delta$: marginal influence of $v$ on sketch $\sketch{r}$. $l$: label of $v$'s CC in sketch $\sketch{r}$ if $v$ is connected to any center; otherwise $l=-1$.
$S$: the current seed set. }
\myfunc{\upshape $\langle \delta, l \rangle=$ \getcenter{}$(\sketch{r},v,S)$}{
Start BFS from $v$ on a sampled subgraph of $G$, where an edge $e\in E$ exists if $\textsc{Sample}(e,r)=\true$. Count the \#reached-vertices as $n'$\\
%\lIf{BFS visits $v'\in S$}{\Return{$\langle 0,-1 \rangle$}}
\If{a center $c_i\in C$ is encountered during BFS\label{line:found_center}}{
$l\gets\sketch{r}.\lbl[i]$\tcp*[f]{Find the label of the center}\\
\Return{$\langle \sketch{r}.\ccsize[l], l \rangle $} \tcp*[f]{the CC size and label of the center}
}
\lIf{any $v'\in S$ has been visited by BFS}{\Return{$\langle 0,-1 \rangle$}\label{line:nocenter}}
\lElse{\Return{$\langle n', -1\rangle $} \tcp*[f]{influence is \#visited vertices in BFS}}
}
\codeskip

\tcp{Marginal gain of $v$ given seed set $S$ on sketches $\sketch{1..R}$}
\myfunc{\sc{\marginal{}}($S,v,\sketch{1..R}$)}{
    %$t\gets 0$\\
    \parForEach{$r\gets 1..R$}{
        $\langle \delta_r, \cdot \rangle \gets \getcenter{}(\sketch{r},v,S)$
    }
    %$t\gets \sum_{r=1}^{R}\delta_r$\\
    \Return{$(\sum_{r=1}^{R}\delta_r)/R$} \tcp*[f]{the sum can be computed in parallel}
}

\codeskip
\myfunc{\sc{MarkSeed}$(v,\sketch{1..r},S)$} {
    \parForEach(\tcp*[f]{For each sketch $\sketch{r}$}){$r\gets 1..R$}{
        $\langle \delta, l \rangle\gets \getcenter(\sketch{r},v,S)$\\
        \lIf {$l>0$} {
        %$l\gets\sketch{r}.\lbl[i]$\\
        $\sketch{r}.\ccsize[l]\gets 0$ \tcp*[f]{Clear the corresponding CC size}
        }
    }
}

\end{algorithm}

\hide{
\begin{algorithm}[ht]
\small
\caption{Our sketch algorithm with compression \label{alg:oursketch} }
\SetKwInput{Global}{Global Variables}
\SetKwInput{Notes}{Notes}
\Global{$G=(V,E)$: the input graph\\
\pushline
$C=\{c_1,c_2,\dots,c_{\numcenters}\}\subseteq V$: randomly selected \emph{centers}. $\numcenters=|C|$\\
$R$: the number of sketches
}
\Notes{A sketch $\sketch{r}$ is a triple $\langle r,\ccsize[1..\numcenters],\lbl[1..\numcenters]\rangle$:\\
\pushline
        $r$: the id of the sketch. \\
	$\ccsize[1..\numcenters]$: $\ccsize[i]$ is the connect component size of center $i$. \\
        $\lbl[1..\numcenters]$: $\lbl[i]$ is the connect component id of center $i$. }
\SetKwProg{myfunc}{Function}{}{}
\SetKwFor{parForEach}{ParallelForEach}{do}{endfor}
\vspace{.05in}
% Parallel sample each vertex as a center with probability $b$\\

\myfunc(\tcp*[f]{$r$: sketch id }){\upshape \textsc{Sketch}($G,r$)
	} {
\hide{
\tcp{$L[v]$ maintains a label for each vertex to track its connectivity information with the centers}
  \tcp{the initial label $L[v]$ is the id of any center connected with $v$}
\lparForEach{$v \in V$}{
  $ L[v]\gets \text{\getcenter{}}(r,v)$
}
% $L\gets UnionFind(E', L)$  \Comment{for each $(u,v)\in E'$, call union $(L(u),L(v))$ if $L[u]\neq L[v]$}\\
\tcp{parallel Union-Find}
\parForEach{$(u,v)\in E$ \label{line:union}} {
    \If {\textsc{Sample}$(u,v,r)$ and $L[u]> 0$ and $L[v]> 0$} {
    \lIf {$\textsc{Find}(L[u]) \neq \textsc{Find}(L[v])$}{
        %$L[u] \gets \textsc{Find}(L[v])$ \tcp*[f]{Union sets $L[u]$ and $L[v]$}\\
        \textsc{Union}$(L[u],L[v])$
    }}
}
\lparForEach{$v\in V$ s.t. $L[v]> 0$} {
$L[v]\gets$\textsc{Find}($L[v]$) \label{line:find}
}
%\tcp{$h[i]$: \#occurrences of $i$ in array $L[\cdot]$}
$h \gets$  \textsc{Histogram}$(L)$\tcp*[f]{$h[i]$: \#occurrences of $i$ in array $L[\cdot]$}\\
}
Compute the connected components of graph $G'=(V,E')$, where $E=\{(u,v)~|~(u,v) \in E, \textsc{Sample}(u,v)=\true\}$\\
\parForEach{$c_i \in C$}{
	$\ccsize[i] \gets$ the connected component size of center $c_i$\\
        %\tcp*[f]{if $c_i$'s connected component is labeled by $c_j$, set $\lbl[i]=j$}
        %$\lbl[i] \gets L[c_i]$
        $\lbl[i] \gets \min_j\{c_j \text{~is in the same CC as~} c_i\}$
}
\Return $\langle r,\ccsize[1..\numcenters], \lbl[1..\numcenters]\rangle$
}

\codeskip

\tcp{sample an edge $(u,v)$ with probability $w_{u,v}$ for sketch $r$}
\myfunc{\sc{Sample}($u,v,r$)}{
% 	Randomly choose $p \in_{u,v,r} [0,1]$ from a uniform dist\\
$p\gets\mathit{random}(u,v,r)$\tcp*[f]{Generate $p\in[0,1]$ from seed $u,v,r$}\\
\Return {$(p \leq p_{u,v})$}
}

\codeskip

\tcp{For sketch $r$, compute the \textbf{center id} for vertex $v$. If no such center found, return a negative value $-1\times$(\#vertices connected with $v$)}
\myfunc{\upshape \getcenter{}$(r,v)$}{
Start BFS from $v$ in $G$, but skip edge $(u,v)$ if \textsc{Sample}$(u,v,r)$ is \false{}. Count \#visited vertices in BFS as $n'$\\
If any seed is encountered during BFS, \Return{0}\\
Stop BFS when a center $c_i\in C$ is encountered and \Return{$i$} (the center id). \label{line:found_center} \\
\Return{$-n'$} \tcp*[f]{return $-n'$ if BFS stops without finding a center}
\hide{
Start BFS from $v$ in $G'$\\
\lIf{any center $c$ encountered during BFS} {\Return{$c$}}
\lElse{
%(\tcp*[f]{No center reachable in the entire BFS})
\Return{$-n'$}, where $n'$ is \#vertices visited by BFS
}
}
}
\codeskip

\tcp{Marginal gain of $v$ given seed set $S$ on sketches $\sketch{1..R}$}
\myfunc{\sc{\marginal{}}($S,v,\sketch{1..R}$)}{
    $t\gets 0$\\
    \parForEach{$r\gets 1..R$}{
        $x \gets \getcenter{}(r,v)$\\
        \If(\tcp*[f]{Center encountered}) {$x>0$} {
        $l\gets\sketch{r}.\lbl[x]$\\
        $t\gets t+ \sketch{r}.\ccsize[l]$\tcp*[f]{Get CC size by the center id}
        }
        %\lElseIf{any $u\in S$ encountered in \getcenter}{
        %$x\gets 0$
        %}
        %\lElse {$x\gets -c$}
        %\ElseIf{no seed $s\in S$ encountered in \getcenter}
        \lElse
        {
        $t\gets t+(-x)$ \tcp*[f]{CC size = \#visited vertices}
        }
    }
    \Return{$t/R$}
}

\codeskip
\myfunc{\sc{MarkAsSeed}$(v)$} {
    \parForEach(\tcp*[f]{For each sketch $\sketch{r}$}){$r\gets 1..R$}{
        $x\gets \getcenter(r,v)$\\
        \If {$x>0$} {
        $l\gets\sketch{r}.\lbl[x]$\\
        $\sketch{r}.\ccsize[l]\gets 0$ \tcp*[f]{Set the corresponding CC size as 0}
        }
    }
}

\end{algorithm}
}

\hide{
\SetSideCommentLeft
\begin{algorithm}[t]
\SetNoFillComment
\small
\caption{Our sketch algorithm with compaction}
\KwIn{$G'=(V,E')$: the sampled influence graph\\}
\KwOut{ ~\\
\pushline
A sketch containing the following information $\langle G',C,s[1..|C|]\rangle$:\\
        $G'$: the sampled graph\\
        $C$: a set of centers\\
	$s$: $s_i$ is the connect component size of center $i$}
\SetKwProg{myfunc}{Function}{}{}
\SetKwFor{parForEach}{ParallelForEach}{}{endfor}
\vspace{.05in}
% Parallel sample each vertex as a center with probability $b$\\
\tcp{$C\in V$ is a set of \emph{centers} from the input graph }
\myfunc{Sketch$(G',C)$} {
  return $(G', \mathit{CC-comact()})$
}
\vspace{.05in}
\tcp{$C$: selected centers}
\myfunc{CC\_Compact($G',C$)
	} {
\parForEach{$v \in V$}{
	compute the nearest center to $v$: $ L[v]\gets FindCenter(G',C,v)$
}
% $L\gets UnionFind(E', L)$  \Comment{for each $(u,v)\in E'$, call union $(L(u),L(v))$ if $L[u]\neq L[v]$}\\
\parForEach{$(u,v)\in E'$ \Comment{parallel Union-Find}} {
    \If {$\textsc{Find}(L[u]) \neq \textsc{Find}(L[v])$}{
        $L[u] \gets \textsc{Find}(L[v])$ \Comment{Union sets $L[u]$ and $L[v]$}\\
    }
}
$hist \gets$  \textsc{Histogram}$(L)$\\
\parForEach{$c_i \in C$}{
	$s_i = hist[L[c_i]]$
}
\Return $s$
}

\myfunc{\sc{GetInfluence}($S, \sketch{r},u,C$)}{
    \If{$u\in S$}{
        \Return 0
    }
    \If {$u \in C$}{
        $i \gets c_i = u $ and $c_i \in C$\\
        \Return $\sketch_r[i]$
    }\Else{
        $c_i \gets \textsc{FindCenter}(G',C,u)$\\
        \If {$c_i$ exists}{
            \Return $\sketch_r[i]$
        }\Else{
            \If {meet seeds during \textsc{FindCenter}}{
                \Return 0
            }\Else{
                \Return |vertices visited during \textsc{FindCenter}|
            }
        }
    }
}

\end{algorithm}
}

% Version with C as parameter
\hide{
\SetSideCommentLeft
\begin{algorithm}[t]
\SetNoFillComment
\small
\caption{Our sketch algorithm with compaction}
\SetKwInOut{Global}{Global Variable}
\Global{$G=(V,E)$: the input graph}
\KwIn{\\
\pushline
$r$: the id of the sketch
}
\KwOut{ ~\\
\pushline
A sketch $\sketch[r]$, which is a triple $\langle r,C,\ccsize[1..\numcenters]\rangle$:\\
        $r$: the label of the sketch. It can be used to reconstruct the sampled graph for this sketch\\
        $C$: a set of centers. $|C|=\numcenters$\\
	$\ccsize[1..\numcenters]$: $\ccsize[i]$ is the connect component size of center $i$. \\
        $\lbl[1..\numcenters]$: $\lbl[i]$ is the connect component size of center $i$. }
\SetKwProg{myfunc}{Function}{}{}
\SetKwFor{parForEach}{ParallelForEach}{do}{endfor}
\vspace{.05in}
% Parallel sample each vertex as a center with probability $b$\\
\myfunc{\upshape \textsc{Sketch}$(G,r)$} {
  $C\gets$ select $\numcenters$ random \textit{centers} from $V$\\
  %$G'\gets$ \textsc{Sample}$(G,r)$\\
  \Return $\langle r,C, \text{\textsc{\compactcc{}}}(r,C)\rangle$
}

\codeskip
\tcp{$r$: sketch id. $C$: selected centers.}
\myfunc{\upshape \textsc{\compactcc{}}($r,C$)
	} {
\tcp{$L[v]$ maintains a label for each vertex to track its connectivity information with the centers}
\parForEach{$v \in V$}{
  \tcp{the initial label $L[v]$ is the id of any center connected with $v$}
  $ L[v]\gets \text{\getcenter{}}(r,C,v)$
}
% $L\gets UnionFind(E', L)$  \Comment{for each $(u,v)\in E'$, call union $(L(u),L(v))$ if $L[u]\neq L[v]$}\\
\tcp{parallel Union-Find}
\parForEach{$(u,v)\in E$} {
    \If {\textsc{Sample}$(u,v,r)$ and $L[u]> 0$ and $L[v]> 0$} {
    \lIf {$\textsc{Find}(L[u]) \neq \textsc{Find}(L[v])$}{
        %$L[u] \gets \textsc{Find}(L[v])$ \tcp*[f]{Union sets $L[u]$ and $L[v]$}\\
        \textsc{Union}$(L[u],L[v])$
    }}
}
\lparForEach{$v\in V$ s.t. $L[v]> 0$} {
$L[v]\gets$\textsc{Find}($L[v]$)
}
\tcp{$h[i]$: \#occurrences of $i$ in array $L[\cdot]$}
$h \gets$  \textsc{Histogram}$(L)$\\
\parForEach{$c_i \in C$}{
	$\ccsize[i] \gets h[c_i]$\\
        %\tcp*[f]{if $c_i$'s connected component is labeled by $c_j$, set $\lbl[i]=j$}
        $\lbl[i] \gets L[c_i]$
}
\Return $\langle \ccsize[\cdot], \lbl[\cdot]\rangle$
}

\codeskip

\tcp{For sketch $r$, compute the \textbf{center id} for vertex $v$. If no such center found, return a negative value $-1\times$(\#vertices connected with $v$)}
\myfunc{\upshape \getcenter{}$(r,C,v)$}{
Start BFS from $v$ in $G$, but skip edge $(u,v)$ if \textsc{Sample}$(u,v,r)$ is not true. Count \#visited vertices by BFS as $n'$\\
Stop BFS when a center $c_i\in C$ is encountered and \Return{$i$} (the center id). \\
\Return{$-n'$} \tcp*[f] {\tcp*[f] return $-n'$ if BFS stops without finding a center}
\hide{
Start BFS from $v$ in $G'$\\
\lIf{any center $c$ encountered during BFS} {\Return{$c$}}
\lElse{
%(\tcp*[f]{No center reachable in the entire BFS})
\Return{$-n'$}, where $n'$ is \#vertices visited by BFS
}
}
}
\codeskip

\tcp{sample an edge $(u,v)$ with probability $w_{u,v}$ for sketch $r$}
\myfunc{\sc{Sample}($u,v,r$)}{
% 	Randomly choose $p \in_{u,v,r} [0,1]$ from a uniform dist\\
$p\gets\mathit{random}(u,v,r)$\tcp*[f]{Generate $p\in[0,1]$ from seed $u,v,r$}\\
\Return {$(p \leq w_{u,v})$}
}

\codeskip

\myfunc{\sc{\marginal{}}($S,v,\sketch[1..R]$)}{
    $t\gets 0$\\
    \For{$r\gets 1..R$}{
        $c \gets \getcenter{}(r,C,v)$\\
        \If {$c>0$} {
        $l\gets\sketch[r].\lbl[c]$\\
        $t\gets t+ \sketch[r].\ccsize[l]$}
        %\lElseIf{any $u\in S$ encountered in \getcenter}{
        %$x\gets 0$
        %}
        %\lElse {$x\gets -c$}
        \ElseIf{no seed $u\in S$ encountered in \getcenter}
        {
        $t\gets t+(-c)$
        }
        }
}

\codeskip
\myfunc{\sc{MarkAsSeed}$(v)$} {
    \For(\tcp*[f]{For each sketch $\sketch[r]$}){$r\gets 1..R$}{
    $c\gets \getcenter(r,\sketch[r].C,v)$\\
    \lIf {$c>0$} {
    $l\gets\sketch[r].\lbl[c]$
    $\sketch[r].\ccsize[l]\gets 0$ \tcp*[f]{set the corresponding CC size as 0}
    }
    }
}

\end{algorithm}

}

This section presents our new technique to construct compressed sketches in parallel for the IC model on undirected graphs.
In this setting, as discussed in \cref{sec:prelim}, a vertex $v$ can activate all vertices in the same CC on a certain sketch.
Thus, \emph{memoizing} per-vertex CC information in sketches~\cite{gokturk2020boosting} can accelerate the influence evaluation but requires $O(Rn)$ space, which does not scale to large input graphs.
Alternatively, one can avoid memoization and run a \emph{simulation} by traversing the sampled graph to find the CC when needed.
This requires no auxiliary space but can take significant time.

Our approach combines the benefit of both, using bounded-size auxiliary space while allowing for efficiency.
Our key idea is to only \emph{partially memoize} CC information in each sketch,
%and retrieve the CC information \emph{on-the-fly} using a partial simulation.
and retrieve this information by \emph{partial simulations}.
%Our algorithm has a parameter ϵ\epsilon between 0 and 1, and our algorithm only stores the connectivity information for \numcenters=ϵn\numcenters=\epsilon n \emph{center} vertices C⊆VC\subseteq V.
In \oursystem{}, we only store the CC information for $\numcenters=\rate \cdot n$ \defn{center} vertices $C\subseteq V$, where $\rate\in[0,1]$ is a user-defined parameter.
%As mentioned, our sketch only stores the influence (i.e., the connectivity information) for a set of \numcenters=ϵn\numcenters=\epsilon n selected \emph{centers}.
%We also use the fusion idea discussed in \cref{sec:prelim} to avoid explicitly storing the sampled graph. \yihan{make sure we discussed the fusion idea.}%, but instead, we generate the rr-th sketch using rr as the random seed. As such, we can determine if an edge is in the sampled graph on the fly.
%Similar to previous work, each sketch \sketchr\sketch{r} in \oursystem{} corresponds to an implicit sampled graph G′rG'_r from GG, where each edge (u,v)(u,v) is retained with probability puvp_{uv}.
%We define each sketch \sketchr\sketch{r} as a triple ⟨r,\ccsize[⋅],\lbl[⋅]⟩\langle r, \ccsize[\cdot],\lbl[\cdot]\rangle:
%Although we do not explicitly store G′rG'_r, we use the notation to denote the sampled graph of the rr-th sketch in the description.
Each sketch $\sketch{r}$ is a triple $\langle r,\lbl[\cdot],\ccsize[\cdot]\rangle$ corresponding to an implicit sampled graph $G'_r$ from $G$, where each edge $e=(u,v)$ is retained with probability $p_{e}$.

\begin{itemize}[leftmargin=*,noitemsep,topsep=.3em]
    \item ${r}$: the sketch id. Similar to the fusion idea mentioned in \cref{sec:prelim}, the sketch id fully represents the sampled graph.
    %\item C={c1,c2,…,c\nmcenters}⊆VC=\{c_1,c_2,\dots,c_\n mcenters\}\subseteq V: a set of \emph{centers} selected from VV uniformly at random.
    \item ${\lbl[1..\numcenters]}$: the CC label of center $c_i$ on this sketch, 
    %We use the smallest center id $j$ where $c_j$ is in the same CC as $c_i$ as the label.
    which is the smallest center id $j$ where $c_j$ is in the same CC as $c_i$.
    \item ${\ccsize[1..\numcenters]}$: if $\lbl[i]=i$ (i.e., $i$ represents the label of its CC),
    $\ccsize[i]$ is the influence of center $c_i$ on this sketch. It is initially the CC size of $c_i$ and becomes 0 when any vertex in this CC is selected as a seed.

    %indicating the connected component that cic_i is in. The centers with the same \lbl[⋅]\lbl[\cdot] value are in the same CC.
\end{itemize}

\hide{
The total space is $O(\numcenters)$ for each sketch.
%Besides the above per sketch information, we also store whether a vertex is a center, which is a boolean array of size nn.
In addition, we also store a boolean flag for each vertex to indicate if it is a center.
Therefore, the total memory usage for storing $R$ sketches is $O(R\numcenters + n)$ or $O((1+\rate R)n)$.
}
\myparagraph{Algorithm Overview}.
We first present the high-level idea of our sketch compression algorithm. %, which is inspired by previous theory work on the compact representation of graph connectivity~\cite{BBFGGMS18}.
%An illustration is presented in \cref{fig:sketch}.
%When $\alpha<1$, $\alpha n$ centers will be randomly drawn and stored using a boolean array.
We present our algorithm in \cref{alg:oursketch} and an illustration in \cref{fig:sketch}. 
As mentioned, we select $\rho=\rate n$ center vertices {uniformly at random}. 
We only store the CC information (label and size) for the centers in sketches. 
We use a global flag array to indicate if a vertex is a center,
and thus the total space is $O((1+\rate R)n)$.
%This roughly limits the auxiliary space in $O(\rate R n)$.
%As such, the total auxiliary space is O(ϵRn)O(\epsilon Rn) based on the user-defined parameter \epsilon.
%This is also experimentally verified (see \cref{fig:tradeoff}).
%To construct the sketches, we run parallel connectivity on each sketch (sampled graph), but only store the result (CC information) for the centers.
To retrieve the CC size of a vertex $v$ on sketch $\sketch{r}$, we start a breadth-first search (BFS) from $v$.
If any center $c_i$ is encountered, $v$ should activate the same set of vertices as $c_i$ on this sketch.
As such, we can stop searching and use $c_i$'s influence (CC size) as the influence for~$v$. 
If $v$ is not connected to any center,
the CC containing $v$ is likely small, and the BFS can visit all of them quickly.
In either case, the number of visited vertices in BFS can be bounded.
%running BFS from $v$ efficiently computes the influence of $v$.
%Based on the theoretical analysis in ~\cite{BBFGGMS18}, the number of vertices searched in the BFS is O(n/\numcenters)O(n/\numcenters) in expectation.
In \cref{thm:sketch}, we show that compressing the auxiliary space by a factor of $\rate$ roughly increases the evaluation time by a factor of $O(1/\rate)$.
%We show in \cref{thm:sketch-cost} that if we limit the auxiliary space by a factor of $\rate$, then we can roughly compute a vertex's true score (marginal gain) on $R$ sketches costs $O(R \cdot n/\numcenters)=O(R/\rate)$ work in expectation.
By controlling the number of centers,
% $\numcenters=\rate n$, 
we can achieve a tradeoff between the 
evaluation time and space usage.

%We present our algorithm in \cref{alg:oursketch}. 
Next, we elaborate on the three functions in \cref{alg:sketch_framework}: \sketchfunc{}$(G,r)$, which constructs the $r$-th sketch from $G$, \marginal$(S,v,\sketch{1..R})$, which computes the score (marginal gain) of a vertex $v$ on top of $S$ using sketches $\sketch{1..R}$, and \textsc{MarkSeed}$(s^*, \sketch{1..R})$, which adjusts the sketches $\sketch{1..R}$ when $s^*$ is selected as a seed.

\myparagraph{Sketch Construction (\sketchfunc{}$(G,r)$).} 
\hide{We start by selecting $\numcenters=\alpha n$ center vertices $C=\{c_1,c_2,\dots,c_{\numcenters}\}$ uniformly at random. We will only maintain the CC information of the $\numcenters$ centers on all sketches.}
%For all sketches, we use the same set of centers $C=\{c_1,c_2,\dots,c_{\numcenters}\}$ that are selected uniformly random from $V$.
Recall that we maintain CC information for $\rho=\rate n$ centers $C=\{c_1,c_2,\dots,c_{\numcenters}\}$ in sketches. 
To construct a sketch $\sketch{r}$, we first compute the CC information of the sampled graph $G'_r$, which can be performed by any parallel connectivity algorithm~\cite{dhulipala2020connectit}. We store the CC information for all centers in two arrays.
%we need to compute the CC information for each center $c_i$ on the sampled graph $G'_r$, and
%we use two arrays $\ccsize{}$ and $\lbl{}$ to store this information for all $\numcenters$ centers.
%We use a parallel union-find algorithm to compute this.
$\sketch{r}.\lbl[i]$ records the label of CC of the $i$-th center.
For multiple centers in the same CC, we simply use the smallest CC id as the label for all of them to represent this CC, so all centers find the CC information by referring to their label.
For a center $c_i$, If $i$ is the label of its CC, we use $\sketch{r}.\ccsize[i]$ to record the size of this CC.
%In our implementation, we use the smallest center id as the label of this CC.
An example of these arrays is given in \cref{fig:sketch}.
%Finally, in sketch $\sketch{r}$, we also need to record the labels of all centers in $\sketch{r}.\lbl[\cdot]$. This is because multiple centers can also be connected with each other, and we need to use the same CC id for all of them.

\myparagraph{Computing the Marginal Gain (\marginal$(S,v,\sketch{1..R})$).} Given the sketches $\sketch{1..R}$ and the current seed set $S$, the function \textsc{\marginal}$(S,v,\sketch{1..R})$ computes the marginal gain of a vertex $v$ by averaging the marginal gains of $v$ on all sketches.
We use a helper function $\langle \delta,l \rangle=$ \getcenter{}$(\sketch{r},v,S)$,
which returns $\delta$ as the marginal gain of $v$ on sketch $\sketch{r}$,
and $l$ as the label of centers connected to $v$ ($l=-1$ if no center is connected to $v$).
This function will run a breadth-first search (BFS) from $v$ on $G'_r$ (i.e., only using edges $e\in E$ s.t.\ \textsc{Sample}$(e,r)$ is \true{}).
If any center is encountered during this BFS (Examples 1 and 3 in \cref{fig:sketch}),
then the influence of $v$ is the same as $c_i$ on this sketch.
%Therefore the function use the information of $c_i$, and return $\delta=\sketch{r}.\ccsize[i]$ and $l=\sketch{r}.\lbl[i]$.
The information of $c_i$ is retrieved by its label $l=\sketch{r}.\lbl[i]$, and thus $\delta=\sketch{r}.\ccsize[l]$.
The influence $\delta$ is either the size of the CC containing $v$ when no vertices in this CC are seeds (Example 1 in \cref{fig:sketch}),
or 0 otherwise, as is updated in \textsc{MarkSeed} (Example 3 in \cref{fig:sketch}).
Otherwise, if the BFS terminates without visiting any centers, it will return -1 for the label $l$.
The influence $\delta$ is either the number of vertices $n'$ visited during the BFS, which is also the size of CC containing $v$ (Example 2 in \cref{fig:sketch}), or $\delta=0$ (\cref{line:nocenter}) if any seed is visited during BFS (Example 4 ).
\hide{
If any seed is visited during BFS, then the marginal gain is $\delta=0$ (\cref{line:nocenter}, Example 4 in \cref{fig:sketch}).
Otherwise, $\delta$ is the number of vertices $n'$ visited during BFS, which is also the size of CC containing $v$ (Example 2 in \cref{fig:sketch}).
}
%The function will return $-n'$ to indicate this case.
%We also deal with a special case where a seed is encountered during BFS, in which case the function simply returns 0 to indicate that no vertices will be further influenced (they have already been activated by a previous seed).
%With the \getcenter{} function, we can compute the influence of $v$ on each sketch. For sketch $\sketch{r}$, we first obtain $\langle \delta \rangle \gets$ \getcenter{}$(\sketch{r},v,S)$ to get a (possible) center connected with $v$. If $x$ is positive, we will use $\sketch{r}.\lbl[\cdot]$ to look up the corresponding CC label $l$.
Using the \getcenter{} function, the marginal influence of $v$ on sketch $\sketch{r}$ can be obtained as the first return value $\delta$ of \getcenter{}$(\sketch{r},v,S)$.
%The influence of $v$ is then $\sketch{r}.\ccsize[l]$.
%If $x\le 0$, $-x=|\CC{v}|$ is the size of the CC containing $v$.
%$v$ should activate all $-x$ vertices in $\CC{v}$.
%\yan{refer to the illustration.}

\myparagraph{Marking a Seed (\textsc{MarkSeed}$(s^*,\sketch{1..R})$).} The function updates the sketches when $s^*\in V$ is selected as a seed.
% For each sketch~$\sketch{r}$, the CC label of $s^*$ is the second return value $l$ of \getcenter{}$(\sketch{r},s^*,S)$.
% If $l\ne -1$, we set $\sketch{}[r].\ccsize[l]$ as 0---since $s^*$ is selected, all other vertices in this CC will get no marginal gain on this sketch.
For each sketch~$\sketch{r}$, the CC label of $s^*$ is the second return value of \getcenter{}$(\sketch{r},s^*,S)$.
If the label is not $-1$, we set $\sketch{}[r].\ccsize[l]$ as 0---since $s^*$ is selected, all other vertices in this CC will get no marginal gain on this sketch.

\hide{
\myparagraph{Selecting Centers. } We will start our algorithm by selecting $\numcenters$ \defn{center} vertices in the graph. In our sketches, we will only maintain the influence information from the $\numcenters$ centers to save space.
For all sketches, we use the same set of centers $C=\{c_1,c_2,\dots,c_{\numcenters}\}$ that are selected uniformly random from $V$.

\myparagraph{The \getcenter{} Algorithm.} With the centers selected, we first present a useful primitive \getcenter{}$(r,v)$.
Given the $r$-th sketch, this function returns the connectivity information for vertex $v$. If $v$ is connected to (can influence) any center $c_i$ on $G'_r$, this means that the influence of $v$ is the same as $c_i$ on this sketch. Therefore the function will return the center id $i$.
Otherwise, it will compute the size $n'$ of the connected component containing $v$, and return $-n'$, which indicates that the influence of $v$ on $G'$ is $-n'$. To get the return value, we start a breadth-first search (BFS) from $v$ on $G'_r$ (i.e., skipping edges $(u,v)$ s.t. \textsc{Sample}$(u,v,r)$ is \False{}). If any center is encountered during the BFS, we can stop and return the center id (line \ref{line:found_center}).
If $v$ is not connected to any center, the BFS will finish after visiting all vertices connected to $v$. Therefore, we can record the total number of vertices visited in the BFS as $n'$ and return $-n'$.

\myparagraph{Computing the Influence of Each Center.} For sketch $\sketch{r}$, the initial influence of center $c_i$ is its the connected component size on the sampled graph $G'_r$.
As such, we use a parallel union-find algorithm to compute this (Lines \ref{line:union}-\ref{line:find}).
For each edge $(u,v)$, this means combining the set of the labels $L[u]$ and $L[v]$.
If any of them is not connected to any center ($L[\cdot]<0$), then we ignore this edge since it does not impact the connectivity of the centers. Finally, the array $L[\cdot]$ stores the connected component id (which is also the id of the representative center) of each vertex or a negative value indicating that $v$ is not connected with any center. We then count the occurrences of each (positive) value $i$ in $L[\cdot]$ in $h[i]$, indicating the size of the connected component $i$ in this sketch.
Finally, in sketch $\sketch{r}$, we also need to record the labels of all centers in $\sketch{r}.\lbl[\cdot]$. This is because multiple centers can also be connected with each other, and we need to use the same CC id for all of them.

\myparagraph{Mark a Seed.} When a vertex $v$ is selected as a seed, we use the \textsc{MarkSeed} to update the sketches. For each sketch $\sketch{r}$, we will find the CC id of $v$ by $x\gets$\getcenter{}$(r,v)$. If $x$ is positive, we will use $\sketch{r}.\lbl[\cdot]$ to look up the corresponding CC label $l$.
Since $v$ is selected, all other vertices in the same CC will not get any marginal gain on this sketch, and thus we set $\sketch{}[r].\ccsize[l]$ as 0.

\myparagraph{Compute the Marginal Gain.} The function \textsc{\marginal}$(S,v,\sketch{1..R})$ computes the marginal gain of a vertex $v$ on sketches $\sketch{1..R}$ given the current seed set $S$.
We will enumerate all sketches and sum up the marginal gain of $v$ for all of them. For sketch $\sketch{r}$, we will find a center connected with $v$ by $x\gets$\getcenter{}$(r,v)$. If $x$ is positive, we will use $\sketch{r}.\lbl[\cdot]$ to look up the corresponding CC label $l$.
The marginal gain of $v$ is then $\sketch{r}.\lbl[l]$. This is either the size of the corresponding CC when no other vertices in the same CC are selected as seeds or 0 otherwise (as is updated in \textsc{MarkSeed}).
If $x$ is negative, $-x=|\CC{v}|$ is the size of the CC containing $v$.
$v$ should activate all $-x$ vertices in $\CC{v}$. If we encounter any seed in this process, the influence of $v$ is 0 on this sketch. Otherwise (Line xx), the influence of $v$ is $-x$.
}

%Based on the discussions above, it is clear that our algorithm computes the sketches and the marginal gains of vertices correctly on a fixed number of $R$ sketches.
%As such, the result computed by \oursystem{} is the same as the \textsc{StaticGreedy} algorithm, while allowing for any user-defined compaction rate to reduce space usage.

Our approach allows for a tradeoff between space and query time in \marginal{}$()$: a smaller $\rho$ (fewer centers) means less space but a higher evaluation cost, as it may take longer to find a center.
\oursystem{} unifies and is a hybrid of \staticgreedy{} and \infuser{}.
Theoretically, using $\rate=1$, our sketch is equivalent to \infuser{} where the CC information for all vertices on all sketches are memoized;
using $\rate=0$, our sketch is equivalent to \staticgreedy{} with no memoization,
and evaluations are done by traversing the sampled graph.
In practice, \oursystem{} is much faster than \staticgreedy{} and \infuser{} even when with no compression due to better parallelism. 
We summarize the theoretical guarantees in \cref{tab:related}, and state them in \cref{thm:sketch}.
\ifconference{We give the proof in the full version of this paper. }
%\revision{Note that to enable the bounds in \cref{thm:sketch}, it is important to use randomly selected vertices. }

\begin{theorem}\label{thm:sketch}
    \oursystem{} with parameter $\rate$ requires $O((1+\rate R)n)$ space to maintain $R$ sketches,
    and visits $O(R\cdot \min(1/\rate,T))$ vertices to re-evaluate the marginal gain of one vertex $v$,
    %where $T$ is the expected influence spread of this vertex on the sampled graphs.
    where $T$ is the average CC size of $v$ on all sketches.
\end{theorem}

%Note that the space is $O((1+\rate R)n)$ because we also maintain global per-vertex information (e.g., $\lazygain{\cdot}$), which takes $O(n)$ space.
%The $T$ term in the bound is true because the worst-case cost of our partial simulation is a \emph{full simulation}, which by definition is $T$.

\iffullversion{\begin{proof}
\hide{
                 Assume that no connectivity is stored in the sketches, then running a simulation for a vertex will visit $RT$ vertices, based on the definition of $T$.
                 This is also what the baseline algorithm \staticgreedy{} does.
                 We now consider the case that $\rho=\alpha n$ centers are selected.

                 On sketch $r$, assume the CC size of $v$ is $T_r$. Therefore, the BFS will visit all $T_r$ vertices in a certain order. Such BFS will terminate either 1) a center is encountered or 2) all $T_r$ vertices have been visited.
                 For case 1, each visited vertex has probability $\rate$ to be a center, since the centers are picked independently and uniformly at random. Therefore, the expected number of visited vertex before termination in the BFS is $\sum_{i=1}^{T_r}(1-\rate)^{i-1}\rate\cdot i$, which solves to $1/\rate$.
                 For case 2, at most $T_r$ vertices will be visited.
                 %Therefore, taking the minimum of the two (reaching a center or exhausting all vertices) gives the stated bound.
                 Therefore, the cost of evaluating $v$ on sketch $r$ is $O(\min(1/\rate, T_r))$. Considering $T=(\sum_{r=1}{R}T_r)/R$, adding the bound for all $R$ sketches gives the stated bound.
}
Assume that no connectivity is stored in the sketches, then running a simulation for a vertex will visit $RT$ vertices, based on the definition of $T$.
This is also what the baseline algorithm \staticgreedy{} does.
We now consider the case that $\rho=\alpha n$ centers are selected.
Note that all centers are picked independently and uniformly at random.
This means that each visited vertex in the BFS has the probability of $\rate$ that is a center.
Let us first focus on a search on a specific sketch $r$, and assume the CC size of $v$ is $T_r$.
The search terminates either 1) all $T_r$ vertices have been visited, or 2) a center is encountered.
Therefore, when visiting the $T'$ vertices in the BFS order, each of them stops the search with a probability of $\rate$.
The expected number of visited vertices in the BFS is therefore also bounded by $\sum_{i=1}^{\infty}(1-\alpha)^{i-1}\alpha\cdot i$, which solves to $1/\alpha$.
The total number of visited vertices in all BFS is $\sum_r \min(T_r,1/\alpha)=\min(RT,R/\alpha)=R\cdot \min(T,1/\alpha)$.
%Note that the search may run out of vertices in the CC before finding a center.
%Thus, taking the minimum of the two (reaching a center or exhausting all vertices) gives the stated bound.
\end{proof}
}

\iffullversion{
  Note that the centers are selected uniformly at random to bound the number of visited vertices in \cref{thm:sketch}. 
  %selection, 
  \hide{
  There may be other center selection heuristics, such as using betweenness centrality or PageRank centrality. 
  We leave this as future work. 
  }
  It is possible to use other center selection heuristics, such as betweenness or PageRank centrality. We leave this as future work. 
  %We did not consider them in this paper because they can not provide theoretical guarantee that is similar to \cref{thm:sketch}. 
  %However, we believe exploring centrality measurements as heuristics may be an interesting future work.
  % Although we can not find any theoretical guarantee for those selections, we believe exploring centrality measurements as heuristics may be an interesting  future work.
}

\hide{
\begin{theorem}
  On undirect graphs with IC model, using sketches generated in \cref{alg:oursketch} and a greedy seed selection approach gives the same result as the \textsc{StaticGreedy} algorithm.
\end{theorem}
\yihan{Not sure if we want to keep the above theorem. Also not sure if we've introduced enough about staticgreedy such that this point is clear.}
}

% \section{Parallel Priority Queues for Seed Selection}\label{sec:selecting}
\section{Parallel Seed Selection}\label{sec:selecting}

We now present the parallel seed selection process in \oursystem{}.
%we present our new data structures to parallelize the  in CELF.
We call each iteration in seed selection (selecting one seed) a \defn{round}.
Recall that prior solutions use the CELF optimization (see \cref{tab:related,sec:prelim}), which maintains (possibly stale) scores of all vertices in a priority queue $Q$ and updates them lazily.
In each round, CELF pops the vertex with the highest (stale) score from $Q$ and re-evaluates it.
The process terminates when a newly evaluated score is higher than all scores in $Q$; otherwise, the new score will be inserted back into $Q$.
For simplicity, we assume no tie between scores.
In the cases when ties exist, we break the tie by vertex id.
%\letong{For simplicity, we assume scores do not have tie.  In practice, tie is very rare. We can break tie using vertex ids.}
Let $F_i=\{ v ~|~ \lazygaini{i-1}{v} \ge \Delta^*\}$ be the set of vertices re-evaluated by CELF in round $i$, where $\lazygaini{i-1}{v}$ is the stale score of $v$ after round $i-1$,
and $\deltas$ is the maximum true score in round $i$ (i.e., the score of the chosen seed in round $i$).
\hide{
Empirically, $|F_i|$ is usually large in the first several rounds, but can be rather small in later rounds, in which case CELF greatly reduces the number of re-evaluations.
}
\revision{
We experimentally study the distribution of $|F_i|$ on all graphs and show three representative graphs in \cref{fig:celf_evaluations}.
Except for road networks (e.g., GER in \cref{fig:celf_evaluations}),
most graphs may require a large number of re-evaluations in certain rounds.} %This indicates the demand for parallelizing CELF to evaluate multiple vertices in parallel.
%Most social and web graphs show similar trends to LJ, which have initial re-evaluations close to $n$ and subsequent re-evaluations in a much smaller range. Road graphs resemble GER, which usually have no more than 5 re-evaluations in each round. CH5 represents $k$-NN graphs, where the number of re-evaluations varies in a large range across rounds. Except for road graphs, there exists a demand for parallel CELF on different kinds of graphs.
%}
However, CELF is inherently sequential.
Some existing parallel implementations (e.g.,~\cite{gokturk2020boosting, minutoli2019fast}) only parallelize the evaluation function $\marginal$ (\cref{line:marginal} in \cref{alg:CELF_framework}),
but leave the CELF process sequential, and perform all $|F_i|$ evaluations one by one.
%still re-evaluate vertices one by one sequentially as in CELF. %pop candidate vertices and re-insert them sequentially.
%This leads to $O(|F_i| \cdot D_{\Delta})$ span, where $D_{\Delta}$ is the span of computing $\gain{\cdot}$ for one vertex.
%This leads to $|F_i|$ iterations of re-evaluations.
%The span will be very large when set $F_i$ is large, which indicates low parallelism.
When $F_i$ is large (as in social and web networks), the sequential CELF results in low parallelism. %since all candidates have to be re-evaluated one by one.
\begin{figure}[t]
  \centering
  \includegraphics[width=\columnwidth]{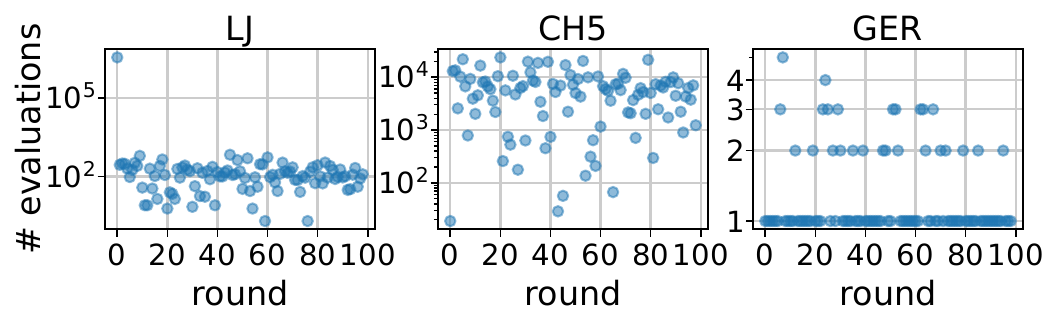}
  % \vspace{-0.2em}
  \caption{\small \revision{\textbf{The number of evaluations by CELF in each round.} A point $(x,y)$ means that CELF does $y$ re-evaluations in the $x$-th round.}}
  \label{fig:celf_evaluations}
\end{figure}

%However, the process of popping and re-inserting the vertex from the priority queue (\cref{line:Q_start} and \cref{line:Q_end} in \cref{alg:CELF_framework}) is inherently sequential.

\hide{Usually, $|F_i|$ is very small. From our observation, it is usually large in the first several rounds (proportional to $n$), and small in the latter rounds (only tens to hundreds).
}

\hide{
Besides memory usage, another challenge of running IM on large-scale graphs is running time.
% We want to speed up the running time by introducing parallelism without much extra work.
The sketch construction process (\cref{line:sketch_begin}-\cref{line:sketch_end} in \cref{alg:sketch_framework}) is easy to parallelize, while parallelizing the seed selection process with asymptotic the same work as CELF
% (\cref{line:next_seed} in \cref{alg:sketch_framework})
(\cref{alg:CELF_framework}) is challenging.
% In this section, we propose two methods to parallelize the CELF seed selection process (\cref{alg:CELF_framework}) without introducing much extra work compared to CELF.

Sequentially, CELF can greatly reduce the number of vertices needed in marginal gain re-evaluation. The set of vertices to be re-evaluated by CELF to select the $i$-th seed is $F_i=\{ v | \Delta_{i-1}[v] > \Delta^*\}$, where $\Delta_{i-1}[v]$ is the lazy-evaluated marginal value of $v$ before round $i$, and $\Delta^*$ is the maximum marginal influence in round $i$. Usually, $|F_i|$ is very small. From our observation, it is usually large in the first several rounds (proportional to $n$), and small in the latter rounds (only tens to hundreds).

However, the process of popping and re-inserting the vertex from the priority queue (\cref{line:Q_start},\cref{line:Q_end} in \cref{alg:CELF_framework}) is inherently sequential and hard to parallelize. As a result, some parallel implementations~\cite{gokturk2020boosting} only parallelize the re-evaluation function $\marginal$ (\cref{line:marginal} in \cref{alg:CELF_framework} ), but pop candidate vertices and re-insert them one by one. The span of this method is $O(|F_i| \cdot D_{\Delta})$, where $D_{\Delta}$ is the span of computing $\Delta$ for one vertex.  The span will be very large when set $F_i$ is large, which is true for the first several rounds in practice.

Alternatively, one can trivially parallelize seed selection without CELF by computing all $\gain{\cdot}$ values in parallel
and finding the maximum score.
Although all evaluations are embarrassingly parallel,
%it always re-evaluates $n$ vertices, which is much more expensive than CELF when $|F_i| \ll n$.
%the total work to re-evaluate all $n$ vertices is much higher than CELF when $|F_i| \ll n$.
evaluating all $n$ vertices is much more expensive than CELF when $|F_i| \ll n$.
%it always needs to re-evaluate $n$ vertices, compared to $|F_i|$ in CELF.
%Therefore, this approach can be more expensive than the sequential CELF, because in most of the rounds $|F_i| \ll n$.
%the work is $O(n \cdot W_{\Delta})$, where $W_{\Delta}$ is the work of computing $\Delta$ for one vertex. Compared to the work of CELF $O(|F_i|\cdot W_{\Delta})$, the work of this method is much larger if $|F_i| \ll n$, which is true for most of the rounds.
Hence, our goal is to make this process \defn{highly parallel}, while \defn{keeping the work close to CELF}. %, ideally re-evaluating asymptotically the same number of vertices as CELF.
To do so, we design \defn{parallel priority queues} to replace the sequential one in CELF.
%Our goal is to re-evaluate multiple vertices in parallel, and the total number of vertices we re-evaluate is not much larger than that of CELF.
}
\myparagraph{Prior Work on Parallel Priority Queue}.
As a fundamental data type,  parallel priority queues
are widely studied~\cite{wang2020parallel,dong2021efficient,blelloch2016just,sun2018pam,sun2019parallel,bingmann2015bulk,sanders2019sequential}.
However, as far as we know, all these algorithms/interfaces require knowing the batch of operations (e.g., the threshold to extract keys or the number of keys to extract) \emph{ahead of time}.
This is not true in CELF---the set $F_i$ is only known during the execution.
Thus, we need different approaches to tackle this challenge.

\myparagraph{Overview of Our Approaches}.
We first formalize the interface of the parallel priority queue needed in \cref{alg:sketch_framework}.
The data structure maintains an array $\lazygain{\cdot}$ of (stale) scores for all vertices.
It is allowed to call \marginal{} function to re-evaluate and obtain the true score of any vertex.
The interface needs to support \nextseed{} function, which returns the vertex id with the highest true score.

We present two parallel data structures to maintain the scores.
The first one is based on a parallel binary search tree (BST) called \textbf{\ptree{}}~\cite{blelloch2016just,sun2018pam,blelloch2022joinable,dhulipala2022pac}.
We prove that using \ptree{s}, our approach has work (number of evaluations) asymptotically the same as CELF,
while is highly parallel: the selection of the $i$-th seed finishes in $\log|F_i|$ iterations of evaluations (each iteration evaluates multiple vertices in parallel),
%We prove that by a prefix-doubling approach plus the algorithms on \ptree{}, the work (number of re-evaluations) is asymptotically the same as that in CELF, while the $i$-th seed selection finishes in $\log|F_i|$ iterations of re-evaluation (each iteration may evaluate multiple vertices in parallel),
instead of $|F_i|$ iterations (one vertex per iteration) in CELF.
% in existing work for both theoretical analysis and implementation.
%We prove that in our \ptree-based approach, the number of re-evaluated vertices is asymptotically the same as that in CELF.
We also propose a new data structure \textbf{\ourtree{}}, based on parallel winning trees.
%Our second approach is based on a new data structure \textbf{\ourtree{}} proposed in this paper, which is based on winning tree (aka.\ tournament tree).
\wintree{} does not maintain the total order of scores, and is simpler and potentially more practical than \ptree{s}.
%has smaller memory footprint and is highly asynchronous (finish in one iteration), leading to better space usage and potentially better practical performance.
We introduce the \ptree{}-based approach in \cref{sec:BST_section} and \ourtree{}-based approach in \cref{sec:WinTree}, and compare their performance in \cref{sec:exp-ana}.
\emph{These two approaches are \textbf{independent of the sketching algorithm} and apply to seed selection on all submodular diffusion models} (not necessarily the IC model and/or on undirected graphs).
Note that most IM diffusion models are submodular (e.g., IC, Linear Threshold (LT), Triggering
 (TR~\cite{kempe2003maximizing}), and more~\cite{tsaras2021collective, zhang2021grain, kempe2003maximizing}).

\hide{
Our goal is to re-evaluate multiple vertices in parallel, and the total number of vertices we re-evaluate is not much larger than that of CELF. In this section, we propose two methods using different parallel priority queues to maintain the vertices by their lazy-evaluated marginal gains thus enabling re-evaluating multiple vertices at once. The \BST method uses a parallel binary search tree (BST) as a priority queue and the prefix-doubling technique to guarantee the theoretical work bound. The \ourtree method uses a tournament tree we designed for this problem as a priority queue and performs more efficiently in practice.
We introduce the \BST method in \cref{sec:BST_section} and \ourtree in \cref{sec:WinTree}. These two seed selection methods are independent of the previous compact sketching algorithm and you can independently apply them or combine them together.
% We use the abbreviations \BST and \ourtree to refer to the two methods respectively.
}
%\begin{figure}[th]
%    \centering
%    \includegraphics[width=\columnwidth]{figures/bst.pdf}
%    \caption{
%        \small
%        BST Based Seed Selection
%        \label{fig:winning_tree}
%    }
%\end{figure}

\begin{figure*}
    \centering % \vspace{-1em}
    \includegraphics[width=0.95\textwidth]{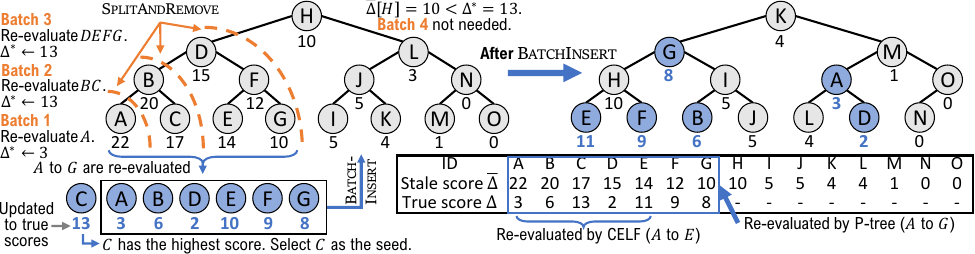}
    \caption{
        \small
        \textbf{Example of \ptree{}-Based Seed Selection.} The letters in the tree nodes represent vertices, and the numbers below them are their stale scores. \ptree{} maintains decreasing order of the (stale) scores. By prefix-doubling, we extract batches of 1, 2, 4 vertices and evaluate each batch in parallel. After batch 3, the highest true score (13) is higher than the current best in the tree (10), and the algorithm stops. We will select the node with the highest true score and insert the rest back to the tree with their new score. \ptree{} may evaluate more vertices than CELF, but the extra work can be bounded (\cref{lemma:efficiency}).
        \label{fig:bst}
    }% \vspace{-.5em}
\end{figure*}

\begin{algorithm}[t]
\small
\caption{Seed Selection based on \ptree\label{alg:BST}}
\SetKwProg{myfunc}{Function}{}{}
\SetKwFor{parForEach}{ParallelForEach}{do}{endfor}
\SetKwFor{mystruct}{Struct}{}{}
\SetKwFor{pardo}{In Parallel:}{}{}
\SetKwInOut{Maintains}{Maintains}
\Maintains{A parallel binary search tree $T$ for all $\lazygain{v}$. }
\DontPrintSemicolon
\myfunc{\upshape\textsc{NextSeed}$(S,\sketch{1..R})$} {
%$\Delta^*\gets 0$\\
$s^{*}\gets \bot$\tcp*[f]{The best seed so far}\\
$j\gets 0$\\
%\While {$\Delta^*<T.$\textsc{Max}}
% \Repeat{$\lazygain{s^*}\ge T.$\textsc{Max(\,) 
\Repeat{$\lazygain{s^*} > T.$\textsc{Max(\,) 
\label{line:BST_stop}}\tcp*[f]{$\lazygain{s^*}$ is better than the top in $T$}} {
\tcp{Extract (remove \& output) the top $2^j$ elements in $T$ into array $B_j$}
  $B_j[1..2^j]\gets T.$\textsc{SplitAndRemove}$(2^j)$ \label{line:split}\\
  %\tcp{Evaluate the true score for each $v\in B_j$}
  \parForEach(\tcp*[f]{Get the true score for each $v\in B_j$}){$v\in B_j$ \label{line:BST_evalue_begin}} {
    $\lazygain{v}\gets\marginal{}(S,v,\sketch{1..R})$ \label{line:BST_evalue_end}\\
  }
  %$\Delta^*\gets \max(\Delta^*, \max_{v\in B}\gain{v})$
  $t\gets \arg\max_{v\in B_j}\lazygain{v}$ \tcp*[f]{can be computed in parallel}\label{line:BST_max}\\
  \lIf{($s^*=\bot$) \textsc{or} ($\lazygain{t}>\lazygain{s^*}$)}{
    $s^*\gets t$ \label{line:BST_update_seed}
  }
  $j\gets j+1$\\
}
  $T.$\textsc{BatchInsert}$(\bigcup_{j'=0}^{j-1} B_{j'} \setminus \{s^*\})$ \label{line:batch_insert}\\
  \Return{$s^*$}
}
\end{algorithm} 

\subsection{Parallel Priority Queue Based on \ptree{}} \label{sec:BST_section}

\hide{
To parallelize CELF, our first approach is to maintain the total (decreasing) order of the scores (marginal gains) of all vertices, using a binary search tree (BST).
As such, suppose we know that CELF will re-evaluate $|F_i|$ vertices in round $i$, we can find $F_i$ based on the total order given by the BST, re-evaluate them in parallel, and update them all in the BST.
We first discuss the high-level idea using an existing data structure \ptree{}, and then show how to deal with the case where we do not know $F_i$ ahead of time.

Our parallel approach is based on the \ptree{}~\cite{blelloch2016just,blelloch2022joinable} from the PAM library~\cite{sun2018pam}.
%, which supports a full interface of parallel tree operations such as insertions and deletions in batches.
In our case, we will use two functions on a \ptree{} $T$: 1) $T.\textsc{SplitAndRemove}(k)$, which removes and outputs (to an array) the first $k$ tree nodes ($k$ largest scores) from $T$, and 2) $T.\textsc{BatchInsert}(B)$, which inserts a set of keys $B$ to $T$.
Both algorithms are parallel with polylogarithmic span.

We now deal with the key challenge that we are unaware of the number of re-evaluations, i.e.\ $|F_i|$, ahead of time.
Our idea is to use \emph{prefix doubling} to achieve both efficiency and high parallelism, inspired by some recent work on parallel algorithms~\cite{blelloch2016parallelism,gu2023parallel,gu2022parallel,shen2022many,gbbs2021}.
The pseudocode and illustration is given in \cref{alg:BST} and \cref{fig:bst}, respectively.
%dual-binary search this value, while maintaining all parts efficient and highly parallel.
When selecting the $i$-th seed, our tree $T$ starts with the stale values $\lazygain{\cdot}$ of the vertices from the previous round (entering the repeat-until loop).
It will then extract the top (largest stale score) nodes in batches of size $1, 2, 4, 8, ...$ from $T$ (\cref{line:split}) as an array $B_i$.
%For each vertex $v\in B_i$, its score will be re-evaluated (\cref{line:BST_evalue_end}).
We will re-evaluate all vertices in $B_i$ in parallel (\cref{line:BST_evalue_end}).
These new scores are used to update the current best seed $s^*$.
The loop terminates when the score of ${s^*}$ is no less than the top score in $T$ (\cref{line:BST_stop}).
Finally, the algorithm selects $s^*$ as the seed and inserts the rest of the newly re-evaluated values $\bigcup B_j \setminus \{s^*\}$ back to $T$.
%An illustration of this process is in \cref{fig:bst}.
}

Our first approach to parallelizing CELF is to maintain the total (decreasing) order of the scores of all vertices, using a parallel binary search tree (BST) called \ptree{}~\cite{blelloch2016just,blelloch2022joinable}.
%As such, suppose we know that CELF will re-evaluate $|F_i|$ vertices in round $i$, we can find $F_i$ based on the total order given by the BST, re-evaluate them in parallel, and update them all in the BST.
%We first discuss the high-level idea using an existing data structure \ptree{}, and then show how to deal with the case where we do not know $F_i$ ahead of time.
%, which supports a full interface of parallel tree operations such as insertions and deletions in batches.
We will use two functions on a \ptree{} $T$: 1) $T.\textsc{SplitAndRemove}(k)$, which extracts (removes and outputs to an array) the first $k$ tree nodes ($k$ largest scores) from $T$, and 2) $T.\textsc{BatchInsert}(B)$, which inserts a set of keys $B$ to $T$. Both algorithms are parallel with a polylogarithmic span.

\hide{Different from many priority queue implementations (e.g., binary heap) that extract the top element one at a time,
\ptree{} maintains the total order of all vertices, and thus we can extract a batch of vertices with top (stale) scores and evaluate them in parallel.}
Since \ptree{} maintains the total order of all vertices, we can extract a batch of vertices with top (stale) scores and evaluate them in parallel.
The key challenge is a similar stop condition as CELF to avoid evaluating too many vertices, since we are unaware of the number of ``useful'' vertices, i.e., $|F_i|$, ahead of time.
Our idea is to use \emph{prefix doubling}~\cite{shun2015sequential,blelloch2016parallelism,gu2023parallel,gu2022parallel,shen2022many,gbbs2021,wang2023parallel} to achieve work efficiency and high parallelism.
The pseudocode is given in \cref{alg:BST} with an illustration in \cref{fig:bst}.
%dual-binary search this value, while maintaining all parts efficient and highly parallel.
To find the next seed, the \ptree{} $T$ starts with the stale scores $\lazygain{\cdot}$ from the previous round.
We then extract the top (largest stale score) nodes in batches of size $1, 2, 4, 8, ...$ from~$T$ (\cref{line:split}).
%For each vertex $v\in B_i$, its score will be re-evaluated (\cref{line:BST_evalue_end}).
Within each batch $B_j$, we re-evaluate all vertices in parallel (\cref{line:BST_evalue_end}).
These new scores are used to update the current best seed~$s^*$.
The loop terminates when the score of ${s^*}$ is
% no worse
better
than the best score in $T$ (\cref{line:BST_stop}).
Finally, we select $s^*$ as the seed and insert the rest of the new true scores $\bigcup B_j \setminus \{s^*\}$ back to $T$.
%An illustration of this process is in \cref{fig:bst}.

Due to prefix doubling, each seed selection finishes in at most $O(\log n)$ rounds.
%This indicates high parallelism even when $|F_i|$ is large.
Note that our approach evaluates more vertices than CELF, but due to the stop condition (\cref{line:BST_stop}),
the extra work is bounded by a constant factor (proved in \cref{lemma:efficiency}).
%However, we need to show that this parallel approach is correct, and more importantly, it does not introduce asymptotically extra work, which is non-trivial.
%Next, we prove the correctness and efficiency of the algorithm. %of our algorithm in \cref{lemma:connectness}.

\begin{theorem}[\ptree{} Correctness]\label{lemma:connectness}
\cref{alg:BST} always selects the next seed with the largest marginal gain, i.e., $\gain{s^*}=\max_{v\in V}\{\gain{v}\}$.
%$\gain{s^*~|~S}=\max_v\{\gain{v~|~S}~|~v \in V\}$ for a given seed set $S$.
\end{theorem}
\ifconference{Due to the space limit, we defer the proof to the full version. }
\iffullversion{
\begin{proof}
  \hide{Let $s$ be the vertex with the largest true score, we will show that at the end of the algorithm, $s^*=s$.
  We first show that $s$ will be re-evaluated by showing that no other vertex $s'$ that is evaluated in an earlier batch than $s$ can terminate the loop.
  Re-evaluating $s'$ will update the score of $s'$ to its true score $\gain{s'}$.
  By definition, $\gain{s'}<\gain{s}<\lazygain{s}$. Since $s$ is still in $T$, the condition on \cref{line:BST_stop} will not be triggered before $s$ is re-evaluated.
  Since $s$ has the highest true score, it will be assigned to $s^*$ and never be updated by \cref{line:BST_update_seed}. Therefore, at the end of the algorithm, $s^*=s$.
  }
  Let $s^*$ be the vertex selected by the algorithm, and we will show $\gain{s^*}\ge \gain{v}$ for all $v\in V$ when the stop condition on \cref{line:BST_stop} is triggered.
  We first show that $\gain{s^*}\ge \gain{v}$ for all $v\in V$ that has been split from the tree.
  This is because all such vertices have been re-evaluated, and $s^*$, by definition, has the highest true score among them.
  We then show $s^*$ has a higher true score than any other vertices still in $T$.
  For any $u\in T$, the stop condition indicates $\gain{s^*}=\lazygain{s^*}
  % \ge
  >
  \lazygain{u}\ge\gain{u}$ (due to submodularity). Therefore, $\gain{s^*}$
  is the highest true score among all vertices.
  \end{proof}
}

%\begin{proof}
%%WLOG we first assume no ties for the marginal gains.
%%If so, we can break tie consistently (e.g., using the vertex labels).
%When running \cref{alg:BST}, the set of vertices that have been re-evaluated are $\bigcup B_j$.
%For $v\in \bigcup B_j$, we have $\lazygain{v}=\Delta [v~|~S]$ since its value is up-to-date.
%Hence, the selected seed $s^*=\argmax_v\{\lazygain{v}\}$ is $\argmax_v\{\Delta [v~|~S]\}$ for $v\in \bigcup B_j$.
%
%Next we show $\lazygain{s^*}\ge \max_v\{\Delta [v~|~S] ~|~v\in V\setminus(\bigcup B_j\cup S)]\}$.
%Note that vertices in $V\setminus(\bigcup B_j\cup S)$ correspond to nodes remain untouched in $T$ throughout \cref{alg:BST}.
%Consider a vertex $u$ in this set.
%Due to the checking on \cref{line:BST_stop}, we know $\lazygain{u}\le \lazygain{s^*}$.
%Also, since $\lazygain{u}$ is evaluated in an earlier round, so $\lazygain{u}=\gain{u~|~S'}$ for $S'\in S$.
%Due to the submodularity, we know $\gain{u~|~S'}\ge \gain{u~|~S}$, so $\lazygain{s^*}\ge \gain{u~|~S}$.
%Combining both cases proves that $s^*=\argmax_v\{\Delta [v~|~S]~|~v \in V\setminus S\}$.
%\end{proof}

%We say the \emph{rank} of a vertex as the order its appearance

\begin{theorem}[\ptree{} Efficiency]\label{lemma:efficiency}
\cref{alg:BST} has the total number of evaluations at most twice that of CELF.
\end{theorem}

%One can consider CELF as \cref{alg:BST} but extracts one vertex at a time.

\begin{proof}
Recall that $F_i$ is the set of evaluated vertices by CELF when selecting the $i$-th seed. Let $F'_i$ be the set of evaluated vertices by \ptree{s} in \cref{alg:BST}.
We first show a simple case---if both CELF and \cref{alg:BST} start with the same stale scores $\lazygaini{i-1}{\cdot}$,
then $|F'_i|\le 2|F_i|$.
Let us reorder vertices in $V$ as $v_1,v_2,\dots,v_n$ by the decreasing order of their stale score $\lazygaini{i-1}{\cdot}$.
Assume $v_l$ is the last vertex evaluated by CELF, so $|F_i|=l$.
$v_l$ must be in the last batch in \ptree{}.
%, and assume that batch has $2^j$ vertices.
Assume the last batch is batch $j$ with $2^j$ vertices.
This indicates that all $2^{j}-1$ vertices in the previous $j-1$ batches are before $v_l$.
Therefore $2^{j}-1 < l=|F_i|$, and $|F'_i|= 2^{j+1}-1$, which proves $|F'_i|\le 2|F_i|$.
%Hence, $|F'_i|=1+2+\cdots+2^j=2^{j+1}-1$.
%Note that
%which proves $|F'_i|\le 2|F_i|$.

We now consider the general case.
%We will call the round to select the $i$-th vertex the \emph{round $i$}.
We first focus on a specific seed selection round $i$.
Due to different sets of vertices evaluated in each round,
at the beginning of round $i$,
CELF and \ptree{} may not see exactly the same stale scores.
We denote the stale score \emph{at the beginning of round $i$} in
CELF as $\lazygainc{\cdot}$ and that for \ptree as $\lazygaint{\cdot}$.
We reorder vertices by the decreasing order of $\lazygainc{\cdot}$ as $\vvv{1},\vvv{2},\dots, \vvv{n}$,
and similarly for $\lazygaint{\cdot}$ as $\uuu{1},\uuu{2},\dots, \uuu{n}$.
%We will focus on analyzing a specific seed selection round $i$.
Denote $\deltas$ as the highest true score in round $i$.
%We then denote $x_i$ as the rank of $\deltas{i}$ in the sequence $\vvv{i}{1..n}$,
%i.e., $x_i=\arg\max\{l: \vvv{i}{l}>\deltas{i}\}$.
%i.e., $\vvv{i}{x_i}$ is the last vertex in $\vvv{i}{1..n}$ with stale score larger than $\deltas{i}$.
%Similarly, we define $y_i$, such that $\uuu{i}{y_i}$ is the last vertex in $\uuu{i}{1..n}$ with stale score larger than $\deltas{i}$.
Let $\vvv{x_i}$ be the last vertex in $v_{1..n}$ such that
% $\lazygainc{\vvv{x_i}}>\deltas$,
$\lazygainc{\vvv{x_i}}\ge\deltas$,
and
$\uuu{y_i}$ the last vertex in $u_{1..n}$ such that
% $\lazygaint{\uuu{y_i}}>\deltas$.
$\lazygaint{\uuu{y_i}}\ge\deltas$.
%$\uuu{i}{y_i}$ be the last vertex in $\uuu{i}{1..n}$ such that $\lazygaint{i}{\uuu{i}{y_i}}>\deltas$.
Namely, $x_i$ is the rank of $\deltas$ in $v_{1..n}$, and $y_i$ is the rank of $\deltas$ in $u_{1..n}$.
By definition, $F_i$ is exactly the first $x_i$ vertices in sequence $v$, and thus $x_i=|F_i|$;
$F'_i$ contains all vertices smaller than $\uuu{y_i}$
% in $\uuu{i}{}$
and possibly some more in the same batch with $\uuu{y_i}$, so $|F'_i|\le 2y_i$.
In the simple case discussed above, where we assume $\lazygainc{\cdot}=\lazygaint{\cdot}$,
we always have $x_i=y_i$.

\hide{
We will use amortized analysis to show that $\sum y_i \le \sum x_i$, which further indicates $\sum |F'_i|\le 2\cdot\sum|F_i|$.
Intuitively, \ptree{} may evaluate more vertices than CELF.
These additional re-evaluations may update $\lazygaint{\cdot}$ to newer (smaller) scores,
and prevent re-evaluating them in later rounds, and thus are not ``wasted'' in many cases.
Indeed, if $\lazygaint{t} \le \lazygainc{t}$ for all $t\in V$,
then the rank of $\deltas$ must be earlier in $u_{1..n}$ than in $v_{1..n}$, indicating that $y_i\le x_i$.

\begin{figure*}[th]
    \centering % \vspace{-1em}
    \includegraphics[width=0.95\textwidth]{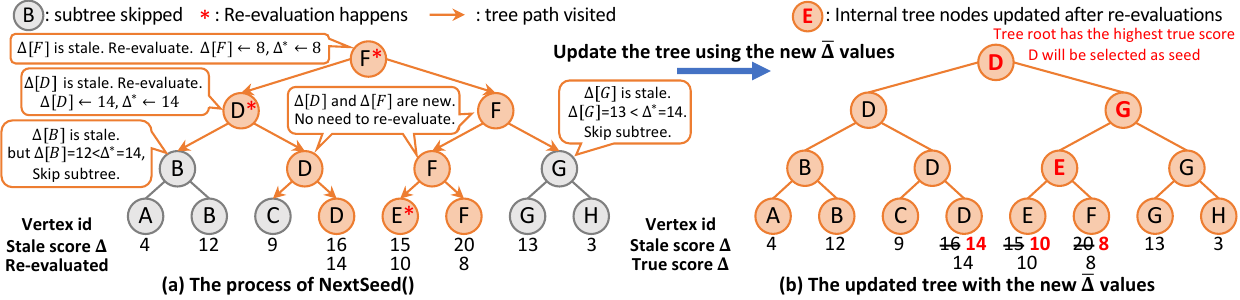}
    \caption{
        \small
        \textbf{Seed selection based on \ourtree{}}. Each leaf stores a vertex id, and each internal node stores the vertex in its subtree with the highest (stale) score. (a) An example of finding the maximum score. For illustration purposes, we assume the parallel threads work at the same speed and all tree nodes on the same level are processed in parallel (in reality the threads run asynchronously in fork-join parallelism).
        Therefore, the subtree at $G$ will see $\Delta^*=14$ updated by $D$, and this subtree will be skipped. %By using this property, only a subset of vertices are re-evaluated.
        (b) Updating the internal nodes with the new $\lazygain{\cdot}$ values. Finally, the root of the tree has the highest true score.
        \label{fig:winning_tree}
        %\vspace{-1em}
    }
\end{figure*}

However, due to different re-evaluation patterns, it is possible that for some $t\in V$, we have $\lazygainc{t} < \lazygaint{t}$.
Such cases happen iff. $t$ is evaluated by CELF in a previous round, but not by \ptree{}.
Namely, there exists a round $j$, such that $t\in F_j$ but $t\notin F'_j$.
Note that $F_j=v^{*}_{1..x_j}$, and $u^{*}_{1..y_j}\subseteq F'_j$,
where $v^{*}_{1..n}$ and $u^{*}_{1..n}$ are the $u$ and $v$ sequences in round $j$, respectively.
Namely, such vertex $t$ was counted in $x_j$, but not $y_j$.
Therefore, we will save a token for such vertex $t$ when it is evaluated by CELF in round $j$,
such that when later $t$ is counted in $u_{1..y_i}$, we will use the token to count $t$ for free.
We can partition $u_{1..y_i}$ into two categories: $U_1$ as the intersection of $u_{1..y_i}$ and $v_{1..x_i}$,
and $U_2$ as the rest.
%Based on the same argument, for any vertices $w\in U_2$, it is always because a previous round in CELF evaluates $w$
All vertices in $U_1$ are counted in $x_i$, and all vertices in $U_2$ can be counted by the saved tokens.
Therefore, by amortized analysis, we have $\sum y_i \le \sum x_i$.

Recall that $x_i=|F_i|$, and $|F'_i|\le 2y_i$. Since $\sum y_i \le \sum x_i$, we proved that
$\sum |F'_i|\le \sum 2y_i\le 2\sum x_i = 2\sum|F_i|$.
}

We will use amortized analysis to show that $\sum y_i \le \sum x_i$, which further indicates $\sum |F_i'|\le 2\cdot\sum|F_i|$. We can partition $u_{1..y_i}$ into two categories: $U_1$ as the intersection of $u_{1..y_i}$ and $v_{1..x_i}$,
and $U_2$ as the rest.
Note that for $t \in U_2$,  $\lazygaint{t} \geq  \deltas > \lazygainc{t}$. This happens iff.\ $t$ is evaluated by CELF in a previous round but not by \ptree{}.
Namely, there exists a round $j$, s.t. $t\in F_j$ but $t\notin \bigcup_{x=j}^{i-1}F'_x$.
%$t\notin (F'_j \cup F'_{j+1} \cup ... \cup F'_{i-1})$.
Let $v^{*}_{1..n}$ and $u^{*}_{1..n}$ be the $u$ and $v$ sequences in round $j$, respectively.
Note that $F_j=v^{*}_{1..x_j}$, and $u^{*}_{1..y_j}\subseteq F'_j$.
Since $t\in F_j$, then vertex $t$ was counted in $x_j$.
Since $t\notin F'_j$, $t$ is not counted in any $y_j$.
Therefore, we can save a token for such $t$ when it is evaluated by CELF in round $j$ but not by \ptree{},
such that when later $t$ is counted in $u_{1..y_i}$ in round $i$, we will use the token to count $t$ for free.
Note that $t\notin (F'_j \cup F'_{j+1} \cup ... \cup F'_{i-1})$,
for the same reason, $t$ is not counted in any $y_{j'}$ for $j\le j'< i$, so the token must still be available in round $i$.
\hide{
Let $v^{l}_{1..n}$ and $u^{l}_{1..n}$ be the $u$ and $v$ sequences in round $l$, respectively.
Then for all $l$, we have $F_l=v^{*}_{1..x_l}$, and $u^{*}_{1..y_l}\subseteq F'_l$.
Since $t\in F_j$, then vertex $t$ was counted in $x_j$.
Therefore, we can save a token for such $t$ when it is evaluated by CELF in round $j$ but not by \ptree{},
such that when later $t$ is counted in $u_{1..y_i}$ in round $i$, we will use the token to count $t$ for free.
Note that $t\notin (F'_j \cup F'_{j+1} \cup ... \cup F'_{i-1})$,
therefore $t$ is not counted in any $y_{j'}$ for $j\le j'< i$, and the counter must still available in round $i$.
}
In summary, all vertices in $U_1$ are counted in $x_i$, and all vertices in $U_2$ can be counted by the saved tokens (charged to some previous $x_j$).
Therefore, using amortized analysis, we have $\sum y_i \le \sum x_i$.

Recall that $x_i=|F_i|$, and $|F'_i|\le 2y_i$. Since $\sum y_i \le \sum x_i$, we proved that
$\sum |F'_i|\le \sum 2y_i\le 2\sum x_i = 2\sum|F_i|$.
\end{proof}

\begin{figure*}[th]
    \centering % \vspace{-1em}
    \includegraphics[width=0.95\textwidth]{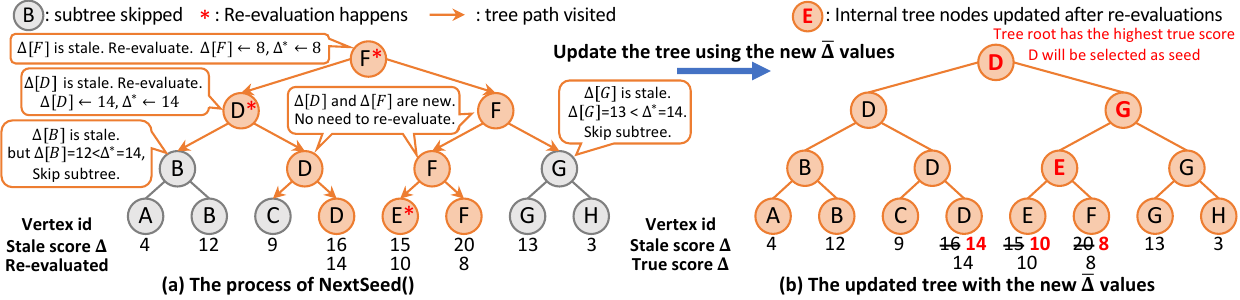}
    \caption{
        \small
        \textbf{Seed selection based on \ourtree{}}. Each leaf stores a vertex id, and each internal node stores the vertex in its subtree with the highest (stale) score. (a) An example of finding the maximum score. For illustration purposes, we assume the parallel threads work at the same speed and all tree nodes on the same level are processed in parallel (in reality the threads run asynchronously in fork-join parallelism).
        Therefore, the subtree at $G$ will see $\Delta^*=14$ updated by $D$, and this subtree will be skipped. %By using this property, only a subset of vertices are re-evaluated.
        (b) Updating the internal nodes with the new $\lazygain{\cdot}$ values. Finally, the root of the tree has the highest true score.
        \label{fig:winning_tree}
        %\vspace{-1em}
    }
\end{figure*}

\hide{
Due to different sets of vertices re-evaluated in each round,
at the beginning of round $i$,
CELF and \ptree{} may not see exactly the same stale scores.
We will denote the stale score at the beginning of round $i$ in
CELF as $\lazygainc{i}{\cdot}$ and that for \ptree as $\lazygaint{i}{\cdot}$.
We reorder vertices by decreasing order of $\lazygainc{i}{\cdot}$ as $\vvv{i}{1},\vvv{i}{2},\dots, \vvv{i}{n}$,
and reorder vertices by decreasing order of $\lazygaint{i}{\cdot}$ as $\uuu{i}{1},\uuu{i}{2},\dots, \uuu{i}{n}$.
%We will focus on analyzing a specific seed selection round $i$.
Denote $\deltas{i}$ as the highest true score in round $i$.
We then denote $x_i$ as the rank of $\deltas{i}$ in the sequence $\vvv{i}{1..n}$,
%i.e., $x_i=\arg\max\{l: \vvv{i}{l}>\deltas{i}\}$.
i.e., $\vvv{i}{x_i}$ is the last vertex in $\vvv{i}{1..n}$ with stale score larger than $\deltas{i}$.
Similarly, we define $y_i$, such that $\uuu{i}{y_i}$ is the last vertex in $\uuu{i}{1..n}$ with stale score larger than $\deltas{i}$.
%such that $\lazygainc{i}{\vvv{i}{x_i}}>\deltas{i}$, and
%$\uuu{i}{y_i}$ be the last vertex in $\uuu{i}{1..n}$ such that $\lazygaint{i}{\uuu{i}{y_i}}>\deltas$.
By definition, $F_i$ is exactly the first $x_i$ vertices in sequence $\vvv{i}{}$, and thus $x_i=|F_i|$.
$F'_i$ contains all vertices in $\uuu{i}{}$ and possibly some more in the same batch with $\uuu{i}{y_i}$, so $|F'_i|\le 2y_i$.
In the simple case discussed above where we assume $\lazygainc{i}{\cdot}=\lazygaint{i}{\cdot}$,
we always have $x_i=y_i$.
}
\hide{
We will use amortized analysis to show that $\sum |F'_i|\le 2\cdot\sum|F_i|$.
Intuitively, \ptree{} may evaluate more vertices than CELF,
but such additional re-evaluations, in many cases, will not be wasted,
because they update $\lazygaint{\cdot}$ to newer values, and avoid them to be re-evaluated in later rounds.
Indeed, if $\lazygainc{t} \ge \lazygaint{t}$ for all $t\in V$, then we must have $y_i\le x_i$ (based on how $x_i$ and $y_i$ are computed).
However, due to different re-evaluation patterns, it is still possible that for some $t\in V$, we have $\lazygainc{t} < \lazygaint{t}$.
Such case happens iff. $t$ is evaluated by CELF in a previous round, but not by \ptree{}.
Namely, there exist some round $j$, such that $t\in F_j$ but $t\notin F'_j$.
Note that $F_j=v_{1..x_i}$, and $u_{1..y_i}\subseteq F'_j$.
In this case, we will save a token for such vertex $t$ when it is evaluated by CELF in round $j$,
such that when later $t$ is counted in $u_{1..y_i}$, we will use the token to count $t$ for free.
Therefore, we partition $u_{1..y_i}$ into two categories: $U_1=F'_i\cap F_i$, and $U_2=F'_i\setminus F_i$, i.e., those also evaluated by CELF or not.
%Based on the same argument, for any vertices $w\in U_2$, it is always because a previous round in CELF evaluates $w$
All vertices in $U_1$ are also in $F_i$, and all vertices in $U_2$ can be counted by the saved tokens. Therefore, by amortized analysis, we have $\sum y_i \le \sum x_i$.

Recall that $x_i=|F_i|$, and $|F'_i|\le 2y_i$. Based on the argument above, we proved that $\sum |F'_i|\le 2\sum|F_i|$.
}

\hide{
Let $F_i$ be the set of re-evaluated vertices by CELF when selecting the $i$-th seed, and $F'_i$ be that of \ptree{s} using \cref{alg:BST}.
We will first show a simple case---if both CELF and \cref{alg:BST} start with the same stale values $\lazygain{\cdot}$, then $|F'_i|\le 2|F_i|$.
%then CELF will re-evaluates $|F_i|$ vertices, and the number of re-evaluated vertices in \cref{alg:BST}, denoted as $|F'_i|$, is no more than $2|F_i|$.
The worst case for \cref{alg:BST} is that the last vertex in $F_i$ is the element in a single batch, say with size $2^j$.
In this case,  $|F_i|=2^j$ and $|F'_i|=1+2+\cdots+2^j=2^{j+1}-1<2|F_i|$.

We now consider the general case for all seed selection rounds.
Unfortunately, the above condition no longer holds.
Consider an example that CELF re-evaluates 8 vertices in a round with $\lazygain{s_1^*}=10$, and \cref{alg:BST} re-evaluates 15 vertices with the other 7 vertices with $\lazygain{\cdot}=8$.
Then the next round, CELF re-evaluates these 7 vertices, resets them to $\lazygain{\cdot}=6$, and finds $s_2^*$ with $\lazygain{s_2^*}=9$; \cref{alg:BST} directly finds $s_2^*$.
In the third round, CELF directly finds $s_3^*$ with $\lazygain{s_3^*}=7$.
However, \cref{alg:BST} needs to re-evaluate these 7 vertices again before it finds $s_3^*$, which is more than twice of what CELF re-evaluates.

By taking a closer look, the reason is that \cref{alg:BST} re-evaluates much fewer vertices in the second round.
Hence, here we will use an amortized analysis to show that $\sum_{i=1}^{k}|F'_i|<2\cdot\sum_{i=1}^{k}|F_i|$.
In particular, in every round, we give a ``token'' to the vertices in $F_i\setminus F'_i$, namely, the vertices re-evaluated by CELF but not \cref{alg:BST}.
When these vertices are re-evaluated by \cref{alg:BST}, we charge the cost using the token these vertices have.

The reason that we give tokens to these vertices since those are the only vertices that can have $\lazygain{\cdot}$ values in CELF smaller than in \cref{alg:BST}.
Given this fact, we can use the same argument at the beginning of this proof.
For $F_i$ for CELF when selecting the $i$-th seed, \cref{alg:BST} re-evaluates a subset of them (some may be re-evaluated in previous rounds and get smaller $\lazygain{\cdot}$ values).
Alternatively, \cref{alg:BST} may re-evaluate vertices with larger $\lazygain{\cdot}$ values than that in CELF, but they have tokens that charge to previous rounds.
Hence, in total, $|F'_i|$ is no more than twice $|F_i|$ plus the tokens consumed in this round.
Since the overall number of tokens is non-negative, we have $\sum_{i=1}^{k}|F'_i|<2\cdot\sum_{i=1}^{k}|F_i|$.
}

\hide{
\begin{lemma}[Work and Span for one seed selection]\label{lemma:work_span_one}
The BST method has $O((W_\Delta + \log n) \cdot |F_i| )$ work and $O((D_{\Delta}+\log n)\cdot \log |F_i| )$ span to select one seed, where $|F_i|$ is the number vertices will be evaluated by CELF algorithm, $W_{\Delta}$ and $D_{\Delta}$ are the work and span to re-evaluate the marginal gain of a vertex.
\end{lemma}
\begin{proof}
According to the PAM paper~\cite{sun2018pam}, \textsc{SplitAndRemove} and \textsc{Max} have $O(\log n)$ work and $O(\log n)$ span, where $n$ is the tree size. \textsc{BatchInsert} has $O(m\log m +m\log n)$ work and $O(\log m \log n)$ span, where $m$ is the size of unordered insertion batch, $n$ is the sum of tree size and batch size.
% Assuming we run $j^*$ rounds to select a seed, the total number of vertices evaluated is $2^{j^{*}}-1$.
% From \cref{alg:BST}, in the $j$-th round, the tree size before split is $n-2^{j-1}+1$, the batch size is $2^{j-1}$.  \textsc{SplitAndRemove} in \cref{line:split} has $O(\log (n-2^{j-1}+1))$ work and span. Re-evaluating marginal gains of the batch in \cref{line:BST_evalue_begin}-\cref{line:BST_evalue_end} has $O(2^{j-1}W_{\Delta})$ work and $(j-1)D_{\Delta}$ span. Finding the maximum marginal value in the batch is bounded by the re-evaluating step. Finding the max value in the tree (\cref{line:BST_stop}) has $O(\log (n-2^j +1))$ work and span. The total work and span for the previous cost in $j^*$ rounds is $O(W_{\Delta} 2^{j^*}
%  + \sum_{j=1}^{j^*} \log(n-2^{j-1}))$ and $O((j^{*})^2 D_{\Delta} + \sum_{j=1}^{j^*} \log(n-2^{j-1}+1))$. After that, we insert all the removed vertices except the seed into the tree, where the number of vertices inserted is $2^{j^*}$ and the tree size is $n-2^{j^*}+1$, the work and span depending on which size is greater. If the tree size is larger, then the work and span for batch insertion are $O(2^{j^*} (j^* + \log (n/2^{j^*})))$ and $O(j^* \log (n-2^{j^*}))$. If the size of inserted vertices is larger, then the work and span are $O(n j^{*} - (n-2^{j^*})\log(n-2^{j^*}))$ and $O(j^* \log (n-2^{j^*}))$.
According to the conclusion from \cref{lemma:efficiency}, the number of vertices we evaluated in total is $O(|F_i|)$, and the number of prefix-doubling rounds is $O(\log |F_i|)$. Note that the tree size in \cref{alg:BST} is different in each round, for simplicity, we all use $O(n)$ to represent tree size.
From \cref{alg:BST}, in the $j$-th round,  the batch size is $B_j$.  \textsc{SplitAndRemove} in \cref{line:split} has $O(\log n)$ work and span. Re-evaluating marginal gains of the batch in \cref{line:BST_evalue_begin}-\cref{line:BST_evalue_end} has $O(B_j \cdot W_{\Delta})$ work and $O(\log B_j + D_{\Delta})$ span. Finding the maximum marginal value in a batch \cref{line:BST_max} is bounded by the re-evaluating step. Finding the max value in the tree (\cref{line:BST_stop}) has $O(\log n)$ work and span.
Note that $\sum_{j} B_j = O(F_i)$, $\sum_j{\log B_j}=O(\log |F_i|)$.
The total work and span for the previous cost in $O(\log|F_i|)$ rounds is $O(W_{\Delta} |F_i| + \log|F_i| \log n)$ and $O(D_{\Delta} \cdot \log |F_i|  + \log |F_i| \log n)$.
After that, we insert all the removed vertices except the seed into the tree, where the number of vertices inserted is $O(|F_i|)$ and the sum of the tree size and insertion batch size is $O(n)$, the work and span for batch insertion in \cref{line:batch_insert} are $O(|F_i|\log |F_i|\log n)$ and $O(\log |F_i|\log n)$.
Therefore, the total work and span of finding the $i$-th seed is $O((W_\Delta + \log n) \cdot |F_i| )$ and $O((D_{\Delta}+\log n)\cdot \log |F_i| )$.
\end{proof}

\begin{theorem}
\letong{replace WΔW_{\Delta}, DΔD_{\Delta}}
Given the same sketches \sketch1..R\sketch{1..R}, the BST method will select the same seed set as CELF with O((WΔ+logn)∑Ki=1(|Fi|))O((W_{\Delta}+\log n)\sum_{i=1}^{K}(|F_i|)) work and \\
O((DΔ+logn)∑Ki=1log|Fi|)O((D_{\Delta}+\log n)\sum_{i=1}^K \log |F_i|) span, where |Fi||F_i| is the number vertices will be evaluated by CELF algorithm to select seed ii, WΔW_{\Delta} and DΔD_{\Delta} are the work and span to re-evaluate the marginal gain of a vertex.
\end{theorem}
\begin{proof}
Firstly, we want to show that the seed set {s1,…,sK}\{s_1, \dots, s_K\} selected by \BST method is the same as CELF. We will prove it by induction.
% By \cref{lemma:connectness}, given the same selected seed sets, \BST method will select the vertex having the maximum marginal gain, which is the same as what CELF selects. Then we will prove the seed sets passed to \BST and CELF are the same in round 1≤i≤K1\leq i\leq K by induction.
In the first round, given the empty seed set SS,  CELF and \BST will both select the vertex with the largest marginal gain. Assume until round i−1i-1, CELF and \BST select the same seeds {s1,…,si−1}\{s_1, \dots, s_{i-1}\}. In round ii, given the same seed set {s1,…,si−1}\{s_1, \dots, s_{i-1}\}, by \cref{lemma:connectness}, they will select the same seed s∗s^*, so S←{s1,…,si−1}∪{s∗}S\gets \{s_1,\dots,s_{i-1}\} \cup \{s^*\} are still the same. Therefore, after the KK-th round (K≥1K\geq1), the seeds selected by CELF and \BST are the same.

Then, we want to show the work and span of \BST method for selecting KK seeds. By \cref{lemma:work_span_one}, selecting the ii-th seed has O((WΔ+logn)⋅|Fi|)O((W_\Delta + \log n) \cdot |F_i| ) work and O((DΔ+logn)⋅log|Fi|)O((D_{\Delta}+\log n)\cdot \log |F_i| ) span. Summing up the work and span in each round from i=1i=1 to KK, we can get the total work is O((WΔ+logn)∑Ki=1(|Fi|))O((W_{\Delta}+\log n)\sum_{i=1}^{K}(|F_i|)) and the total span is O((DΔ+logn)∑Ki=1log|Fi|)O((D_{\Delta}+\log n)\sum_{i=1}^K \log |F_i|).
\end{proof}
}

\begin{theorem}[\ptree{} Cost Bound]
\label{thm:ptreecost}
Given the same sketches $\sketch{1..R}$, our seed selection based on \ptree{} will select the same seed set as CELF with
$O(n\log n+W_{CELF})$ work and $\tilde{O}(kD_{\Delta})$ span, where k is the number of seeds,
$W_{CELF}$ is the work (time complexity) by CELF,
and $D_{\Delta}$ is the span to evaluate one vertex.
\end{theorem}

\ifconference{
\cref{thm:ptreecost} can be proved by \cref{lemma:efficiency} and the cost bound of \ptree{s}. For the page limit, we present the proof in the full version.}

\iffullversion{
\begin{proof}
Using the analysis of \ptree{}~\cite{blelloch2016just}, \textsc{SplitAndRemove} have $O(\log n)$ work and span, where $n$ is the tree size. \textsc{BatchInsert} has $O(n'\log n)$ work and $O(\log n' \log n)$ span, where $n'$ is the size of unordered insertion batch.
In addition, constructing the \ptree{} uses $O(n\log n)$ work and $O(\log n)$ span.

For each seed selection round that evaluates a set of vertices $F'_i$,
\cref{alg:BST} will first split them out from the tree, re-evaluate them, and insert all but one back to tree.
Based on the aforementioned cost bounds, the work per element in $F'_i$ is $O(\log n)$, the same as in CELF per element in $F_i$.
\cref{lemma:efficiency} indicates that $\sum|F'_i|=O(\sum |F_i|)$, so the total work for P-tree is asymptotically the same as the sequential CELF except for the $O(n\log n)$ preprocsssing cost.

\hide{
\cref{alg:BST} mainly comes from 1) constructing the tree with $O(n\log n)$ work and polylogarithmic span, 2) $\log |F'_i|=O(\log n)$ splits and output the elements to an array, in $O(F'_i+\log^2 n)$ total work and polylogarithmic span,
3) $O(|F'_i|)$ re-evaluations, and 4) inserting a batch back to the tree with size $|F'_i|-1$, costing $O(|F'_i|\log n)$ work and polylogarithmic span.
Considering in CELF, the work involves 1) at least $O(n)$ work to construct the priority queue, 2)
$O(|F_i|)$ re-evaluations and 3) at least $O(1)$ cost per evaluation in operating (popping and re-inserting) the priority queue.
From \cref{lemma:efficiency}, we know $\sum|F'_i|=O(\sum |F_i|)$, therefore, the overhead of \ptree{} compared to CELF is at most a polylogarithmic factor, which proves the $\tilde{O}(W_{CELF})$ bound.
}
For the span, note that \ptree{} needs $O(\log n)$ batches per round.
Each batch requires a split, and evaluates multiple vertices in parallel. Therefore, the span is $O(D_{\Delta}+\polylog(n))$ in each round.
After all the batches are processed, finally a batch-insertion adds all but one vertices back to the tree with polylogarithmic span.
Combine all pieces together, the span is $\tilde{O}(D_{\Delta})$ per round, and $\tilde{O}(kD_{\Delta})$ for the entire \cref{alg:BST}.
\end{proof}
}

% \myparagraph{\ourtree}
\subsection{Parallel Priority Queue: \WinTree} \label{sec:WinTree}

%Although \BST method has theoretical guarantees, we are concerned about its performance in practice.
While \ptree{} provides theoretical efficiency for seed selection,
\hide{it maintains the total order of all vertices, which is not needed in many implementations for priority queues (e.g., a binary heap) and may cause performance overhead.}
it maintains the total order of all vertices, which is not needed in priority queues and may cause performance overhead.
Also, \ptree{} explicitly maintains the parent-child pointers, which causes additional space usage.
We now propose a more practical data structure based on a winning tree that overcomes these two challenges, although it does not have the same bounds as in \cref{thm:ptreecost}.

A classic \emph{winning tree} (aka. tournament tree) is a complete binary tree with $n$ leaf nodes and $n-1$ interior nodes.
The data are stored in the leaves. Each interior node records the larger key of its two children.
Since a winning tree is a complete binary tree, it can be stored \emph{implicitly} in an array $T[1..2n-1]$, which consumes smaller space.
%This generally makes winning trees space-efficient and more I/O friendly than the pointer-based BST structures.
In our case, each tree node stores a vertex id (noted as $t.id$).
%The left/right child and the parent of a node can be accessed by their inferred index in a complete binary tree.
The key of a node $t$ is the (stale) score of this vertex, i.e., $\lazygain{t.id}$.
Each interior node stores the id of its children with a larger score.
As a result, the vertex at the root has the highest (stale) score.
%Constructing such a tree with given initial values simply takes linear work and $O(\log n)$ span using divide-and-conquer: recursively constructing the subtrees of children in parallel, and updating the id as the child with a larger key.
%A complete binary tree can be stored in an array with size $2n$ and access the parent to children of a node by their inferred index. Compared to a \BST, a tournament tree is more cache-friendly and memory-efficient. For brevity, in the following description, re-evaluating a node means re-evaluating the vertex stored in the node.

%\SetSideCommentLeft
\hide{
\begin{algorithm}[t]
\caption{Struct of WinningTree\label{alg:winning tree}}
\SetKwProg{myfunc}{Function}{}{}
\SetKwFor{parForEach}{ParallelForEach}{do}{endfor}
\SetKwFor{mystruct}{Struct}{}{}
\SetKwFor{pardo}{ParallelDo}{}{}
\mystruct{WinningTree}{
  Array tree$[]$ \tcp{index of tree nodes}
  Array last$[]$ \tcp{the last time compute the value}
  $\sigma_{\max}$ \tcp{the max marginal value}
}

\myfunc{\textsc{Construct}$(\sigma)$}{
    parallel initialize $latest[v]$ to $0$ for $v \in V$\\
    \textsc{Construct\_rec}$(\sigma, 0, |V| )$
}

\myfunc{\textsc{Construct\_rec}$(\sigma, s, t)$}{
\tcp{$\sigma$ is the marginal value array, $s,t$ are start and end index}
    \If {$s +1 = t$}{
        \Return $s$\\
    }
    $m \gets (s+t)/2$\\
    \pardo{}{
        $L \gets \textsc{Construct\_rec}(\sigma, s, m)$\\
        $R \gets \textsc{Construct\_rec}(\sigma, m+1, t)$
    }
    $\text{tree}[m]\gets \arg\max(\sigma[L], \sigma[R])$\\
    \Return $\text{tree}[m]$
}

\myfunc{\textsc{UpdatePriority}$(\sigma, \Phi[1\dots R], G, k)$}{
\tcp{$k$ is the number of seeds in $S$ now}
    $\sigma_{max}\gets 0$\\
    \textsc{Update\_rec}$(\sigma, 0, n, \Phi[1\dots R], G, k)$
}

\myfunc{\textsc{Update\_rec}$(\sigma, s,t, \Phi[1\dots n], G, k$)}{
    \If{$s + 1 \gets t$}{
        \If {$\sigma[s]>\sigma_{max}$ and $\text{last}[s] \neq k$}{
            $\sigma[s] \gets \textsc{MarginalGain}(S,s,\Phi[1\dots R])$\\
            $\text{last}[s]\gets k$\\
            \If {$\sigma[s] > \sigma_{\max}$}{
                \textsc{Atomic\_Write\_Max}$(\&\sigma_{\max}, \sigma[s])$
            }
            \Return $s$
        }
    }
    $m \gets (s+t)/2$\\
    $root \gets tree[m]$\\
    \If {$last[root] \neq k$ and $\sigma[root]<\sigma_{\max}$}{
        \Return $root$
    }
    \If {$last[root] \neq k$ and $\sigma[root] > \sigma_{\max}$}{
        $\sigma[root] \gets \textsc{MarginalGain}(S, root, \Phi[1\dots R])$\\
        $last[root] \gets k$\\
        \If {$\sigma[root] > \sigma_{\max}$}{
            \textsc{Atomic\_Write\_Max}$(\&\sigma_{\max},\sigma[root])$
        }
    }
    \pardo{}{
        $L \gets \textsc{Update\_rec}(\sigma, s, mid, \Phi[1\dots n], G,k)$\\
        $R \gets \textsc{Update\_rec}(\sigma, mid+1, t, \Phi[1\dots n], G,k)$
    }
    $\text{tree}[m] \gets \arg\max(\sigma[L], \sigma[R])$\\
    \Return $\text{tree}[m]$
}
\end{algorithm}
}

\begin{algorithm}[t]
\small
\caption{Seed Selection based on \wintree{} \label{alg:winning tree}}
\SetKwProg{myfunc}{Function}{}{}
\SetKwFor{parForEach}{ParallelForEach}{do}{endfor}
\SetKwFor{mystruct}{Struct}{}{}
\SetKwFor{pardo}{In Parallel:}{}{}
\SetKwInput{Maintains}{Maintains}
\Maintains{Global variable $\globalmax$: the highest true score evaluated so far\\
\pushline
A winning tree $T$ with $n$ leaf nodes each storing a record.  \\
For a tree node $t\in T$, we use the following notations:\\
$t.id$: the id of the vertex stored in this node\\
$t.\parent$ / $t.\lleft$ / $t.\rright$ : the parent/left child/right child of node $t$\\
%$t.\parent$: the parent of node $t$\\
%$t.\renew$: a flag indicating that this vertex has been re-evaluated.
The \wintree{} is a max-priority-queue based on (stale) score $\lazygain{t.id}$ for each vertex.\\
}
\DontPrintSemicolon
%\tcp{$\sigma_{\max}$ store the global maximum influence found so far}
%Shared variable $\sigma_{\max}$ \\

\myfunc{\upshape\textsc{FindMax}(tree node $t$, current seed set $S$, sketches $\sketch{1..R}$)}{
    \lIf(\tcp*[f]{evaluated by parent} \label{line:stale_false}){$t.id=t.\parent.id$ }{$\stale\gets \false$}
    \lElse{$\stale\gets \true$ \label{line:stale_true}}
    \tcp{Skip a subtree if the max is stale and is smaller than the current best score; no re-evaluation needed for the entire subtree}
    \lIf{$\stale=\true$ and $\lazygain{t.id} < \globalmax$ \label{line:WinTree_skip}}{\Return}
    \If(\tcp*[f]{Current value is stale, re-evaluation needed}){$\stale=\true$}{
        $\lazygain{t.id} \gets$ \textsc{MarginalGain}$(S, t.id, \sketch{1..R})$\tcp*[f]{Re-evaluate} \label{line:WinTree_re-evaluate}\\
        \textsc{WriteMax}($\globalmax{}, \lazygain{t.id}$)\tcp*[f]{update the best score so far}\label{line:wintree-writemax}
    }
    \lIf{$t$ is a leaf}{\Return}
    \pardo{\label{line:divide_begin}}{
        \textsc{FindMax}($t.\lleft$, $S$, $\sketch{1..R}$) \\
        \textsc{FindMax}($t.\rright$, $S$, $\sketch{1..R}$)
        \label{line:divide_end}
    }
    \tcp{compare two branches and reset max}
    \lIf {$\lazygain{t.\lleft.id}> \lazygain{t.\rright.id}$ \label{line:update_id_begin}}{
        $t.id \gets t.\lleft.id$
    }\lElse{
        $t.id \gets t.\rright.id$
        \label{line:update_id_end}
    }
}
\myfunc{\upshape\textsc{NextSeed}($S, \sketch{1..R}$)}{
    $\globalmax \gets 0$\\ %\tcp*[f]{The best marginal gain evaluated so far} \label{line:WinTree_global_max}
    \textsc{FindMax}$(T.\rt, S, \sketch{1..R})$
    \label{line:call_findmax}\\
    % \textsc{Sync}($T.\rt$)\\
    %$s^*\gets T.\rt.id$ \label{line:seed_root}\\
    % Remove $T.\rt$ from the tree\\
    %\Return{$s^*$}
    \Return{$T.\rt.id$}
}

\end{algorithm}

\hide{
\begin{algorithm}[ht]
\small
\caption{The Tournament-tree based \textsc{SuM-PQ} \label{alg:winning tree}}
\SetKwProg{myfunc}{Function}{}{}
\SetKwFor{parForEach}{ParallelForEach}{do}{endfor}
\SetKwFor{mystruct}{Struct}{}{}
\SetKwFor{pardo}{In Parallel:}{}{}
\SetKwInput{Maintains}{Maintains}
\Maintains{\\
\pushline
A tournament tree $T$ with $n$ leaf nodes each storing a record.  \\
For a tree node $t\in T$, we use the following notations:\\
$t.id$: the id of the vertex stored in this node\\
$t.\lleft$ / $t.\rright$: the left/right child of node $t$\\
$t.\parent$: the parent of node $t$\\
$t.\renew$: a flag indicating that this vertex has been re-evaluated.
The tournament tree is a max-priority-queue based on the (lazily-evaluated) marginal gain for each vertex,
i.e., $\lazygain{t.id}$.
}
\DontPrintSemicolon
%\tcp{$\sigma_{\max}$ store the global maximum influence found so far}
%Shared variable $\sigma_{\max}$ \\
\myfunc{\upshape\textsc{Sync}(tree node $t$)}{
    \lIf(\tcp*[f]{skip subtree if not renewed}){$t.\renew=\false$}{\Return{}}
    $t.\renew \gets \false$\tcp*[f]{set the flag back to $\false$}\\
    \lIf{$t$ is a leaf}{\Return{}}
    \pardo{}{
        $\textsc{Sync}(t.\lleft)$\\
        $\textsc{Sync}(t.\rright)$
    }
    \lIf {$\lazygain{t.\lleft.id}> \lazygain{t.\rright.id}$}{
        $t.id \gets t.\lleft.id$
    }\lElse{
        $t.id \gets t.\rright.id$ \tcp*[f]{compare two branches and reset max}
    }
}

\myfunc{\upshape\textsc{FindMax}(a pointer to a shared variable $\globalmax$, tree node $t$, current seed set $S$, sketches $\sketch{1..R}$)}{
    \lIf{$t.id=t.\parent.id$ }{$t.\renew\gets \true$}
    \tcp{Skip a subtree if the max (even if it can be stale) is smaller than the current best score}
    \lIf{$t.\renew = \false$ and $\lazygain{t.id} \leq \globalmax$}{\Return}
    \If(\tcp*[f]{Current value is stale}){$t.\renew=\false$}{
        $\lazygain{t.id} \gets$ \textsc{MarginalGain}$(S, t.id, \sketch{1..R})$\tcp*[f]{Re-evaluate}\\
        $t.\renew\gets \true$\\
        \textsc{WriteMax}($\globalmax{}, \lazygain{t.id}$)
    }
    \lIf{$t$ is a leaf}{\Return}
    \pardo{}{
        \textsc{FindMax}($\globalmax$, $t.\lleft$, $S$, $\sketch{1..R}$)\\
        \textsc{FindMax}($\globalmax$, $t.\rright$, $S$, $\sketch{1..R}$)
    }
}
\myfunc{\upshape\textsc{NextSeed}($S, \sketch{1..R})$)}{
    $\globalmax \gets 0$\tcp*[f]{The best marginal gain evaluated so far}\\
    \textsc{FindMax}$(\&\globalmax, T.\rt)$\\
    \textsc{Sync}($T.\rt$)\\
    $s^*\gets T.\rt.id$\\
    Remove $T.\rt$ from the tree\\
    \Return{$s^*$}
}
\end{algorithm}

\begin{algorithm}[ht]
\small
\caption{The Tournament-tree based \textsc{SuM-PQ} \label{alg:winning tree}}
\SetKwProg{myfunc}{Function}{}{}
\SetKwFor{parForEach}{ParallelForEach}{do}{endfor}
\SetKwFor{mystruct}{Struct}{}{}
\SetKwFor{pardo}{In Parallel:}{}{}
\SetKwInOut{Maintains}{Maintains}
\Maintains{A tournament tree $T$ with $n$ leaf nodes each corresponding to a record.
}
%\tcp{$\sigma_{\max}$ store the global maximum influence found so far}
%Shared variable $\sigma_{\max}$ \\
\myfunc{\textsc{Mark}(node $t$)}{
    \While{$t$ is not a leaf node}{
        $t.renew \gets 1$\\
        \lIf {$t.id = t.left.id$}{$t\gets t.left$}
        \lElse{$t\gets t.right$}
    }
    $t.renew \gets 1$
}

\myfunc{\textsc{Sync}(node $t$)}{
    \lIf{$t.renew=0$}{\Return{}}
    $t.renew \gets 0$\\
    \lIf{$t$ is a leaf}{\Return{}}
    \pardo{}{
        $\textsc{Sync}(t.\lleft)$\\
        $\textsc{Sync}(t.\rright)$
    }
    \lIf {$\sigma[t.\lleft.id]> \sigma[t.\rright.id]$}{
        $t.id \gets t.\lleft.id$
    }\lElse{
        $t.id \gets t.\rright.id$
    }
}

\myfunc{\upshape\textsc{FindMax}(a pointer to a shared variable $\sigma_{\max}$, tree node $t$)}{
    \lIf {$t.renew = 0$ and $\sigma[t.id] \leq \sigma_{\max}$}{\Return}
    \If{$t.renew=0$}{
        $\sigma[t.id] \gets \lazygain(t.id)$\\
        \textsc{Mark}($t$)\\
        \lIf {$\sigma[t.id] > \sigma_{\max}$}{\textsc{WriteMax}($\sigma_{\max}, \sigma[t.id]$)}
    }
    \lIf{$t$ is a leaf}{\Return}
    \pardo{}{
        \textsc{FindMax}($\sigma_{\max}$, $t.\lleft$)\\
        \textsc{FindMax}($\sigma_{\max}$, $t.\rright$)
    }
}
\myfunc{\textsc{NextSeed}($S, \sketch{1..R})$)}{
    $\sigma_{\max} \gets 0$\\
    \textsc{FindMax}$(\&\sigma_{\max}, T.root)$\\
    \textsc{Sync}($T.root$)
    \Return{$T.root$}
}
\end{algorithm}
} 

To support CELF efficiently, we use the internal nodes of \wintree{} to prune the search process.
%we will take advantage of the information maintained in the internal nodes.
Suppose the best true score evaluated so far is $\deltas$,
then if we see a subtree root $t$ with a stale score smaller than $\deltas$, we can skip the \emph{entire subtree}.
This is because all nodes in this subtree must have smaller stale scores than $t.id$,
which indicates even smaller true scores.
Although this idea is simple, we must also carefully maintain the \wintree{} structure, with the newly evaluated true scores.
We presented our algorithm in \cref{alg:winning tree} with an illustration in \cref{fig:winning_tree},
and elaborate on more details below.

\hide{
\myparagraph{CELF Optimization in \WinTree} \WinTree borrows the high-level idea of CELF when selecting the next seed: taking advantage of submodular property to skip re-evaluating vertices that can not be the next seed.  If a node is not updated and its marginal gain is smaller than the largest re-evaluated marginal gain so far, this node is impossible to be the next seed because, according to the sub-modular property, re-evaluation will not increase its marginal gain.
% If a node is not re-evaluated and its old marginal gain is smaller than the largest new value so far, by the CELF optimization, it will be skipped.
Besides, the root of a \WinTree stores the node with the maximum key in the tree, so all the nodes in the subtree rooted at this node are no larger than it and can also be skipped. In summary, we traverse the \WinTree from the top in parallel, whenever we visit a node that has not been re-evaluated and its marginal gain is smaller than the largest re-evaluated marginal gain so far, we skip traversing the whole subtree rooted at this node, and therefore skip re-evaluating the vertices stored in this subtree.
}

% \myparagraph{Maintaining a Valid \WinTree}
% If a node is skipped, then the subtree rooted at it is  preserved as it was and is a valid subtree.
% % and the subtree is a valid \WinTree ordered by their old values.
% Otherwise, if an interior node is re-evaluated, the child having the same id as it may,decrease its marginal
% then the new value is not larger than before because of the sub-modular property, so the node can not guarantee to store a larger one of its children, thus we need to update the node recursively to make the subtree to be valid.

% \hide{
% According to the previous discussion, it is important to know whether a node is re-evaluated or not. We say that a node is not re-evaluated if its id is different from that of its parent. Half of the vertices are only stored in leaf nodes. All the vertices are stored in at least one node in the tree. The tree nodes storing the same vertices are on a tree path to a leaf node. If we re-evaluate this vertex, we will re-evaluate it when we traverse the shallowest node on the path. If the shallowest node is not re-evaluated, then the whole subtree will be skipped. Therefore, if the id of a node is the same as that of its parent, then the node is not the shallowest node on the tree path of this vertex and it has been re-evaluated; if the id of a node is not the same as that of its parent or it does not have a parent (the root), the node is the shallowest node on the tree path of this vertex and has not been re-evaluated.
% }

%The pseudo-code of \WinTree is shown in \cref{alg:winning tree}.
For $\textsc{NextSeed}$, we keep a global variable $\globalmax$ as the largest true score obtained so far, initialized to 0.
The algorithm calls the \textsc{FindMax}$(t, \dots)$ routine starting from the root, which explores the subtree rooted at $t$.
%and calls the recursive  function $\textsc{NextSeed}(..., t, ...)$ (see \cref{line:call_findmax}), which
%will traverse the tree rooted at node $t$ and find the largest re-evaluated marginal gains in the tree with CELF optimization. The global variable $best$ will be updated if the algorithm finds a larger new marginal value.
%There are three cases.
We first check if the score at $t$ is stale:
if $t$'s id is the same as its parent, the true score has been re-evaluated at its parent (\cref{line:stale_false}). %\yan{add}
%this vertex's true score has been re-evaluated
Based on the node's status, %there are three cases.
we discuss three cases.
First, if the score is stale and is lower than $\globalmax$,
as discussed above, we can skip the entire subtree and terminate the function (\cref{line:WinTree_skip}).
Second, if the score is stale but is higher than $\globalmax$, we have to re-evaluate the vertex $t.id$,
since it may be a candidate for the seed (\cref{line:WinTree_re-evaluate}).
%If the true score obtained is higher than $\globalmax$, we will update $\globalmax$ by this better true score using atomic operation \writemax{}.
We then use the atomic operation \writemax{} to update $\globalmax$ by this true score if it is better.
The third case is when the score is not stale.
Although no evaluation is needed on $t.id$, we should still explore the subtrees, since with the newly evaluated score, the subtree structure (i.e., the ids of the internal nodes) may change.
Hence, in cases 2 and 3, we recursively explore the two subtrees in parallel (Lines \ref{line:divide_begin} to \ref{line:divide_end}).
After the recursive calls, we set the vertex id as one of its two children with a higher score (Lines \ref{line:update_id_begin} to \ref{line:update_id_end}).

\revision{
\begin{theorem}[\wintree{} Correctness]\label{lemma:wintree_connectness}
\cref{alg:winning tree} always selects the next seed with the largest marginal gain. %, i.e., $\gain{T.root.id}=\max_{v\in V}\{\gain{v}\}$.
\end{theorem}

\iffullversion{
\begin{proof}
\hide{
  Let $\globalmax$ be the largest true score obtained by the algorithm after \cref{line:call_findmax} in \cref{alg:winning tree}.
  We first show that tree $T$ after \cref{line:call_findmax} is a valid \wintree{}.
  Based on the node's status, there are two cases. First,  if the score is stale, then it must be lower than $\globalmax$, otherwise, it will be re-evaluated by the algorithm (\cref{line:WinTree_re-evaluate}). In this case, the algorithm will skip the entire subtree, so the node is still the larger one of its children. Second, if the score is not stale, the algorithm will update its subtree recursively (\cref{line:divide_begin}-\cref{line:divide_end}) and update the node as the larger one of its children (\cref{line:update_id_begin}-\ref{line:update_id_end}). In both cases, the node is the larger one of its children, so the tree $T$ is a valid \wintree{}.

  Then we will show that $\gain{T.root.id}=\max_{v\in V}\{\gain{v}\}$. $T.root$ is not stale because when calling \textsc{FindMax} at \cref{line:call_findmax}, $\globalmax$ is 0, so $T.root$ must be re-evaluated. Then by the property of $\wintree{}$, the root stores the largest score among all the nodes. For any  stale node $v$, $\truegain{T.root}>\lazygain{v} \geq \truegain{v}$. For any not stale node $v$, $\truegain{T.root}>\truegain{v}$. Therefore, $T.root$ has the largest $\truegain$ among all the vertices.
  }
  \hide{First, note that the \wintree{} is always valid after each round, where each internal node records the vertex with the highest (stale) score in its children. This is because any re-evaluation at an internal node must result in a comparison later to ensure the stored value is the larger one in its children (\cref{line:update_id_begin,line:update_id_end}). }

  Due to \cref{line:update_id_begin,line:update_id_end}, the \wintree{} is always valid after each round, where each internal node records the vertex with the higher (stale) score in its children. We then show that the final selected seed has the highest true score. First, using write-max, the value of $\globalmax$ is non-decreasing during the algorithm. Let $\Delta_m$ be the final highest score selected by the algorithm. This means that throughout the algorithm, $\globalmax \le \Delta_m$. For any vertex $v$, there are two cases.
  % The first is when $v$ has been re-evaluated, so its true score has updated $\Delta_m$ by write-max (\cref{line:wintree-writemax}), and therefore $\gain{v}\le \Delta_m$.
   The first is when $v$ has been re-evaluated, so its true score has updated $\globalmax$ by write-max (\cref{line:wintree-writemax}), and therefore $\gain{v}\le \globalmax \le \Delta_m$.
  % Otherwise, $v$ is skipped by \cref{line:WinTree_skip} at a node $t$, indicating $\Delta_m\ge\globalmax\ge\lazygain{t.id}\ge\lazygain{v}\ge \gain{v}$.
  Otherwise, $v$ is skipped by \cref{line:WinTree_skip} at a node $t$, indicating $\gain{v} \le \lazygain{v} \le\globalmax \le \Delta_m$.
  Hence, $\Delta_m \ge \gain{v}$ for all $v$.

  Note that when $\globalmax$ obtained its final value $\Delta_m$, the corresponding vertex will be carried all the way up to the root (by \cref{line:update_id_begin}-\ref{line:update_id_end}), which proves the stated theorem.
  \end{proof}
}
}

\ifconference{
We present the proof in the full version of this paper.
%The proof sketch is to show that 1) $\globalmax$ is non-decreasing due to \emph{write-max}, and 2) a subtree is only skipped when $\globalmax$ is larger than its stale value.
%the tree $T$ is a valid win-tree after \cref{line:call_findmax} in \cref{alg:winning tree} and the root's marginal gain is not stale.
}Unlike \ptree{s},
we cannot prove strong bounds for the number of re-evaluations in \wintree{}---since the parallel threads are highly asynchronous,
the progress of updating $\globalmax$ and pruning the search cannot be guaranteed.
However, we expect \wintree{} to be more practical than \ptree{} for a few reasons. First,
\wintree{} is a complete binary and can be maintained in an array,
which requires smaller space (no need to store metadata such as pointers in \ptree{s}).
Second, the \ptree{} algorithm requires $O(\log n)$ batches and synchronizing all threads between batches.
Such synchronization may result in scheduling overhead, while the \wintree{} algorithm is highly asynchronous.
Most importantly, \wintree{} does not maintain the total order, and the construction time is $O(n)$ instead of $O(n\log n)$.
In \cref{sec:exp-ana}, we experimentally verify that although \wintree{} incurs more re-evaluations than \ptree{s}, it is faster in most tests.

\hide{
During the traversing, whenever we visit a node whose id is different from its parent, we encounter a new vertex that we have never seen before, and therefore, it is not re-evaluated (or say stale, \cref{line:stale_true}). Otherwise, the vertex in this node is already re-evaluated by its ancestor, so the value is not stale (\cref{line:stale_false}). If the node is stale and the marginal gain is not larger than the largest gain so far, by the CELF optimization, we will skip traversing the subtree (\cref{line:WinTree_skip}). If the node is stale but the marginal gain is greater than the largest gain so far, it is a possible candidate for the seed. We will re-evaluate its marginal gain (\cref{line:WinTree_re-evaluate}) and update the best value so far by \textsc{WriteMax}. \textsc{WriteMax} is an atomic operation that guarantees correctness when the global variable is updated simultaneously by multiple threads.
Then we continue traversing the tree in parallel (\cref{line:divide_begin} to \cref{line:divide_end}). After recursively traversing the sub-trees, the marginal value of the child may be different, we need to update the id to the same as the larger child (\cref{line:update_id_begin} and \cref{line:update_id_end}), which re-ordered the tree to be a valid winner tree.
When $\textsc{FindMax}$ called on the root of the \WinTree returns (\cref{line:call_findmax}), the \WinTree is a valid winner tree, so the seed vertex is stored in the root node. which is the vertex with the largest re-evaluated marginal gain(see \cref{line:seed_root}).
}
\hide{
It is worth noting that, although we use $t.\lleft$, $t.\rright$ and $t.\parent$ in the pseudocode,
they do not need to be stored explicitly in tree nodes:
as mentioned above, \wintree{} can be maintained in an array, and we can compute the index of a node's left/right child and parent in $O(1)$ cost.
}

% In $\textsc{FindMax}(best, t)$, we first judge whether the node $t$ stores a vertex with old value (stale is true) or a vertex has been re-evaluated (stale is false) (see \cref{line:stale_begin}-\cref{line:stale_end}).  If the node is stale and its marginal value is smaller than the $best$, then we can skip searching the subtree rooted at this node (see \cref{line:WinTree_skip}). Otherwise, we re-evaluate the marginal value of the vertex stores in the node (see \cref{line:WinTree_re-evaluate}).

\hide{
Ideally, we want to re-evaluate the vertices in $F_i$ in parallel and select the maximum vertex among them as the $i$-th seed.
% , where $F_i = \{v | \Delta_{i-1}[v] > \Delta_{max}\}$ is the vertices evaluated by CELF.
Note that $F_i$ are the top $|F_i|$ largest vertices in $\Delta_{i-1}[\cdot]$ (recall $F_i = \{v | \Delta_{i-1}[v] > \Delta^*\}$).
Assuming we could know $|F_i|$ ahead, we need a container for vertices that can maintain them by their marginal gains, pop the top $|F_i|$ vertices and push back the vertices in $|F_i|$ back to the container efficiently. A binary search tree (BST) can store items in the order of their keys.
Besides, an existing Parallel Augmented Maps (PAM) uses the underlying balanced binary tree structure using join-based algorithms \cite{blelloch2016just, sun2018pam}.
It supports batch deletion ($\textsc{SplitAndRemove}$ in \cref{alg:BST}) and batch insertion ($\textsc{BatchInsert}$ in \cref{alg:BST}). $\textsc{SplitAndRemove}(k)$ splits the tree at the $k$-th node, returns and removes the subtree of the first $k$ nodes and keeps the remaining nodes. \textsc{SplitAndRemove($|F_i|$)} can implement the function of popping the top $|F_i|$ vertices.
\textsc{BatchInsert}($F_i$) can implement the function of pushing $F_i$ back into the tree in parallel.

% We maintain the marginal value of all the vertices in a parallel BST, which supports $\textsc{SplitAndRemove}$, $\textsc{Max}$ and $\textsc{BatchInsert}$ operations. $\textsc{SplitAndRemove}(i)$ splits out the top $i$ nodes and keep the remaining nodes. $\textsc{Max}$ finds out the maximum value in the tree. $\textsc{BatchInsert}(S)$ inserts a set of vertices $S$ into the tree in parallel.

However, we can not know $|F_i|$ in advance.  To make the number of vertices popped as close as possible to $|F_i|$  and gain parallelism as well, we pop the top $1, 2, 4, 8, ...$ nodes in a prefix-doubling manner until the total number of vertices popped exceeds $|F_i|$ (the condition to stop is that the maximum re-evaluated marginal gain of popped vertices is greater than the maximum old marginal value of not popped vertices). The pseudo-code of this method is shown in \cref{alg:BST}. In round $i$, we pop the top $2^i$ nodes by $\textsc{SplitAndRemove}(2^i)$ (see \cref{line:split}). Then we re-evaluate these nodes and update their marginal value stored in the $\Delta[\cdot]$ (see \cref{line:BST_evalue_begin}-\cref{line:BST_evalue_end}). We keep track of the vertex with the largest new marginal value so far as the temporary seed (see \cref{line:BST_update_seed}). We repeat this process until the new value of the temporary seed is larger than the largest value in the tree (see \cref{line:BST_stop}). When we stop, the temporary seed is the vertex with the largest new marginal value among all the vertices, which means it is the seed we need. Finally, we insert all the vertices that have been popped out back into the tree with their new marginal value by $\textsc{BatchInsert}$ (see \cref{line:batch_insert}).

\begin{lemma}[Correctness]\label{lemma:connectness}
  Given a set of selected seeds $S$, the next seed selected by the \BST method is $\argmax \{\Delta [v~|~S]~|~v \in V\}$.
\end{lemma}
\begin{proof}
The vertices in $V$ can be classified into two sets, $V_T$ is the set containing vertices remaining in the tree after prefix-doubling ($T$ in \cref{line:batch_insert}) and  $V_{\bar{T}}$ is the set containing vertices removed from the tree.
$V_{\bar{T}}$ stores vertices re-evaluated in the current round. Seed vertex $s^*$ keeps track of the vertex having the largest re-evaluated marginal gain, which means $s^* \gets \argmax \{ \Delta[v~|~S]~|~v\in V_{\bar{T}}\}$.
When the prefix-doubling re-evaluation stops, $\Delta[s^*] > \max\{\Delta[v]~|~v\in V_T\}$ (see \cref{line:BST_stop}), where $\Delta[\cdot]$ stores lazily-computed marginal values. Because $s^*$ is re-evaluated, $\Delta[s^*] = \Delta[s^*~|~S]$. Because vertices in $V_T$ are not re-evaluated and the property of submodular, $\Delta[v~|~S] \leq \Delta[v]$ for $v\in V_T$.
$\Delta[s^*] = \Delta[s^*~|~S] > \max \{\Delta[v]~|~v\in V_T\} \geq \max\{\Delta[v~|~S]~|v\in V_T\}$. Therefore, $\Delta[s^*] \geq \max \{\Delta[v]~|~v \in V_T \cup V_{\bar{T}}\}$, which proofs $s^* \gets \argmax \{\Delta[v~|~S]~|v\in V\}$.
\end{proof}

\begin{lemma}[Efficiency]\label{lemma:efficiency}
The number of vertices re-evaluated by the \BST method to select the one seed is at most twice that of CELF.
\end{lemma}

\begin{proof}
Let's first prove that given the same lazy-evaluated marginal gain array $\Delta_{i-1} [\cdot]$, the \BST method re-evaluates at most twice as many vertices as that of CELF.
Remember that the vertices CELF will evaluate to select the $i$-th seed are $F_i = \{ v | \Delta_{i-1} [v] > \Delta^*\}$, where $\Delta_{i-1}[cdot]$ is the lazy-updated marginal value prior to the $i$-th round (the $\Delta[cdot]$ that passed to the $i$-th round $\textsc{NextSeed}$).
$|F_i|$ is the number of vertices CELF re-evaluated, and is also the rank of the seed in $\Delta_{i-1}[\cdot]$.
In our \BST method, we pop the top $1,2,4,8,...$ vertices until the seed is popped. If we use $j$ as the number of prefix doubling rounds we run, the number of vertices we evaluate is $2^j-1$. The actual number of rounds we run to find the seed is the $j^*$ that satisfies $2^{j^* -1} \leq |F_i| \leq 2^{j^*} -1$. The number of our evaluates over that of CELF is $\frac{2^{j^*} -1}{|F_i|} \leq \frac{2^{j^*} -1}{2^{j^* -1}} < 2$.

% Note that given the same lazy-evaluated marginal gain array $\Delta$ at the beginning of a round, because \BST may re-evaluate more vertices than CELF, the $\Delta$ arrays in \BST and CELF are different at the end of the round.
Starting from the second round, $\Delta[\cdot]$ passed to CELF and \BST method are different, because \BST evaluates a superset of $F_1$ in the first round. Then we will show that CELF will evaluate fewer vertices given $\Delta_{\BST}[\cdot]$ than given $\Delta[\cdot]$, which are the lazy-evaluated marginal gain arrays maintained by \BST and CELF respectively before round $i$.

Given that \BST evaluates a superset of vertices evaluated by CELF , by the property of submodular, $\Delta_{\BST}[v] \leq \Delta[v] $, for $v \in V$.
% By \cref{lemma:connectness}, the seeds $\{s_1, s_2, \dots s_{i-1}\}$ selected by \BST are the same as those selected by CELF, so the marginal gain of the $i$-th seed: $\Delta_{\max} \gets \Delta[s_i | \{s_1, \cdots, s_{i-1}\}]$, are the same in \BST and CELF.
Let $F_{\BST, i} = \{v | \Delta_{\BST}[v] > \Delta^* \}$ be the vertices evaluated by CELF given $\Delta_{\BST}[\cdot]$.
The number of vertices in $\Delta_{\BST}[\cdot]$ that are greater than $\Delta^*$ is no larger than that in $\Delta[\cdot]$, which means $|F_{\BST,i}| \leq |F_i|$.
% $F_{\BST, i} \subseteq F_i$.
According to the previous proof, the number of vertices evaluated by \BST is smaller than $2\cdot |F_{\BST, i}|$ , which is also no more than $2\cdot|F_i|$.

\letong{The proof is not very formal, we do not proof $\Delta_{\BST, i-1}[v] \leq \Delta_{\text{CELF}, i-1}[v]$ always hold for all $v\in V$ and all $2 \leq i \leq K$.}

\end{proof}

\begin{lemma}[Work and Span for one seed selection]\label{lemma:work_span_one}
The BST method has $O((W_\Delta + \log n) \cdot |F_i| )$ work and $O((D_{\Delta}+\log n)\cdot \log |F_i| )$ span to select one seed, where $|F_i|$ is the number vertices will be evaluated by CELF algorithm, $W_{\Delta}$ and $D_{\Delta}$ are the work and span to re-evaluate the marginal gain of a vertex.
\end{lemma}
\begin{proof}
According to the PAM paper~\cite{sun2018pam}, \textsc{SplitAndRemove} and \textsc{Max} have $O(\log n)$ work and $O(\log n)$ span, where $n$ is the tree size. \textsc{BatchInsert} has $O(m\log m +m\log n)$ work and $O(\log m \log n)$ span, where $m$ is the size of unordered insertion batch, $n$ is the sum of tree size and batch size.
% Assuming we run $j^*$ rounds to select a seed, the total number of vertices evaluated is $2^{j^{*}}-1$.
% From \cref{alg:BST}, in the $j$-th round, the tree size before split is $n-2^{j-1}+1$, the batch size is $2^{j-1}$.  \textsc{SplitAndRemove} in \cref{line:split} has $O(\log (n-2^{j-1}+1))$ work and span. Re-evaluating marginal gains of the batch in \cref{line:BST_evalue_begin}-\cref{line:BST_evalue_end} has $O(2^{j-1}W_{\Delta})$ work and $(j-1)D_{\Delta}$ span. Finding the maximum marginal value in the batch is bounded by the re-evaluating step. Finding the max value in the tree (\cref{line:BST_stop}) has $O(\log (n-2^j +1))$ work and span. The total work and span for the previous cost in $j^*$ rounds is $O(W_{\Delta} 2^{j^*}
%  + \sum_{j=1}^{j^*} \log(n-2^{j-1}))$ and $O((j^{*})^2 D_{\Delta} + \sum_{j=1}^{j^*} \log(n-2^{j-1}+1))$. After that, we insert all the removed vertices except the seed into the tree, where the number of vertices inserted is $2^{j^*}$ and the tree size is $n-2^{j^*}+1$, the work and span depending on which size is greater. If the tree size is larger, then the work and span for batch insertion are $O(2^{j^*} (j^* + \log (n/2^{j^*})))$ and $O(j^* \log (n-2^{j^*}))$. If the size of inserted vertices is larger, then the work and span are $O(n j^{*} - (n-2^{j^*})\log(n-2^{j^*}))$ and $O(j^* \log (n-2^{j^*}))$.
According to the conclusion from \cref{lemma:efficiency}, the number of vertices we evaluated in total is $O(|F_i|)$, and the number of prefix-doubling rounds is $O(\log |F_i|)$. Note that the tree size in \cref{alg:BST} is different in each round, for simplicity, we all use $O(n)$ to represent tree size.
From \cref{alg:BST}, in the $j$-th round,  the batch size is $B_j$.  \textsc{SplitAndRemove} in \cref{line:split} has $O(\log n)$ work and span. Re-evaluating marginal gains of the batch in \cref{line:BST_evalue_begin}-\cref{line:BST_evalue_end} has $O(B_j \cdot W_{\Delta})$ work and $O(\log B_j + D_{\Delta})$ span. Finding the maximum marginal value in a batch \cref{line:BST_max} is bounded by the re-evaluating step. Finding the max value in the tree (\cref{line:BST_stop}) has $O(\log n)$ work and span.
Note that $\sum_{j} B_j = O(F_i)$, $\sum_j{\log B_j}=O(\log |F_i|)$.
The total work and span for the previous cost in $O(\log|F_i|)$ rounds is $O(W_{\Delta} |F_i| + \log|F_i| \log n)$ and $O(D_{\Delta} \cdot \log |F_i|  + \log |F_i| \log n)$.
After that, we insert all the removed vertices except the seed into the tree, where the number of vertices inserted is $O(|F_i|)$ and the sum of the tree size and insertion batch size is $O(n)$, the work and span for batch insertion in \cref{line:batch_insert} are $O(|F_i|\log |F_i|\log n)$ and $O(\log |F_i|\log n)$.
Therefore, the total work and span of finding the $i$-th seed is $O((W_\Delta + \log n) \cdot |F_i| )$ and $O((D_{\Delta}+\log n)\cdot \log |F_i| )$.
% Let $|B_j|$ be the batch size of a round, and the sum of batch sizes is the number of vertices re-evaluated, which is $O(F_i)$, and $|B_j|=O(|F_i|)$
% From \cref{alg:BST}, in the $j$-th round, the batch size is $2^{j-1}$.  \textsc{SplitAndRemove} in \cref{line:split} has $O(\log n)$ work and span. Re-evaluating marginal gains of the batch in \cref{line:BST_evalue_begin}-\cref{line:BST_evalue_end} has $O(2^{j-1}W_{\Delta})$ work and $O(\log|F_i| D_{\Delta})$ span. Finding the maximum marginal value in the batch (\cref{line:BST_max}) is bounded by the re-evaluating step. Finding the max value in the tree (\cref{line:BST_stop}) has $O(\log n)$ work and span. The total work and span for the previous cost in $O(\log|F_i|)$ rounds is $O(W_{\Delta} |F_i|
%  + \log (|F_i|) \log(n))$ and $O((j^{*})^2 D_{\Delta} + \sum_{j=1}^{j^*} \log(n-2^{j-1}+1))$. After that, we insert all the removed vertices except the seed into the tree, where the number of vertices inserted is $2^{j^*}$ and the tree size is $n-2^{j^*}+1$, the work and span depending on which size is greater. If the tree size is larger, then the work and span for batch insertion are $O(2^{j^*} (j^* + \log (n/2^{j^*})))$ and $O(j^* \log (n-2^{j^*}))$. If the size of inserted vertices is larger, then the work and span are $O(n j^{*} - (n-2^{j^*})\log(n-2^{j^*}))$ and $O(j^* \log (n-2^{j^*}))$.
\end{proof}

\begin{theorem}
\letong{replace $W_{\Delta}$, $D_{\Delta}$}
Given the same sketches $\sketch{1..R}$, the BST method will select the same seed set as CELF with $O((W_{\Delta}+\log n)\sum_{i=1}^{K}(|F_i|))$ work and \\
$O((D_{\Delta}+\log n)\sum_{i=1}^K \log |F_i|)$ span, where $|F_i|$ is the number vertices will be evaluated by CELF algorithm to select seed $i$, $W_{\Delta}$ and $D_{\Delta}$ are the work and span to re-evaluate the marginal gain of a vertex.
\end{theorem}
\begin{proof}
Firstly, we want to show that the seed set $\{s_1, \dots, s_K\}$ selected by \BST method is the same as CELF. We will prove it by induction.
% By \cref{lemma:connectness}, given the same selected seed sets, \BST method will select the vertex having the maximum marginal gain, which is the same as what CELF selects. Then we will prove the seed sets passed to \BST and CELF are the same in round $1\leq i\leq K$ by induction.
In the first round, given the empty seed set $S$,  CELF and \BST will both select the vertex with the largest marginal gain. Assume until round $i-1$, CELF and \BST select the same seeds $\{s_1, \dots, s_{i-1}\}$. In round $i$, given the same seed set $\{s_1, \dots, s_{i-1}\}$, by \cref{lemma:connectness}, they will select the same seed $s^*$, so $S\gets \{s_1,\dots,s_{i-1}\} \cup \{s^*\}$ are still the same. Therefore, after the $K$-th round ($K\geq1$), the seeds selected by CELF and \BST are the same.

Then, we want to show the work and span of \BST method for selecting $K$ seeds. By \cref{lemma:work_span_one}, selecting the $i$-th seed has $O((W_\Delta + \log n) \cdot |F_i| )$ work and $O((D_{\Delta}+\log n)\cdot \log |F_i| )$ span. Summing up the work and span in each round from $i=1$ to $K$, we can get the total work is $O((W_{\Delta}+\log n)\sum_{i=1}^{K}(|F_i|))$ and the total span is $O((D_{\Delta}+\log n)\sum_{i=1}^K \log |F_i|)$.
\end{proof}
}

\hide{
\myparagraph{Discussion about other search trees.}
Our approach in \cref{alg:BST} works with any search tree that supports \textsc{SplitAndRemove} and parallel \textsc{BatchInsert}.
To alleviate the space and I/O issue mentioned above, one may also consider using search trees that group multiple data in one tree node.
In a recent paper that compares \ptree{s} with search tree with blocked leaves, it was observed that blocking multiple data into leaves generally improve the time and space of \ptree{s}, but may be unfriendly to batch-inserting a small batch into the tree, which is important in \cref{alg:BST}. Therefore, we did not consider such data structures in our implementation.
}

\hide{Although BST allows for theoretical guarantees for seed selection,
its performance in practice may be limited by overhead in space and I/O-unfriendliness.
%we are concerned about its performance in practice.
%Our second solution is based on parallel winning trees.
Like other BST structures, to store a data entry, \ptree{} requires to store large metadata, including two child pointers, a subtree size (which is important in \textsc{SplitAndRemove}), and some auxiliary data to aid memory management. In addition, the pointer-based structure is generally I/O-unfriendly.
%Like other pointer-based trees, \ptree{} needs extra memory to store the child pointers,
%and the memory accesses in \BST are highly random.
In addition, the prefix-doubling process has $\log |F'_i|$ rounds of thread synchronization,
which may incur high overhead in scheduling.
%where the overhead may outweigh the benefits of parallelism especially when $|F_i|$ is small.
%Therefore, we propose another data structure, \WinTree, to maintain the vertices by their marginal gains.
To overcome such issues, we propose a solution based on winning trees, and call it \wintree{}.
}

\section{Experiments}
\label{sec:exp}
% \input{fig_compact}

% \subsection{Setup.}
\myparagraph{Setup.} We implemented \oursystem{} in C++.
We run our experiments on a 96-core (192 hyperthreads) machine with four Intel Xeon Gold 6252 CPUs and 1.5 TB of main memory.
We use \texttt{numactl -i all} in experiments with
more than one thread to spread the memory pages across CPUs in a round-robin fashion.
We run each test four times and report the average of the last three runs.

We tested 17 graphs with information shown in \cref{tab:graph_info}. %. The information of the graphs is given in \cref{tab:graph_info}.
%Previous works usually tested on small graphs up to 4.85 M vertices and 85.7 M edge (soc-LiveJournal1~\cite{backstrom2006group}) limited by the complexity of the algorithms and limited memory sizes.
%Besides those commonly used benchmarks, we also include some large web graphs up to 978 M vertices and 74744 M edges (ClueWeb~\cite{webgraph}).
%Besides those commonly used benchmarks of social networks, we include graphs with a wide range of sizes and distributions,
%As an algorithmic paper mainly focuses on performance improvement, we believe it is helpful to study more graphs with a wide range of sizes and distributions.
%To understand the generality of our approaches, 
We include real-world graphs with a wide range of sizes and distributions, including \textbf{five billion-scale or larger graphs}.
In addition to the commonly used benchmarks of social networks, we also include web graphs, road networks, and
\knn{} graphs (each vertex is a multi-dimensional data point connecting to its $k$-nearest neighbors~\cite{wang2021geograph}).
Solving IM on such graphs simulates the influence spread between websites (web graphs), geologically connected objects (road networks),
and geometrically close objects (\knn{} graphs).
Based on graph patterns, we call social and web graphs \defn{scale-free} graphs,
and the rest \defn{sparse} graphs.
%The two different types of graphs have different patterns in performance.
%For synthetic graphs, we create two 2D grids (SQR and REC) and two sampled grids (SQR' and REC', each edge is sampled with probability 0.6), where each row and column are circular.
We symmetrize the directed graphs to make them undirected.
\oursystem{} (with compression) is the only tested system that can process
the largest graph ClueWeb~\cite{webgraph} with 978M vertices and 74B edges.
%and several other billion-scale graphs.
Even with 1.5TB memory, \oursystem{} can process Clueweb only when $\rate\le 0.25$ ($4\times$ or more compression ratio in sketches), which shows the necessity of compression.

We select $k=100$ seeds in all tests.
We use the IC model with constant propagation probability $p$ in the same graph.
\revision{
For scale-free networks, we use $p=0.02$, similar to previous work~\cite{chen2009efficient, kim2013scalable, cohen2014sketch, gokturk2020boosting}. 
For sparse graphs, we set $p=0.2$ since the average vertex degrees are mostly within 5. 
We also tested two other edge probability distributions similar to previous papers~\cite{gokturk2020boosting, minutoli2019fast}.
Among different distributions, we observed similar relative performance among the tested algorithms. 
Thus, we provide the result using fixed $p=0.02$ or $0.2$ here. Full results on other distributions can be found in \ifconference{the full version~\cite{wang2023fast}}\iffullversion{\cref{sec:distribution}}. 
}
% \knn graphs are widely used in machine learning algorithms \cite{XXX}.
%Since most of the selection time is spent on the earlier rounds, using a larger $k$ will not significantly increase the total time.
When comparing the \defn{average} numbers across multiple graphs, we use the \defn{geometric mean}.
% \revision{
% \st{
% For scale-free graphs, we set $p=0.02$ following the literature. 
% For sparse graphs, we set $p=0.2$ as the vertices have smaller degrees.
% }
% }
%Since we picked many graphs with different edge-to-vertex ratios, our experiment covers a wide range of cases.

\myparagraph{Software Libraries.}
We use ParlayLib~\cite{blelloch2020parlaylib} for fork-join parallelism and some parallel primitives (e.g., sorting).
We use the \ptree{} implementation from the PAM library~\cite{sun2018pam,sun2019implementing}, 
\hide{When computing connectivity in sketch construction, we use the union-find implementation \textsf{UniteRemCAS} from \textsf{ConnectIt}~\cite{dhulipala2020connectit}.}
and the \textsf{UniteRemCAS} implementation from \textsf{ConnectIt}~\cite{dhulipala2020connectit} for parallel connectivity. 

% Table generated by Excel2LaTeX from sheet 'graph info'
\begin{table}[t]
  \centering
  \small

    \caption{\textbf{Graph Information.} %$n$: number of vertices. $m$: number of edges. 
    {Influence}: the maximum influence spread of 100 seeds selected by \oursystem{} ($R=256$), \infuser{} ($R=256$), and \ripples{} ($\epsilon=0.5$), rounded to integers. Most influence values are evaluated by 20000 simulations. The {underline} numbers on large graphs use 2000 simulations.
    ``\textbf{-}'': unable to evaluate within 50 hours using 2000 simulations.%\yihan{check}
    %``$\boldsymbol{*}$'': graphs with more than a billion edges.
    } % \vspace{-.5em}
    
    \begin{tabular}{clrr|r|l@{ }}
          &       & \multicolumn{1}{c}{$|V|$} & \multicolumn{1}{c|}{$|E|$} & \multicolumn{1}{c|}{\textbf{Influence}} & \multicolumn{1}{c}{\textbf{Notes}} \\
    \midrule
    \multirow{9}[1]{*}{\begin{sideways}\textbf{Social}\end{sideways}}
            % & \textbf{HP} & 0.01M & 0.24M & 1209  & HepPh~\cite{leskovec2007graph} \\
          & \textbf{EP} & 0.08M & 0.81M & 5332  & Epinions1~\cite{richardson2003trust} \\
          & \textbf{SLDT} & 0.08M & 0.94M & 6342  & Slashdot~\cite{leskovec2009community} \\
          & \textbf{DBLP} & 0.32M & 2.10M & 1057  & DBLP~\cite{yang2015defining} \\
          & \textbf{YT} & 1.14M & 5.98M & 29614 & com-Youtube~\cite{yang2015defining} \\
          & \textbf{OK} & 3.07M & 234M  & 1460000 & com-orkut~\cite{yang2015defining} \\
          & \textbf{LJ} & 4.85M & 85.7M & 376701 & soc-LiveJournal1~\cite{backstrom2006group} \\
          & \textbf{TW} & 41.7M & 2.40B & \underline{11776629} & Twitter~\cite{kwak2010twitter} \\
          & \textbf{FT} & 65.6M & 3.61B & \underline{19198744} & Friendster~\cite{yang2015defining} \\
    \midrule
    \multirow{2}[2]{*}{\begin{sideways}\textbf{Web}\end{sideways}}
          & \textbf{SD} & 89.2M & 3.88B & \underline{15559737} & sd\_arc~\cite{webgraph} \\
          & \textbf{CW} & 978M  & 74.7B &   -    & ClueWeb~\cite{webgraph} \\
    \midrule
    \multirow{2}[2]{*}{\begin{sideways}\textbf{Road}\end{sideways}}
            & \textbf{GER} & 12.3M & 32.3M & 384   & Germany~\cite{roadgraph} \\
            & \textbf{USA} & 23.9M & 57.7M & 370   & RoadUSA~\cite{roadgraph} \\
    \midrule
    \multirow{5}[2]{*}{\begin{sideways}\textbf{k-NN}\end{sideways}} & \textbf{HT5} & 2.05M & 13.0M & 1018  & HT~\cite{uciml,wang2021geograph}, $k$=5 \\
          & \textbf{HH5} & 2.05M & 13.0M & 2827  & Household~\cite{uciml,wang2021geograph}, $k$=5 \\
          & \textbf{CH5} & 4.21M & 29.7M & 355065 & CHEM~\cite{fonollosa2015reservoir,wang2021geograph}, $k$=5 \\
          & \textbf{GL5} & 24.9M & 157M  & 11632 & GeoLife~\cite{geolife,wang2021geograph}, $k$=5 \\
          & \textbf{COS5} & 321M  & 1.96B & 4753  & Cosmo50~\cite{cosmo50,wang2021geograph}, $k$=5 \\
    % \midrule
    % \multirow{4}[2]{*}{\begin{sideways}\textbf{Synthetic}\end{sideways}} & \textbf{SQR} & 16M   & 64M   & 285   & 2D grid $4000 \times 4000$ \\
    %       & \textbf{SQR'} & 16M   & 32.7M & 246   & sampled SQR \\
    %       & \textbf{REC} & 10M   & 40M   & 285   & 2D grid $10^3 \times 10^4$ \\
    %       & \textbf{REC'} & 10M   & 20.4M & 245   & sampled REC \\
    \bottomrule
    \end{tabular}%
    
  \label{tab:graph_info}%
\end{table}% 

\begin{table*}[htbp]
\small
  \centering % \vspace{-1.5em}

  \caption{\textbf{Running time, memory usage, and influence spread (normalized to the maximum) of all tested systems on a machine with 96 cores (192 hyperthreads).}
    \revision{Relative influence is the influence spread normalized to the maximum influence spread among \ours{} ($R=256$), \infuser{} ($R=256$), and \ripples{} 
 ($\epsilon=0.5$).}
    ``-'': out of memory (1.5 TB) or time limit (3 hours).
    \ours{1} is our implementation with \wintree{} without compression. 
    {\ours{0.1}} is our implementation with $\rate=0.1$ ($10\times$ compression for sketches).
    {\infuser{}}~\cite{gokturk2020boosting} and {\ripples{}}~\cite{minutoli2019fast, minutoli2020curipples} are baselines.
    We report the \emph{best time} of \infuser{} and \ripples{} by varied core counts (the scalability issue of \infuser{} and \ripples{} are shown in \cref{fig:scale}).
    {CSR} is the memory used to store the graph in CSR format (see more in \cref{sec:overall_performance}).
    The bold numbers are the fastest time/smallest memory among all implementations on each graph.
    %The bold numbers in memory usage are the smallest memory of graphs among all implementations.
    The underlined numbers in memory usage are the smallest memory among systems that do not use compression (\ours{1}, \ripples{} and \infuser{}). A heatmap of the full result is in \cref{fig:heatmap}. % \vspace{-.5em}
    }
    
    \begin{tabular}{clrrr|rrrr|rrrrr}
          &       & \multicolumn{3}{c|}{\textbf{Relative Influence}} & \multicolumn{4}{c|}{\textbf{Total Running Time (second)}}     & \multicolumn{5}{c}{\textbf{Memory Usage (GB)}} \\
          &       & \multicolumn{1}{c}{\textbf{Ours}} & \multicolumn{1}{@{}c}{\textbf{\infuser{}}} & \multicolumn{1}{c|}{\textbf{\ripples{}}} & \multicolumn{1}{c}{\textbf{\ours{1}}} & \multicolumn{1}{c}{\textbf{\ours{0.1}}} & \multicolumn{1}{@{}c}{\textbf{\infuser}} & \multicolumn{1}{c|}{\textbf{\ripples{}}} & \multicolumn{1}{c}{\textbf{CSR}} & \multicolumn{1}{c}{\textbf{\ours{1}}} & \multicolumn{1}{c}{\textbf{\ours{0.1}}} & \multicolumn{1}{@{}c}{\textbf{\infuser{}}} & \multicolumn{1}{c}{\textbf{\ripples{}}} \\
    \midrule
    \multirow{9}[2]{*}{\begin{sideways}\textbf{Social}\!\!\end{sideways}}
        % & \textbf{HP} & 100\% & 99.2\% & 96.5\% & 0.16  & 0.23  & \textbf{0.07} & 0.42  & 0.002 & 0.05  & 0.04  &\textbf{\textbf{0.03}} & 0.04 \\
          & \textbf{EP} & 100\% & 99.9\% & 98.9\% & \textbf{0.29} & 0.49  & 0.38  & 3.70  & 0.01  & \underline{0.14} & \textbf{0.06}  & 0.17  & 0.21 \\
          & \textbf{SLDT} & 100\% & 99.8\% & 99.1\% & \textbf{0.30} & 0.53  & 0.48  & 6.88    & 0.01  & \underline{0.15} & \textbf{0.05}  & 0.17  & 0.24 \\
          & \textbf{DBLP} & 100\% & 99.0\% & 98.3\% & \textbf{0.35} & 0.37  & 0.86  & 2.71  & 0.02  & 0.48  & \textbf{0.10}   & 0.67  & \underline{0.29} \\
          & \textbf{YT} & 100\% & 99.9\% & 98.2\% & \textbf{1.22} & 2.44  & 6.20   & 26.6  & 0.05  & 1.63  & \textbf{0.28}  & 2.36  & \underline{1.50} \\
          & \textbf{OK} & 100\% & 100\% & 99.8\% & \textbf{8.79} & 39.6  & 80.3  & 325   & 1.77  & \underline{6.13} & \textbf{2.44}  & 8.07  & 78.3 \\
          & \textbf{LJ} & 100\% & 99.9\% & 99.4\% & \textbf{6.00} & 20.5  & 61.3  & 130   & 0.68  & \underline{5.98} & \textbf{1.63}  & 10.5  & 20.4 \\
          & \textbf{TW} & 100\% & 100\% & -     & \textbf{93.4} & 378   & 639   & 5863  & 18.2  & \underline{63.0} & \textbf{25.9}  & 103   & 718 \\
          & \textbf{FT} & 100\% & 100\% & -     & \textbf{128} & 609   & 1973  & -     & 27.4  & \underline{97.9} & \textbf{39.9}  & 161   & - \\
    \midrule
    \multirow{2}[2]{*}{\begin{sideways}\textbf{Web}\!\end{sideways}}
          & \textbf{SD} & 100\% & 97.1\% & -     & \textbf{150} & 627   & 1684  & -     & 29.6  & \underline{125} & \textbf{44.8}  & 211   & - \\
          & \textbf{CW} & 100\% & -     & -     & -     & 9776 & -     & -     & 564   & -     & 738 & -     & - \\
    \midrule
    \multirow{2}[2]{*}{\begin{sideways}\textbf{Road}\!\end{sideways}} 
        & \textbf{GER} & 100\% & 70.4\% & 94.4\% & 9.35   & \textbf{8.53} & 26.7  & 392  & 0.33  & \underline{13.3} & \textbf{2.58}  & 25.1  & 22.8 \\
        & \textbf{USA} & 100\% & 74.7\% & 92.8\% & 14.6  & \textbf{13.7} & 53.1  & 8534 & 0.61  & \underline{25.8} & \textbf{5.05}  & 49.0    & 46.1 \\
    \midrule
    \multirow{5}[2]{*}{\begin{sideways}\textbf{$k$-NN}\!\!\end{sideways}} & \textbf{HT5} & 100\% & 85.7\% & 94.8\% & 0.72  & \textbf{0.68} & 2.73 & 9.26  & 0.11  & 1.37  & \textbf{0.24}  & 1.93  & \underline{0.79} \\
          & \textbf{HH5} & 100\% & 79.5\% & 97.6\% & 2.28  & \textbf{2.11} & 8.26  & 14.4  & 0.11  & 3.00     & \textbf{0.51}  & 4.26  & \underline{1.14} \\
          & \textbf{CH5} & 100\% & 92.3\% & 97.8\% & \textbf{3.52}  & 5.32 & 124   & 10.2  & 0.25  & 4.87  & \textbf{1.05}  & 8.77  & \underline{1.64} \\
          & \textbf{GL5} & 100\% & 76.9\% & 98.9\% & 19.5  & \textbf{17.9} & 116   & 232   & 1.36  & 27.6  & \textbf{6.02}  & 51.6  & \underline{8.18} \\
          & \textbf{COS5} & 100\% & 37.7\% & -     & 348 & \textbf{284}   & 2319  & -     & 17.0    & \underline{355}   & \textbf{66.1}  & 666   & - \\
    % \midrule
    % \multirow{4}[2]{*}{\begin{sideways}\textbf{Grid}\end{sideways}} & \textbf{SQR} & 100\% & 100\% & -     & 9.2   & \textbf{8.65} & 51.7  & -     & 0.6   & \underline{17.4} & \textbf{3.54}  & 32.9  & - \\
    %       & \textbf{SQR'} & 100\% & 89.3\% & -     & 8.9   & \textbf{8.28} & 37.7  & -     & 0.36  & \underline{17.2} & \textbf{3.33}  & 32.7  & - \\
    %       & \textbf{REC} & 100\% & 100\% & 99.8\% & 6.2   & \textbf{5.85} & 36.4  & 873   & 0.37  & \underline{10.9} & \textbf{2.23}  & 20.6  & 20.8 \\
    %       & \textbf{REC'} & 100\% & 89.6\% & 93.9\% & 6.1   & \textbf{5.76} & 25    & 1068  & 0.23  & \underline{10.8} & \textbf{2.09}  & 20.4  & 22.2 \\
    \bottomrule
    \end{tabular}%
    
  \label{tab:baselines}%
\end{table*}%

\myparagraph{Tested Algorithms.}
We tested both \ptree{} and \ourtree{} for seed selection.
%Our experiments show that \ourtree{} is more efficient in both time and space in most of our test cases.
In most cases, \ourtree{} is more efficient in both time and space,
so use \ourtree{} as the default option in \oursystem{}. We present more results comparing the two options in \cref{sec:compare_trees}.

We compare with three existing parallel IM systems: \infuser{}~\cite{gokturk2020boosting}, \NoSingles{}~\cite{popova2018nosingles} and \ripples{}~\cite{minutoli2019fast, minutoli2020curipples},
and call them \defn{baselines}.
%\infuser{} and our algorithms are based on the  algorithm~\cite{cheng2013staticgreedy}.
As introduced in \cref{sec:prelim}, \infuser{} uses a similar sketch-based approach as \oursystem{} but does not support compression or parallel CELF.
%is similar to \oursystem{} with no compression.
\NoSingles{} and \ripples{} both use Reverse Influence Sampling~\cite{borgs2014maximizing}.
In our tests, \revision{\ripples{} is always better than \NoSingles{} in time and space, so we only report the results of \ripples{}.
We have also tested some sequential algorithms, such as PMC~\cite{ohsaka2014fast} and IMM~\cite{tang2015influence}, but their running times are not competitive to the parallel implementations.}
%slower than the parallel \oursystem{} on all graphs and slower than \infuser{} on most graphs
We observe that \infuser{} and \ripples{} have \emph{scalability issues} when the number of threads increases (see examples in \cref{fig:scale}).
Hence, we report \emph{their shortest time among all the tested numbers of threads}.

\revision{Each algorithm has a parameter that controls the solution quality: $R$ for \infuser{} and \oursystem{} and $\epsilon\in (0,0.5]$ for \ripples{}.
The solution quality increases with larger $R$ or smaller $\epsilon$.
When comparing running time, we guarantee that \oursystem{} \emph{always gives a better solution} than the baselines.}
%we control the parameter to get almost the same spread and compare their performance.
For \infuser{} and \oursystem{}, we set the number of sketches $R=256$.
\oursystem{} with $R=256$ is on average more than 99\% of the quality when using $R=2^{15}$, which is consistent with the observation in the \staticgreedy{} paper~\cite{cheng2013staticgreedy} (see our full version for quality analysis).
\iffullversion{\begin{figure}[t]
  \centering
  \includegraphics[width=\columnwidth]{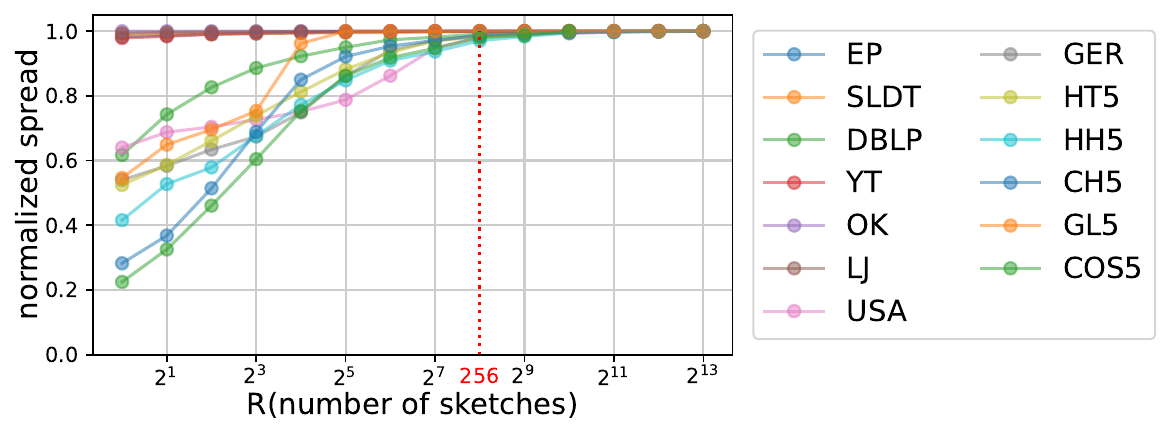}
  \caption{\small \textbf{Spread curves for our algorithm.}  In each plot, $x$-axis is the number of sketches (R) and $y$-axis is the spread normalized to the maximum spread of each graph. Higher is better.
  \label{fig:spread}}
\end{figure} }
\iffullversion{We independently verified this and shown it in \cref{fig:spread}}
%which empirically considered as a reasonable from the original \textsf{StaticGreedy} paper~\cite{cheng2013staticgreedy}.
%We independently verified this in \cref{fig:spread}---taking $R=256$ yields over $99\%$ influence over that on $R=2^{15}$ on average.
For \ripples{},
%the parameter $\epsilon$ controls the deviation,
smaller $\epsilon$ means better accuracy but more time.
We tested $\epsilon$ in $[0.13, 0.5]$ as in their paper. \oursystem{} under $R=256$ yields about the same solution quality as $\epsilon=0.13$ (the best tested setting in their paper).
%\yan{shouldn't it be 0 for best quality?}
When reporting time, we use $\epsilon=0.5$, which gives the fastest running time, and the quality is still reasonably high (at least 93\% of our best influence).
%for \ripples{}, as that is the best parameter for their running time, and the quality is still reasonably high.

We observe that on sparse graphs,
%as $R$ increases from $2^5$ to $2^{15}$, the spread of \infuser{} first increases and then decreases significantly.
the influence spread of \infuser{} is only 38--92\% of the best achieved by \oursystem{} and \ripples{}.
Although theoretically, \oursystem{} and \infuser{} should give the same output,
% (assuming a fixed seed for the random number generator)
\infuser{} uses many optimizations that sacrifice solution quality.
We tried to increase $R$ and various other attempts in \infuser{}, but they did not improve the solution quality.
Therefore, we keep the same value $R=256$ for \oursystem{} and \infuser{}.
%We believe this is caused by a small bug causing race condition. We fixed it using atomic operations and
%achieved reasonable influence spread in most of the cases. We tested the version with this fix.
%Our fix has little impact in their running time.
%We found some bugs that may cause the answer to be incorrect.
%Some of them are easy to fix, such as fixing data race problems by using atomic operations; some are hard to fix, such as the inappropriate hash function written in vectorization instructions.
%In order to truly report its best running time, we keep its code unchanged.

\subsection{Overall Time and Space}
\label{sec:overall_performance}

\cref{tab:baselines} shows the running time, memory usage, and normalized influence spread of all systems. 
\hide{\revision{
The relative influence is the influence spread normalized to the maximum influence spread selected by \ours{} ($R=256$),  \infuser{} ($R=256$), and \ripples{} ($\epsilon$). The value of maximum influence is shown in \cref{tab:graph_info}. 
}}
\ours{1} and \ours{0.1} are \oursystem{} with $\rate=1$ (no compression) and $\rate=0.1$ ($10\times$ sketch compression), respectively.
To illustrate the relative performance, we present a heatmap in \cref{fig:heatmap}, where all the numbers (time and space) are normalized to \ours{1}.
\ours{0.1} is the \emph{only algorithm that can process the largest graph CW}~\cite{webgraph}.
With similar quality, \ours{1} is \emph{faster than all baselines on all graphs}, and \ours{0.1} is just slower than \infuser{} on the two smallest graphs.
\ours{0.1} has the smallest space usage on \emph{all} graphs.
%Our algorithms are slightly worse than \infuser{} in time on the smallest graph HP, since our optimizations are designed to improve the scalability on large graphs.
The advantage of \oursystem{} is more significant on larger graphs, both in time and space.
%More details of the advantage of \oursystem{},
%From \cref{fig:heatmap}, we can see that the advantage of our algorithm,  than on small graphs.
%Ours$_{0.1}$ \textbf{uses memory smaller than all baselines on 21 graphs}, except on the smallest graph HP.

\hide{
On all graphs, our algorithms achieve the best influence spread, and \ripples{} also achieves at least 93\% of the best influence.
Although theoretically \oursystem{} and \infuser{} should give the same output (assuming a fixed seed for the random number generator),
we note that \infuser{} uses several performance optimizations that sacrifices solution quality.
As a result, on sparse graphs, \infuser{} only achieves 38--92\% of the best score.
}

\myparagraph{Running Time.}
\oursystem{} is significantly faster than the baselines on almost all graphs.
\iffullversion{As mentioned, we report the best running time among all tested core counts for \infuser{} and \ripples{},
since they may not scale to 192 threads (see \cref{fig:scale,sec:scalability}).
Even so, \oursystem{} is still faster on all graphs.}
On average, \ours{1} is $5.7\times$ faster than \infuser{} and $18\times$ faster than \ripples{}.
\ours{0.1} is slightly slower than \ours{1}, but is still $3.2\times$ faster than \infuser{} and $10\times$ faster than \ripples{}.
%Here we note that \infuser{}'s solution quality is not as good as the other two systems.

%To better understand the time and space usage for \oursystem{}, we vary our parameter $\rate$ and show the running time for the sketching and seeding, as well as space usage, in \cref{fig:compression}.
%From the figure, we can see a clear trend that $\rate$ trades off (smaller) auxiliary space with (more) time in seed selection.
%In general, using compression reduces space usage by sacrificing the running time in query.
In general, the compression in \oursystem{} saves space by trading off more time. %We show this trend in \cref{fig:compression}.
When $\rate=1$, the CC sizes for all vertices are stored in the sketches, and a re-evaluation only needs a constant time per sketch.
When $\rate=0.1$, each query involves a search to either find a center, or visit all connected vertices, which roughly costs $O(1/\rate)$ on each sketch.
Indeed, on all scale-free graphs, \ours{0.1} takes a longer time than \ours{1}.
Interestingly, on most sparse graphs, \ours{0.1} can be faster than \ours{1}.
This is because seed selection only takes a small fraction of the total running time (except for CH5), so the slow-down in seed selection is negligible for the overall performance. 
Meanwhile, avoiding storing $O(Rn)$ connectivity sizes reduces memory footprint and makes the sketching step slightly faster, which overall speeds up the running time.

\myparagraph{Memory Usage.}
We show space usage of \oursystem{} and baselines in \cref{fig:heatmap,fig:compression,tab:baselines}.
We also show the size of representing the graph in standard Compressed Sparse Row (CSR) format as a reference in \cref{tab:baselines,fig:compression},
which roughly indicates the space to store the input graph. CSR uses 8 bytes for each vertex and each edge. %\letong{\oursystem{} and \infuser{} store CSR uses 8 bytes for each vertex and edge. Not sure for \ripples{}.}
Using $\rate=0.1$, \oursystem{} uses the least memory on all graphs.
Even without compression, our algorithm uses less memory than the baselines on 10 out of 16 graphs.
%\ours{0.1} is always more space efficient than \ours{1}.
Note that although the compression rate for sketches is $10\times$ in \ours{0.1}, the total space also includes
%there is also space usage,
the input graph and the data structure for seed selection.
Therefore, we cannot directly achieve a $10\times$ improvement in space. In most cases, the total memory usage is about $5\times$ smaller.

\hide{For example, on EP, the running time of \infuser{} under 192 threads is 19$\times$ slower than under 8 threads. For \ripples{}, we even need to set the number of threads used for sketch construction and seed selection separately. We first select the best number of threads for seed selection (2 for HP, 4 for EP, and 16 for other graphs), then test \ripples{} under different threads to report the best total time.
}
%The in-depth analysis of scalability issues will be discussed later.

\myparagraph{Summary.} Overall, \oursystem{} has better performance than the baselines in both time and space.
We note that the space usage of \ripples{} can be (up to $3.4\times$) better than \ours{1} on certain graphs (but still worse than \ours{0.1}),
but in these cases the running time is also much longer (by $2\times$ to $583\times$ ).
On scale-free graphs, \ours{0.1} is $2.5\times$ slower than \ours{1} on average, but uses $3\times$ less space.
On sparse graphs, \ours{0.1} is almost always better in both time and space.

%\myparagraph{Memory Usage.}

\subsection{Scalability}\label{sec:scalability}
\hide{
\begin{figure*}[th]
  \centering % \vspace{-1em}
  \begin{minipage}[b]{0.77\textwidth}
  \includegraphics[width=\textwidth]{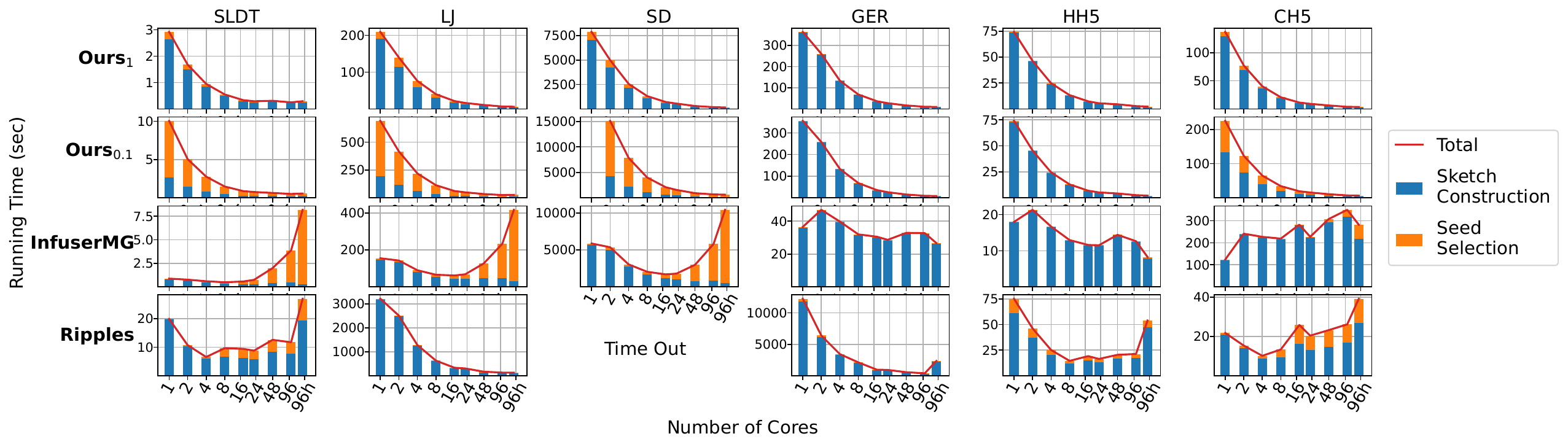}
  \caption{\small \textbf{Breakdown running time under different threads for different IM algorithms.}  We select four representative graphs to show. Each row represents an algorithm. In each plot, the $x$-axis is core counts (last data point is 96 cores with hyperthreading) and the $y$-axis is running time in seconds.
  \label{fig:scale}}
  \end{minipage}
  \begin{minipage}[b]{0.22\textwidth}
    \includegraphics[width=\textwidth]{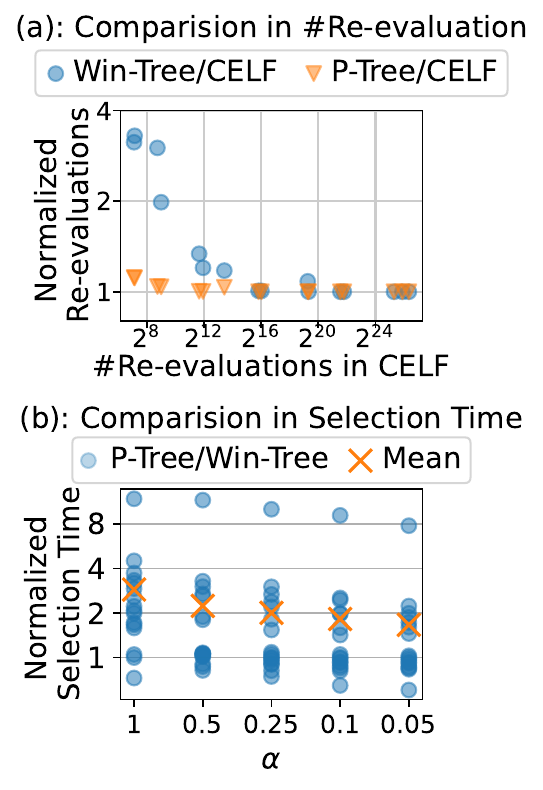}
  \caption{\small \textbf{Compare \ptree and \WinTree{} in the number of re-evaluations and selection time.}
 %  Figure (a) shows the comparison in
 % the number of re-evaluations, where each point represents a test case whose $x$-index is \#re-evaluations needed in CELF, and $y$-index is that needed in \ptree (blue circle) or \WinTree (orange triangle). The red line is $y=x$. A point above it means a test case that parallel priority queues evaluate more vertices than CELF.
 % Figure (b) shows the distribution of \ptree selection time/\WinTree selection time under different $\alpha$.
\label{fig:trees}}
  \end{minipage}
\end{figure*}
}
\hide{
\begin{figure}
  \includegraphics[width=.9\columnwidth]{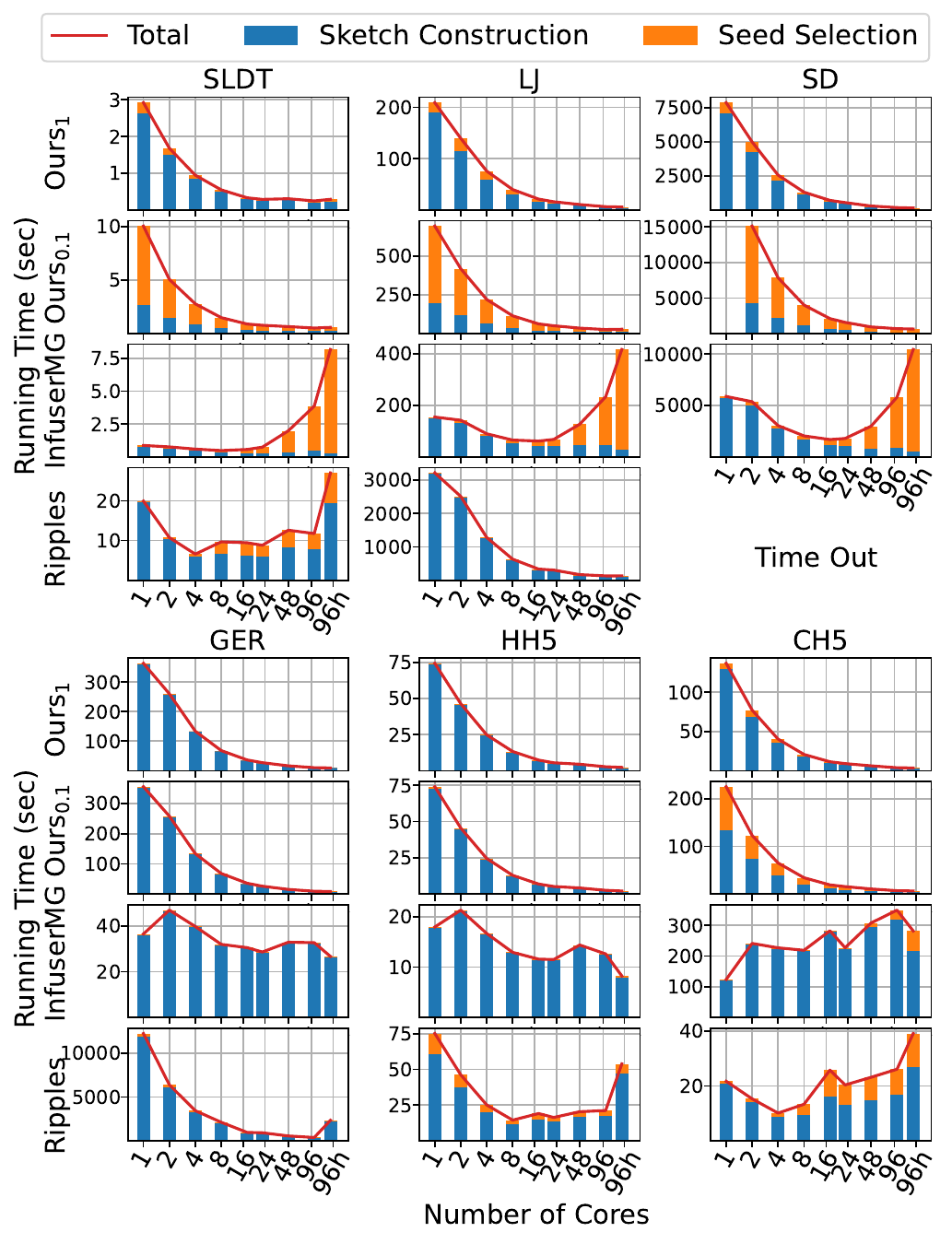}
  \caption{\small \textbf{Running time using different core counts for different IM algorithms.}
  In each plot, the $x$-axis shows core counts (96h means 96 cores with hyperthreading) and the $y$-axis is running time in seconds.
  \label{fig:scale}}
\end{figure}
}

\begin{figure*}[th]
  \centering % \vspace{-1em}
  \includegraphics[width=\textwidth]{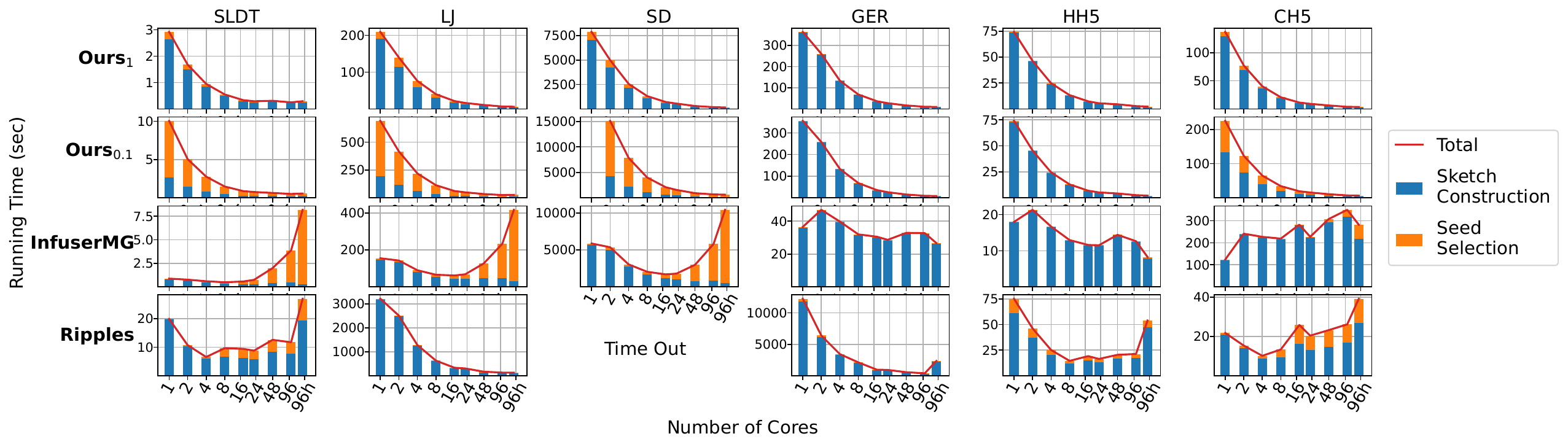} % \vspace{-.25em}
  \caption{\small \textbf{Running time using different core counts for different IM algorithms.}
  In each plot, the $x$-axis shows core counts (96h means 96 cores with hyperthreading) and the $y$-axis is running time in seconds.
  \label{fig:scale}} % \vspace{-.5em}
\end{figure*} 
We study the scalability of all systems. 
We present the performance on six representative graphs in %. %(one in each category).
\cref{fig:scale} %presents the running time 
with varying core counts $P$.
We separate the time for sketch construction (the ``sketch time'') and seed selection (the ``selection time'') to study the two components independently.
For \oursystem{}, both sketch and selection time decrease with more cores and achieve almost linear speedup.
\hide{
Unfortunately, such scalability trend is not always true for both \infuser{} and \ripples{}.
They scale to 8--16 cores on many graphs, but the performance may get worse with more cores.
}
\revision{
In contrast, \infuser{} and \ripples{} get the best speedup with 8--16 cores on many graphs, and perform worse with more cores. 
For both \infuser{} and \ripples{}, the sketch time scales better than the selection time.
%The poor performance is mostly from the selection time. 
This is because both algorithms use well-parallelized algorithms to construct sketches.
For example, \infuser{} uses a standard coloring~\cite{slota2014bfs, orzan2004distributed}
idea for parallel connectivity.
% However, for seed selection, both systems use sequential CELF, and only parallelize the re-evaluation on $R$ sketches.\letong{Ripples not use CELF}
%However, for seed selection, both systems only parallelize the evaluation on $R$ sketches.
%When $P$ becomes comparable to $R$, such an approach may cause scheduling overhead and slow down the performance.
However, in seed selection, the baseline algorithms can have longer running times with more threads used. 
This is because they only use parallelism within one evaluation, and perform all evaluations one by one.
For both baselines, a single evaluation does not cause much computation.
In this case, the overhead of scheduling the parallel threads can be more expensive than the computation (and can even dominate the cost), which increases with the number of threads. 

The scalability curves indicate a major performance gain of \oursystem{} over baselines is from \defn{better scalability}. On SLDT, LJ, and SD, although the sequential running time of \infuser{} is better than \oursystem{}, \oursystem{} achieves better performance when more than 8 cores are used. The advantage is more significant with more cores. 
%due to our good scalability, we get our best self-speedup under 96 or 192 threads, while \infuser{} get its best speedup only under 8 or 16 threads.  Our 192 threads running time is faster than the best running time of \ripples among all core numbers because our algorithm can take the advantage of more cores. 
}

The scalability curves also \emph{indicate the necessity of our parallel CELF}.
% \hide{
% For \ours{0.1}, since the re-evaluation time becomes longer than \ours{1} due to compression, the selection time may be longer than sketch construction,
% where parallelizing seed selection is important in improving the overall performance.
% \revision{Even for \ours{1} where the selection time is a small fraction of the total running time when $P=1$,
% the sequential selecting time is still much longer than the total running time when $P>16$.}
% This means that if this step is not well-parallelized, it will become the performance bottleneck when $P$ is large.
% }
\revision{
For \ours{0.1}, the sequential selection time takes a large portion of the total running time,
so parallelizing seed selection is crucial for improving the overall performance.
Even for \ours{1}, %where the sequential selection time takes a small fraction of the total running time,
the selection time on 1 core is still much longer than the total parallel running time. 
This means that \emph{if} this step was not well-parallelized, it would become the performance bottleneck when $P$ is large.
Using our new techniques from \cref{sec:compact_sketching,sec:selecting}, both steps scale well.
}
%parallelism actually negative impacts the performance based on our experiment.

%its parallelism only comes from processing $R$ sketches (embarrassingly parallelized), but each seed candidate is re-evaluated one by one using sequential CELF.
%With the number of threads similar to the sketch number $R$, this loses parallelism and may cause load imbalance.

%\yihan{Still not finished. Add ours0.1 to the result. Maybe also add SD to the set of graphs}
\hide{
For \infuser{}, on Slashdot and LJ, although the sketch construction time decreases, the selection time even increases as the number of cores increases. For \ripples{}, note that when the number of threads is greater than 16, we still run the seed selection step under 16 threads, because otherwise, it will take a much larger time. From \cref{fig:scale}, the sketch construction time for \ripples{} on Slashdot, HH5 and CH5 first decreases then increases as the number of threads increases.
Both \infuser{} and \ripples{} use OpenMP as their parallel scheduler, where parallel overhead is proportional to the number of parallel threads.
\yihan{not sure if we are going to argue our benefit is from not using OpenMP. Also not sure if it's safe to claim so for OpenMP, since almost all schedulers should have overhead proportional to the number of parallel threads.}
However, their algorithm lacks parallelism (either only in the selection step or in both steps), and increasing threads can not bring enough speedup to overcome the overhead of creating and synchronizing threads.
%\yihan{not finished, wait until we have finalized the graphs in the scalability figure.}
}

% % We find an interesting fact that compression may not sacrifice running time, and the opposite, may speed up the running time.
% Comparing Ours$_1$ and Ours$_{0.1}$, we find that Ours$_1$ is faster on social/web graphs and Ours$_{0.1}$ is faster on road/\knn{}/grid graphs. The reason behind this is that the effects of compression on the selection step and sketch step are different.
% The in-depth analysis of the effect of compression is in \cref{sec:compression_effect}.

\subsection{Analysis of the Proposed Techniques}\label{sec:exp-ana}

\begin{figure}[t]
  \centering
  \includegraphics[width=\columnwidth]{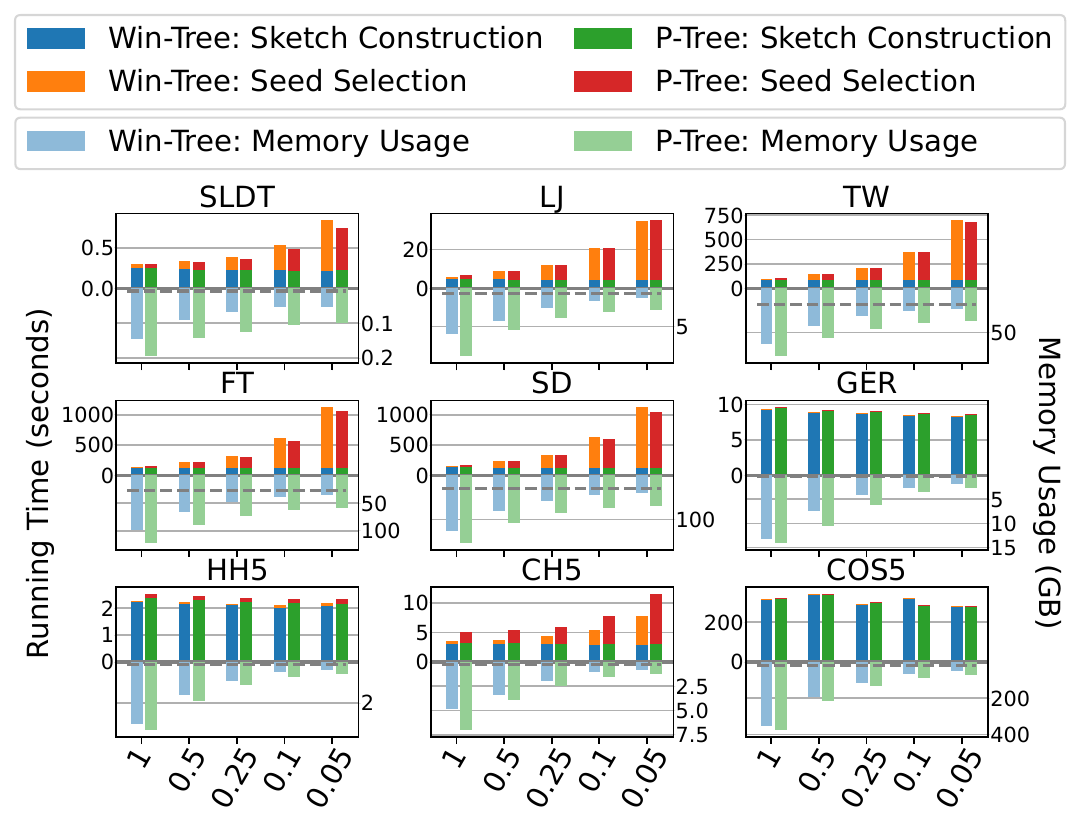}%\vspace{-.5em}
  \caption{\small \textbf{Running time and memory with different values of $\alpha$.} The $x$-axis represents the compression rate.  The growing up $y$-axis represents the running time (in seconds) of the sketching and selecting process. The growing down $y$-axis represents the total memory usage (in GB). The gray horizontal line represents the CSR size of each graph, which is the basic memory we need to load the graph.%\vspace{-.5em}
  \label{fig:compression}}
\end{figure} 

%Next, we evaluate the two proposed techniques in this paper: sketch compression and data structures for parallel seed selection.

\myparagraph{Compression.}
\label{sec:compression_effect}
We evaluate the time and space usage of different compression ratios by controlling the parameter $\rate$, and present the results in \cref{fig:compression}.
The gray dashed line represents the CSR size of each graph, which is the space to store the input.
Compression always reduces memory usage, but may affect running time differently.
%As shown in \cref{fig:compression}, the bars on the bottom represent the memory usage of \ptree{} or \ourtree{} under different compression rates.
%The smaller the compression rate is, the less memory is used.
Memory usage always decreases with the value $\rate$ decreases.
As mentioned previously, the actual compression rate can be lower than $1/\rate$, since our compression only controls the memory for storing sketches, and there are other space usage in the algorithm.
Roughly speaking, using $\rate=0.05$ can save space on graphs by up to $8\times$ and shrink the space very close to the input graph size.

Compression affects the running time differently for sketch construction and seed selection.
Smaller $\rate$ indicates longer running time in seed selection,
%but it does not increase the time in sketching.
%In fact, due to the reduced memory footprint, the sketching time can be improved slightly by smaller $\rate$.
but may improve the sketch time slightly due to the reduced memory footprint.
The effect of compression on total running time depends on the ratio of sketching and selection time,
where scale-free and sparse graphs exhibit different patterns.
%On sparse graphs, the sampled graph is usually not well-connected.
\hide{In scale-free graph, the sampled graph usually has one or several large CCs,
while the sparse graphs usually have many small CCs in the sampled graphs.}
Scale-free graphs usually have one or several large CCs in the sampled graphs since they are dense,
while sparse graphs usually have many small CCs.
This can also be seen by the overall influence in \cref{tab:graph_info}:
even though we use a smaller $p$ on scale-free networks, the total influence is still much larger than the sparse graphs.
On scale-free graphs, due to large CCs on sampled graphs, selecting any seed in a large CC may significantly lower the score of other vertices,
% leading to much large total re-evaluations than sparse graphs.
leading to much more total evaluations than sparse graphs.
% As a result, smaller $\rate$ causes a clear time increase for seed selection on scale-free networks since each re-evaluation becomes slower,
% while the compression only has a small impact on most of the sparse graphs.
As a result, a smaller $\rate$ causes a clear time increase for seed selection on scale-free networks because each evaluation becomes slower while only having a small impact on most sparse graphs.

\ifconference{We provide the total number of evaluations on each graph as a reference in the full version \cite{wang2023fast}.}
\iffullversion{We provide the number of re-evaluations on each graph in \cref{tab:re-evaluations} as a reference.}

\iffullversion{
    % Table generated by Excel2LaTeX from sheet 're-evaluations'
\begin{table}[t]
  \centering
  \small
    \begin{tabular}{clrrrr}
          &       & \multicolumn{1}{c}{\textbf{\emph{n}}} & \multicolumn{1}{c}{\textbf{CELF}} & \multicolumn{1}{c}{\textbf{P-Tree}} & \multicolumn{1}{c}{\textbf{Win-Tree}} \\
    \midrule
    \multirow{9}[2]{*}{\begin{sideways}\textbf{social}\end{sideways}} & \textbf{HP} & 12008 & 7936  & 7964  & 8322 \\
          & \textbf{EP} & 75879 & 57980 & 58010 & 58398 \\
          & \textbf{SLDT} & 77360 & 66618 & 66649 & 67035 \\
          & \textbf{DBLP} & 317080 & 3255  & 3269  & 4075 \\
          & \textbf{YT} & 1134890 & 656014 & 656068 & 656727 \\
          & \textbf{OK} & 3072627 & 3066983 & 3066995 & 3067514 \\
          & \textbf{LJ} & 4847571 & 3609995 & 3610105 & 3611583 \\
          & \textbf{TW} & 41652231 & 41200415 & 41200437 & 41200774 \\
          & \textbf{FT} & 65608366 & 61946405 & 61946434 & 61947064 \\
    \midrule
    \multirow{2}[2]{*}{\begin{sideways}\textbf{web}\end{sideways}} & \textbf{SD} & 89247739 & 84554009 & 84576755 & 84576236 \\
          & \textbf{CW} &    978408098   &   & 873230222   &    873256302   \\
    \midrule
    \multirow{2}[2]{*}{\begin{sideways}\textbf{road}\end{sideways}} & \textbf{USA} & 23947348 & 143   & 159   & 465 \\
          & \textbf{GER} & 12277375 & 139   & 155   & 416 \\
    \midrule
    \multirow{5}[2]{*}{\begin{sideways}\textbf{k-NN}\end{sideways}} & \textbf{HT5} & 928991 & 516   & 536   & 1007 \\
          & \textbf{HH5} & 2049280 & 3931  & 3943  & 4717 \\
          & \textbf{CH5} & 4208261 & 628918 & 629501 & 653157 \\
          & \textbf{GL5} & 24876978 & 11042 & 11430 & 12996 \\
          & \textbf{COS5} & 321065547 & 431   & 449   & 1264 \\
    \bottomrule
    \end{tabular}%
  \caption{Numbers of evaluations for each graph instance using different algorithms.}
  \label{tab:re-evaluations}%
\end{table}%

}

%Therefore, the selection time on sparse graphs (except for CH5) is usually negligible
\hide{
Unlike the scale-free graphs where the sampled graph usually has a large CC, the sparse graphs usually have many small CCs after sampling.
Thus, the influence is generally smaller (see \cref{tab:graph_info}), making the re-evaluation faster, expect for CH5.
Meanwhile, on scale-free graphs, even we picked small value of $p$, the influence is generally larger.
This will lead to more stale values after selecting each seed and longer re-evaluation time.
In this case, smaller $\rate$ causes a clear time increase for seed selection on scale-free networks.
We provide the total number of re-evaluations on each graph as a reference in the full version\yan{check}.
}

%Therefore, selecting a seed is unlikely to affect the score of many other vertices.
%As a result, the seed selection process only re-evaluates a few candidates to select the next seed, thus making the process cheap.
%This is true in all the tested sparse graphs except for CH5.
%\letong{refer to the total number of revaluations in full version?}
%Therefore, the running time on sparse graphs is mostly dominated by sketch construction, and compression slightly improves the overall running time.
%On scale-free graphs, since there usually exists at least one large CC in the sampled graph, selecting any seed in this CC significantly lowers the score of many other vertices.
%As a result, %many of the top candidates in the priority queue have a much lower true score,
%we may have to re-evaluate all such candidates and re-insert them back into the queue.
%Therefore, the seed selection affects the overall running time more significantly,
%and compression, which increases the selection time, will slow down the overall time.
%he speedup in sketch construction is smaller than the slowdown in seed selection that compression brings. On such graphs, compression will slow down the total running time.

% \myparagraph{Memory Usage}

% \myparagraph{Performance}

\myparagraph{\ptree vs. \ourtree.} \label{sec:compare_trees}
% \myparagraph{Number of Re-evaluates}
% \input{figures/fig_evaluates}
% \input{figures/fig_trees}
We now study both data structures for seed selection.
Our goal is to achieve high parallelism without causing much overhead in work.
%We now compare the two data structures for seed selection: \ptree{} and \wintree{}.
%Recall that both of them allow for evaluating multiple vertices in parallel,
%Both data structures are highly-parallelized, and
%but both evaluate more vertices than CELF to achieve parallelism.
\ptree{} has the theoretical guarantee that the total number of evaluations is no more than twice that in CELF.
%To verify this argument, we plot
To evaluate the work overhead,
we compare the number of evaluations for both data structures to CELF in \cref{fig:trees}(a).
For each graph, we count the number of evaluations done by CELF as $x$ and that by \ptree{} and $\wintree{}$ as $y_1$ and $y_2$, respectively. We plot all $(x,y_1/x)$ as blue dots and all $(x,y_2/x)$ as orange triangles. %We also give the line $y=x$ as a reference.
%Note that $x$ can be viewed as a lower bound of both $y_1$ and $y_2$.
The number of evaluations done by \ptree{s} is very close to that by CELF ($1.03\times$ on average), while \wintree{s} require slightly more ($1.7\times$ on average).
The result is consistent with our theoretical analysis.
%Although \wintree{} may have up to $4\times$ re-evaluations than CELF, such value occurs only when the number of re-evaluations is small,
%and therefore

\begin{figure}
    \centering
    \includegraphics[width=\columnwidth]{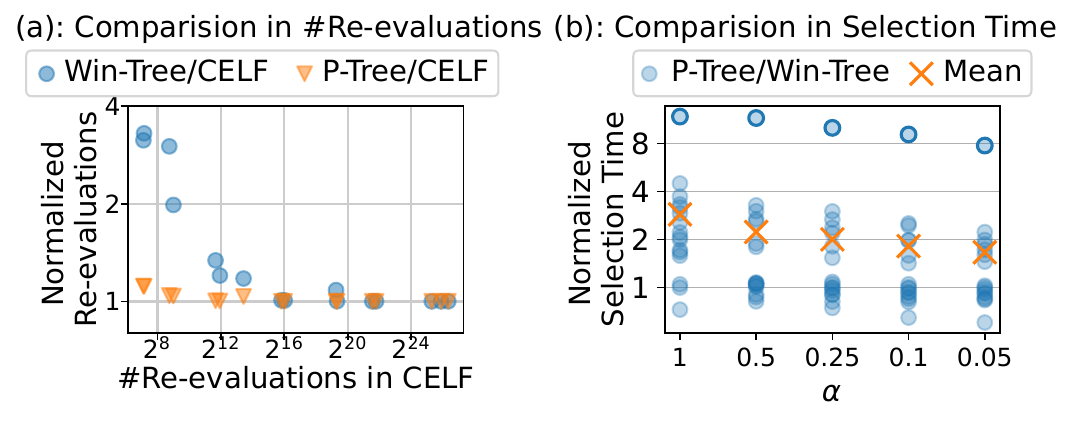}
    \caption{\small \textbf{Compare \ptree and \WinTree{} in the number of re-evaluations and selection time.}
    \textbf{(a)}: Let $\# CELF$, $\#\ptree$ and $\#\WinTree$ represent the total number of re-evaluations needed in the corresponding method. Each point represents a tested graph. Each blue circle is a data point $(\#CELF,\#\WinTree/\#CELF)$. Each orange triangle is a data point $(\#CELF,\#\ptree/\#CELF)$.
    %where $x=\#CELF$,  $y= \#\WinTree/\#CELF$ for blue circle, $y=\#\ptree/\#\WinTree$ for orange triangle.
    \textbf{(b)}: Each point represents a tested graph under different values of $\alpha$. The y-axis is \ptree selection time/\WinTree selection time. The orange cross is the average value under each $\alpha$.
    \hide{The cross is the mean of y under each $\alpha$.}
      % Figure (a) shows the comparison in
     % the number of re-evaluations, where each point represents a test case whose $x$-index is \#re-evaluations needed in CELF, and $y$-index is that needed in \ptree (blue circle) or \WinTree (orange triangle). The red line is $y=x$. A point above it means a test case that parallel priority queues evaluate more vertices than CELF.
     % Figure (b) shows the distribution of \ptree selection time/\WinTree selection time under different $\alpha$.
     %\yihan{For (b), can we circle the topmost points for COS5 or use some special notations to highlight it? We refers to it in the text. }
    } % \vspace{-.5em}
    \label{fig:trees}
\end{figure} 
As mentioned in \cref{sec:WinTree}, we expect \wintree{s} to perform better in practice due to various reasons.
To study this, we plotted \cref{fig:trees}(b) as a comparison for selection time between \ptree{} and \wintree{} under different compression ratios.
Each data point represents a graph, and the value is the ratio of selection time between \ptree{} and \wintree{}.
%\wintree{} is always better than \ptree{} in selection time (the average ratio is larger than 1)
The average ratios are always larger than 1, indicating better performance for \wintree{},
but the advantage decreases as $\rate$ becomes smaller (higher compression).
This is because when~$\rate$ is large, the evaluation is fast, and the major time is on the tree operations,
where \wintree{s} is more advantageous due to the reasons mentioned in \cref{sec:WinTree}.
With higher compression (small~$\rate$), the evaluation becomes more expensive.
Since the \ptree{} evaluates fewer vertices, the overall selection time is more likely to be better.
To further understand this, we focus on the five topmost points (circled) in \cref{fig:trees}(b): those data points are from COS5, where \ptree{} is 8--12$\times$ slower than \wintree{}.
For COS5, the number of vertices $n$ is large, but only hundreds of vertices are evaluated in total.
Thus, the seed selection time is dominated by constructing the data structures, i.e.,
$O(n\log n)$ work for \ptree{} to maintain total ordering, and $O(n)$ for \wintree{} that is a lot faster.
%\yan{Next next few sentences need another round of working. The advantage of \wintree{} can be illustrated by the topmost points (circled) in \cref{fig:trees}(b).
%This data point is from COS5, where \ptree{} is 8-12$\times$ slower than \wintree{}.
%In COS5, the number of vertices $n$ is large, but only hundreds of vertices are re-evaluated,
%so the time is dominated by constructing the data structures, i.e.,
%$O(n\log n)$ work for \ptree{} to maintain total ordering, and $O(n)$ for \wintree{}.}
%In this case, the work to maintain the total order is expensive but redundant,
Since the number of evaluations is small, most of the
vertices in the tree are never touched, and thus, maintaining their order wastes the work.
%Such work is avoided in \wintree{} since the vertices are kept in the original order.
%This is generally true for sparse graphs.

%However, the advantage is quite small, since the total number of re-evaluations of \wintree{} is still close to \ptree{} as shown in \cref{fig: trees}.

In \cref{fig:compression},
we also compare the time and space between \ptree{} and \wintree{}.
%we compare the ratio of selection time between \ptree{} and \wintree{} under different compression ratios.
%With a low compression ($\rate$ close to 1), \wintree{} always outperforms \ptree{} on all graphs.
Similar to the discussions above, when $\rate$ is small, \wintree{} is almost always faster than \ptree{},
but \ptree{} may perform better when $\rate$ is large.
However, the advantage is quite small since the number of evaluations of \wintree{} is still close to \ptree{} (see \cref{fig:trees}(a)).
\wintree{} also uses smaller memory than \ptree{}. As discussed in \cref{sec:WinTree}, this is because \ptree{} needs to explicitly maintain tree pointers and balancing criteria, while each \wintree{} node only needs to store the vertex id ($2n$ integers in total).
%\letong{COS5 runs around $10\times$ faster on \wintree{} than \ptree{}, that because COS5 has large number of vertices while only re-evaluating hundreds of times, the time is dominated by building \ptree{}.}

\hide{In summary, \ptree{} almost always incurs fewer re-evaluations than \wintree{}, but is worse in running time in most cases, especially with low compression ($\rate$ close to 1).
When higher compression (smaller $\rate$), \ptree{} can perform slightly better.
\wintree{} is also more space-efficient than \ptree{}.}
Based on these observations, \oursystem{} always uses \wintree{} by default, but also provides the interface for users to choose \ptree{s}.

\hide{
In \cref{lemma:efficiency} we prove that the number of re-evaluates \ptree{} does is at most twice that of CELF. As shown in \cref{fig:evaluates}, each circle point represents an experiment on a graph, whose $x$-index is the number of re-evaluates CELF does, and $y$-index is that of \ptree. The circles almost lie on the $y=x$ line, which means, in practice, \ptree re-evaluates almost the same number of vertices as CELF.
}

% \myparagraph{Memory Usage}

% \myparagraph{Selection Time}

\hide{
We have two different methods for seed selection, \ptree and \ourtree. Our experiments show that \ourtree is more time-efficient in most of our test cases and is more memory-efficient in all cases (more details in \cref{sec:compare_trees}), so our default algorithm is using \ourtree for seed selection.
We call all other existing algorithms that we compare to the \defn{baselines}. We compare to three existing available influence maximization algorithms \infuser{}~\cite{gokturk2020boosting}, \NoSingles{}~\cite{popova2018nosingles} and \ripples{}~\cite{minutoli2019fast, minutoli2020curipples}.
\infuser{} and our algorithms are based on the static greedy algorithm~\cite{cheng2013staticgreedy}.
\NoSingles{} and \ripples{} are both based on the Reverse Influence Sampling (RIS)~\cite{borgs2014maximizing} idea. Through our experiments, \ripples{} is always better than \NoSingles{} on running time and memory usage, so we only report the results of \ripples{} in our paper.
We need to mention that \infuser{} and our algorithms are only to deal with undirected graphs, while \ripples{} is capable of both directed and undirected graphs.
}

\section{Related Work} \label{sec:related}
\myparagraph{Influence Maximization (IM).} IM has been widely studied for decades with a list of excellent surveys~\cite{banerjee2020survey, peng2018influence,AACGS10,zhou2021survey,li2018influence} that review the applications and papers on this topic.
IM is also of high relevance of the data management community, and many excellent papers published recently regarding benchmarking~\cite{ohsaka2020solution, arora2017debunking}, new algorithms~\cite{guo2020influence}, new propagation models and applications~\cite{huang2022influence, tsaras2021collective, bian2020efficient, xie2022hindering}, and interdisciplinary extensions~\cite{zhu2021analysis,zhang2021grain}.

A survey paper~\cite{li2018influence} roughly categorizes IM algorithms into three methodologies: simulation-based, proxy-based, and sketch-based.
Among them, (Monte Carlo) simulation-based approaches (e.g.,~\cite{kempe2003maximizing,leskovec2007cost,goyal2011celf++, zhou2015upper, wang2010community,dynamicinf}) are the most general and apply to most settings (i.e., graph types and diffusion models); however, they do not take advantage of the specific settings, so generally, their performance is limited.
Proxy-based approaches (e.g.,~\cite{liu2014influence,page1999pagerank,jung2012irie,kimura2006tractable,chen2010scalable,kim2013scalable}) use simpler algorithms/problems (e.g., PageRank or shortest-paths) to solve IM.
While they can be fast in practice, their solution quality has no theoretical guarantees.
Sketch-based solutions, as mentioned in \cref{sec:intro,sec:prelim}, generally have good performance and theoretical guarantees, but only apply to specific diffusion models.
Our sketch compression focuses on the IC model. We note that the parallel data structures in \cref{sec:selecting} are general to submodular diffusion models such as linear threshold~\cite{kempe2003maximizing}, and more~\cite{tsaras2021collective, zhang2021grain}.

Sketch-based algorithms can further be categorized into forward (influence) sketches and reverse (reachable) sketches~\cite{li2018influence}.
As the names suggest, forward sketches record the influence that each vertex can propagate to in sampled graphs.
Most algorithms~\cite{gokturk2020boosting,chen2009efficient, cheng2013staticgreedy,ohsaka2014fast,cohen2014sketch, kempe2003maximizing} mentioned in this paper, including \oursystem{}, use forward sketches.
Reverse sketches find a sample of vertices $T$ and keep the sets of vertices that can reach them.
Many IM algorithms (e.g.,~\cite{borgs2014maximizing,tang2014influence,tang2015influence,wang2016bring}), including \NoSingles{}~\cite{popova2018nosingles} and \ripples{}~\cite{minutoli2019fast, minutoli2020curipples} that we compared to, use reverse sketches.
These algorithms can trade off (lower) solution quality for (better) performance/space, by adjusting the size of $T$. %\yihan{Isn't $S$ the seed set? Why we can adjust it?}
%Unlike these algorithms, \oursystem{} provides the tradeoff between space and performance, without sacrificing the solution quality.
%\yihan{Is it safe to say so, since \ripples{} actually has very good solution quality.}
\revision{
This paper focuses on forward sketches because our compression technique is designed for them. However, the parallel data structures we proposed for seed selection, \ptree{} and \wintree{}, are independent of forward sketches. Applying \ptree{} and \wintree{} to reverse-reachable sketches can be interesting for future work.
}

\myparagraph{Space-Efficient Connectivity.}
We are aware of a few algorithms that can compute graph connectivity using $o(n)$ space~\cite{edmonds1999tight,kosowski2013faster,broder1994trading,BBFGGMS18},
which share a similar motivation with \oursystem{}.
However, \oursystem{} does not require computing connectivity in $o(n)$ space, but only requires storing it in $O(n)$ space.
Hence, the goal here is essentially different, although these approaches are inspiring.

\hide{
\oursystem{} is a \textbf{Sketch-based Method}. Besides, we are also aware of \textbf{Monte-Carlo-Simulation-based Methods} and \textbf{Proxy-based Methods} methods for IM problems.
MC simulation-based methods are generally applicable to all propagation models. However, with all the optimizations to improve the efficiency of the raw greedy algorithm, MC simulation-based methods still suffer from expensive computational costs and are thus very limited in scalability.
Proxy-based methods estimate the influence spread of seed sets using influence ranking proxy~\cite{liu2014influence,page1999pagerank,jung2012irie} or diffusion model reduction proxy~\cite{kimura2006tractable,chen2010scalable,kim2013scalable}  Although the computational overhead is largely reduced, proxy-based methods lack theoretical guarantees to the quality of solutions.
Sketch-based methods improve the theoretical efficiency of the simulation-based methods and preserve the approximation guarantee that proxy-based methods lack.

Based on how the sketches are generated, the survey~\cite{li2018influence} classifies the algorithms of sketch-based methods into two branches: forward influence sketch and reverse reachable sketch. The idea of the forward influence sketch is to construct a sketch by extracting a subgraph that corresponds to an instance of propagation under a specific diffusion model. \oursystem{} is based on the forward influence sketch idea, and we discuss more forward influence sketch algorithms in \cref{sec:prelim} and \cref{tab:related}.
Besides optimized only on undirected graphs, PMC \cite{ohsaka2014fast} is a forward influence sketch algorithm that accelerates directed graphs.
The reverse reachable sketch approach constructs lots of random Reverse Reachable (RR) sets and estimates the influence of a seed set by the number of RR sets containing seeds. One of our main baseline \ripples{} is an RR-sketch method.
Reverse reachable methods~\cite{borgs2014maximizing,tang2014influence,tang2015influence,wang2016bring} are faster than the forward influence method in general because each forward influence sketch is generated by examining the entire graph whereas an RR sketch is constructed by only visiting the nodes that can active the randomly sampled node.
However, a large number of RR sets have to be maintained in the memory for the seed selection to guarantee the quality of seeds, which results in large memory usage. While \oursystem{} can balance the memory and efficiency without loss of quality.
}
\hide{
\myparagraph{Monte-Carlo-Simulation-based Methods}
\citet{kempe2003maximizing} first find the submodularity of influence maximization problem and propose the greedy algorithm, which has a theoretical guarantee of $(1-1/e-\epsilon)$ of the optimal solution for any $\epsilon>0$. For each influence estimation, the necessary number of MC simulations is usually tens of thousands.
\citet{leskovec2007cost} propose the Cost-Effective Lazy Forward (CELF) optimization, which utilize the submodularity to prune the number of influence estimations. CELF++ \cite{goyal2011celf++} further improves the idea of reducing unnecessary MC simulations by computing some sets of estimations in advance. UBLF \cite{zhou2015upper} propose a fast method to compute the upper bound of $\influence(\{v\})$, which can be combined with CELF in the first round of seeds selection. \citet{wang2010community} propose Community based Greedy algorithm (CGA), which use divide-and-conquer to reduce the complexity of individual MC simulation. However, with all the optimizations to improve the efficiency of the raw greedy algorithm, it still suffers from expensive computational costs and is thus very limited in scalability.

\myparagraph{Proxy-based Methods}
Another idea to estimate the influence spread is to use heuristic models instead of heavy MC simulations. \citet{liu2014influence} extends the Page Rank model \cite{page1999pagerank} from a single node to a group of nodes. \citet{jung2012irie} proposes IRIE which estimates the influence with a system of $n$ linear equations.
These ranking models hugely reduced the computational cost of influence estimation, but the transformed problems are not directly related to IM, which causes the inaccuracy of influence estimation.
Some proxy models are directly derived from the diffusion model. SPM \cite{kimura2006tractable} simplifies the IC model based on shortest-path computation, while PMIA \cite{chen2010scalable} and IPA \cite{kim2013scalable} are based on maximum influence arborescence (MIA).
However, these proxy-based methods lack theoretical guarantees.
Even in practice, the quality of solutions of these methods is quite sensitive and unstable on different graph scales and propagation probabilities \cite{li2018influence}.

\myparagraph{Sketch-based Methods}
To overcome the drawbacks of MC-simulation-based methods and heuristic-based methods, others have proposed to first compute sampled graphs (sketches) under the diffusion model, which can be used to speed up seed evaluation.
StaticGreedy \cite{cheng2013staticgreedy} samples $R$ (usually 100 to 200) sketches on the original graph, and these sketches keep static in each round of seed selection.
PMC \cite{ohsaka2014fast} also keeps $R=200$ sketches, but further optimizes the seed selection phase by pruning BFS.
\citet{borgs2014maximizing} are the first to propose the reverse reachable sketch approach (RR-sketch), where an RR-sketch is constructed by randomly selecting a node $v$ and searching on the transpose graph. Then the influence of node $u$ can be estimated by counting the RR-sketches it appears in. Once $u$ is selected as a seed, the RR-sketches containing $u$ are deleted for the following seed selection.
TIM \cite{tang2014influence}, IMM \cite{tang2015influence} and BKRIS \cite{wang2016bring} further optimize the required number of RR-sketches. However, when the graph is large and the required $\epsilon$ is small, the RR-based algorithms need to sample a large number of RR-sketches to ensure the theoretical bound and thus have an expensive overhead of time and memory.

\myparagraph{Parallel Influence Maximization Algorithms}
Considering the growing graph size and plateaued single-processor performance, it is crucial to consider parallel solutions for influence maximization.
\citet{gokturk2020boosting} parallelized the greedy algorithm with hash functions and SIMD vectorization. \cite{popova2018data, popova2018nosingles} parallelized the RR-sketch-based algorithm using Java. \cite{minutoli2019fast, minutoli2020curipples} parallelized IMM, and further expand it to a distributed, multi-CPU and multi-GPU manner.} 
\section{Conclusion and Discussions}\label{sec:conclusion}

This paper addresses the scalability issues in existing IM systems by novel techniques including sketch compression and parallel CELF.
Our sketch compression (\cref{sec:compact_sketching}) applies to the IC model and undirected graphs, 
and avoids the $O(Rn)$ space usage in SOTA systems, 
which allows \oursystem{} to run on much larger graphs without sacrificing much performance. 
%The idea is to combine simulation and memorization. 
%It combines the ideas of simulation and memorization, which we believe is of algorithmic interest.
To the best of our knowledge, our new data structures (\cref{sec:selecting}) are the first parallel version of CELF seed selection, which is general to any submodular diffusion models.

In addition to new algorithms, our techniques are carefully analyzed (\cref{thm:sketch,thm:ptreecost}) and have good theoretical guarantees regarding work, parallelism, and space.
These analyses not only lead to good practical performance but also help to understand how the techniques interplay. 
The techniques are also experimentally verified in \cref{fig:compression,fig:scale,fig:trees}.
The theory and our careful implementation also lead to stable speedup with increasing core counts (see \cref{fig:scale}).
%This means that \oursystem{} will likely have even better performance in the future given stronger processors with more core counts.

\revision{
\myparagraph{Limitations, Generalizations, and Future Work.} Our paper discussed two techniques: sketch compression and parallel CELF. Sketch compression uses the idea of memoization. As in previous work~\cite{chen2009efficient,gokturk2020boosting}, this approach only applies to IC models on undirected graphs, so that each edge can be sampled in advance regardless of the actual influence propagation direction. One future direction is to study similar approaches for sketch compression on directed graphs, but it may require a different technique, such as a compressed representation for strongly connected components. 

For parallel CELF, we believe that our approaches based on \ptree{} and \ourtree{} can be generalized to other settings. %than parallelizing CELF seed selection in IM. 
Within the scope of IM, this technique can also be combined with reverse (reachable) sketches, as long as the diffusion models are submodular. 
Moreover, CELF is a general greedy approach to accelerate optimization problems with submodular objective functions~\cite{krause2014submodular,mirzasoleiman2015lazier,leskovec2007cost}. 
Therefore, our approaches can potentially provide parallelism for these problems. We leave these extensions as future work. %how well our new technique accelerates other CELF's applications as future work.

One limitation of \ourtree{} is that it does not have strong bounds as \ptree{s}. 
An interesting future work is to derive a strong bound regarding the number of evaluations. 
In the worst case, \ourtree{} may require $O(n)$ evaluations due to the asynchrony of the threads. %(i.e., the thread that finds the highest true score can be arbitrarily slow such that all other vertices have been evaluated). 
However, such a bad case is very unlikely in practice, and \ourtree{} has demonstrated good performance in our experiments. 
Giving a tighter bound under some practical assumptions may be interesting.
%, and we leave it as future work. 

\hide{We also note that the high-level ideas behind parallel CELF can be extended to other parallel models that even do not support shared memory access. For example, our prefix-doubling strategy used in the \ptree{} method can also be applied in the distributed setting, providing a distributed search tree that supports batch operations. This can also be interesting for future work.}
}

\hide{
=======
% \section{Conclusion and Future Work}\label{sec:conclusion}
\revision{\section{Limitations and Potential Future Work}\label{sec:conclusion}
}

In this paper, we address the scalability issues in existing IM systems by novel techniques including sketch compression and parallel data structures for seed selection.
\revision{The techniques we proposed are novel and may be applied to other IM algorithms or other problems that have the submodular property. They also have limitations, and we discuss their limitations and potential future work here.}

\revision{\myparagraph{Sketch Compression.} Our sketch compression technique (in \cref{sec:compact_sketching}) avoids the $O(Rn)$ space usage in existing forward-sketch systems, 
which allows \oursystem{} to run on much larger graphs, without sacrificing the performance much. 
However, it requires the sampled graphs (sketches) to be undirected graphs, which essentially asks the edge $(u,v)$ and $(v,u)$ to have the same probability of being sampled. The IC model on undirected graphs is the only propagation model we know that satisfies this requirement. If another novel propagation model also meets the requirement, it can also apply our sketch compression technique. 
}
%The idea is to combine simulation and memorization. 
%It combines the ideas of simulation and memorization, which we believe is of algorithmic interest.

\myparagraph{Parallel CELF.}
Our new data structures, \ptree{} and \ourtree{} (in \cref{sec:selecting}), are the first solutions to parallelize the CELF seed selection step, to the best of our knowledge, which is general to any submodular diffusion models. 
We note that our approaches based on \ptree{} and \ourtree{} can be much more general than parallelizing CELF seed selection in IM. In fact, CELF is a general greedy approach to accelerate optimization problems with submodular objective functions~\cite{krause2014submodular,mirzasoleiman2015lazier,leskovec2007cost}.
We plan to evaluate our data structures for other CELF-based algorithms in future work. %how well our new technique accelerates other CELF's applications as future work.
\hide{
In addition to new algorithmic ideas, our new techniques are carefully analyzed (see \cref{thm:sketch,thm:ptreecost}) and have good theoretical guarantees regarding work, parallelism, and space usage.
These analyses not only lead to good practical performance (see \cref{fig:heatmap}), but also help to understand how the techniques interplay. 
The techniques are also experimentally verified in \cref{fig:compression,fig:scale,fig:trees}.
The theory and our careful implementation also lead to stable speedup with increasing core counts (see \cref{fig:scale}).
This means that \oursystem{} will likely have even better performance in the future given stronger processors with more core counts.
}

\revision{
\myparagraph{Re-evaluations in \WinTree.} We can prove a good theoretical guarantee regarding the number of re-evaluations for \ptree \cref{lemma:efficiency}. However, we can not prove a good bound for \wintree since the parallel threads are highly asynchronous. We do not know whether \wintree has interesting bound regarding the number of re-evaluations. We leave it as future work. 
}

\revision{
\myparagraph{Reverse-reachable Sketch Algorithms}
Our compression technique in \oursystem is designed for forward-sketch algorithms, so in this paper, we focus on forward-sketch algorithms. However, our parallel techniques can also be applied to the other kind of sketch based IM algorithm: reverse-reachable sketch algorithms. Most existing reverse-reachable sketch algorithms actively update the marginal gains of related vertices after selecting a seed. However, we can lazily updates the marginal gains in the reverse-reachable sketch algorithms, thus applying our parallel techniques. We plan to evaluate our data structures for reverse-reachable sketch algorithm in future work.
}

\revision{
\myparagraph{Extend to other Parallel Architectures.}
Our parallel techniques are novel and the high-level idea can be extended to another parallel computational model that does not support shared memory access. For example, if you can find an implementation of binary search tree (BST) that supports split and batch insertion operation on the desired parallel architecture, you can extend our \ptree method by replacing the \ptree with that BST implementation. For \wintree method, you can store sub-trees separatly. When query the \textsc{NextSeed}, the algorithm can find the maximum vertex on each sub-tree and take the maximum among them as the next seed. This method may lead to more computation compared to the original \wintree method, but can still and bring in parallelism avoid re-evaluating all the vertices.
}
}

\iffullversion{
\section{Acknowledgement}
This work is supported by NSF grants CCF-2103483, IIS-2227669, NSF CAREER awards CCF-2238358, CCF-2339310, and UCR Regents Faculty Fellowships.
}

\clearpage
% \bibliography{bib/strings, bib/main}
\bibliography{strings,main}

% \newpage
\iffullversion{
\appendix\ifconference{\section*{Supplemental Material: Performance Evaluation under Different Edge Probability Distributions}\label{sec:distribution}}
\iffullversion{
\section{Performance Evaluation under Different Edge Probability Distributions}\label{sec:distribution}
}
% Table generated by Excel2LaTeX from sheet 'p=uniform'
\begin{table*}[htbp]
  \centering
  \small
    \begin{tabular}{clrrr|rrrr|rrrrr}
          &       & \multicolumn{3}{c|}{\textbf{Relative Influence}} & \multicolumn{4}{c|}{\textbf{Total Running Time (second)}}     & \multicolumn{5}{c}{\textbf{Memory Usage (GB)}} \\
          &       & \multicolumn{1}{c}{\textbf{Ours}} & \multicolumn{1}{@{}c}{\textbf{\infuser{}}} & \multicolumn{1}{c|}{\textbf{\ripples{}}} & \multicolumn{1}{c}{\textbf{\ours{1}}} & \multicolumn{1}{c}{\textbf{\ours{0.1}}} & \multicolumn{1}{@{}c}{\textbf{\infuser}} & \multicolumn{1}{c|}{\textbf{\ripples{}}} & \multicolumn{1}{c}{\textbf{CSR}} & \multicolumn{1}{c}{\textbf{\ours{1}}} & \multicolumn{1}{c}{\textbf{\ours{0.1}}} & \multicolumn{1}{@{}c}{\textbf{\infuser{}}} & \multicolumn{1}{c}{\textbf{\ripples{}}} \\
    \midrule
    \multirow{9}[2]{*}{\begin{sideways}\textbf{Social \& Web}\end{sideways}} & \textbf{EP} & \textbf{11.3K} & 11.3K & 11.2K & 0.36  & 0.49  & \textbf{0.34} & 4.30  & 0.01  & \underline{0.14}  & \textbf{0.05} & 0.16  & 0.46 \\
          & \textbf{SLDT} & \textbf{15.4K} & 15.3K & 15.3K & \textbf{0.35} & 0.52  & 0.39  & 10.4  & 0.01  & \underline{0.13}  & \textbf{0.05} & 0.17  & 0.53 \\
          & \textbf{DBLP} & \textbf{10.4K} & 10.3K & 10.2K & \textbf{0.52} & 0.79  & 2.15  & 7.55  & 0.02  & 0.47  & \textbf{0.08} & 0.67  & \underline{0.40} \\
          & \textbf{YT} & \textbf{86.6K} & 86.6K & 86.1K & \textbf{1.60} & 2.95  & 5.21  & 56.0  & 0.05  & \underline{1.68}  & \textbf{0.27} & 2.37  & 3.84 \\
          & \textbf{OK} & \textbf{2.32M} & 2.32M & 2.32M & \textbf{9.93} & 29.0  & 52.7  & 290   & 1.77  & \underline{6.09}  & \textbf{2.43} & 8.09  & 129 \\
          & \textbf{LJ} & \textbf{1.19M} & 1.19M & 1.19M & \textbf{7.30} & 22.3  & 54.2  & 179   & 0.68  & \underline{6.08}  & \textbf{1.62} & 10.5  & 77.5 \\
          & \textbf{TW} & \textbf{20.1M} & 20.1M & -     & \textbf{122} & 327   & 571   & -     & 18.2  & \underline{63.2}  & \textbf{26.5} & 103   & - \\
          & \textbf{FT} & \textbf{29.4M} & 29.4M & -     & \textbf{150} & 492   & 1213  & -     & 27.4  & \underline{97.6}  & \textbf{39.4} & 161   & - \\
          & \textbf{SD} & \textbf{29.2M} & 29.0M & -     & \textbf{187} & 570   & 1523  & -     & 29.6  & \underline{125}   & \textbf{44.5} & 211   & - \\
    \midrule
    \multirow{2}[2]{*}{\begin{sideways}\textbf{Road}\end{sideways}} 
        & \textbf{GER} & \textbf{0.40K} & 0.29K & 0.38K & 9.26  & \textbf{8.44} & 27.6  & 388  & 0.61  & \underline{13.3}  & \textbf{2.59} & 25.2  & 21.8 \\
        & \textbf{USA} & \textbf{0.39K} & 0.30K & 0.36K & \textbf{14.4} & 13.9  & 64.3  & 8048 & 0.33  & \underline{25.8}  & \textbf{5.06} & 49.1  & 45.0 \\
    \midrule
    \multirow{5}[2]{*}{\begin{sideways}\textbf{k-NN}\end{sideways}} 
          & \textbf{HT5} & \textbf{1.07K} & 0.89K & 1.02K & 0.74  & \textbf{0.71} & 2.71  & 9.09  & 0.11  & 1.37  & \textbf{0.23} & 1.94  & \underline{0.77} \\
          & \textbf{HH5} & \textbf{2.93K} & 2.31K & 2.85K & 2.28  & \textbf{2.10} & 9.92  & 14.6  & 0.11  & 3.04  & \textbf{0.51} & 4.27  & \underline{1.14} \\
          & \textbf{CH5} & \textbf{360K} & 320K  & 355K  & \textbf{3.61} & 5.58  & 118.24 & 9.82  & 0.25  & 4.88  & \textbf{1.04} & 8.80  & \underline{1.63} \\
          & \textbf{GL5} & \textbf{11.6K} & 8.5K  & 11.5K & 19.3  & \textbf{17.7} & 94.5  & 231   & 1.36  & 27.6  & \textbf{6.02} & 51.7  & \underline{8.17} \\
          & \textbf{COS5} & \textbf{5.0K} & 2.1K  & -     & \textbf{310} & 322   & 2024  & -     & 17    & \underline{355}   & \textbf{66} & 666   & - \\
    \bottomrule
    \end{tabular}%

    \caption{\textbf{Under uniform edge distribution: running time, memory usage, and the influence spread of all tested systems on a machine with 96 cores (192 hyperthreads).}
    For social and web graphs, the edges are sampled under the uniform distribution $U(0,0.1)$; for road and k-NN graphs, the edges are sampled under $U(0.1, 0.3)$.
    ``-'': out of memory (1.5 TB) or time limit (3 hours).
    \ours{1} is our implementation with \wintree{} without compression. 
    {\ours{0.1}} is our implementation with $\rate=0.1$ ($10\times$ compression for sketches).
    {\infuser{}}~\cite{gokturk2020boosting} and {\ripples{}}~\cite{minutoli2019fast, minutoli2020curipples} are baselines.
    We report the \emph{best time} of \infuser{} and \ripples{} by varied core counts (the scalability issue of \infuser{} and \ripples{} are shown in \cref{fig:scale}).
    {CSR} is the memory used to store the graph in CSR format (see more in \cref{sec:overall_performance}).
    The bold numbers are the highest influence spread/fastest time/smallest memory among all implementations on each graph.
    %The bold numbers in memory usage are the smallest memory of graphs among all implementations.
    The underlined numbers in memory usage are the smallest memory among systems that do not use compression (\ours{1}, \ripples{} and \infuser{}).
    }
  \label{tab:baselines_uniform}%
\end{table*}%

% Table generated by Excel2LaTeX from sheet 'WIC'
\begin{table*}[htbp]
  \centering
  \small
    \begin{tabular}{clrrr|rrrr|rrrrr}
          &       & \multicolumn{3}{c|}{\textbf{Relative Influence}} & \multicolumn{4}{c|}{\textbf{Total Running Time (second)}}     & \multicolumn{5}{c}{\textbf{Memory Usage (GB)}} \\
          &       & \multicolumn{1}{c}{\textbf{Ours}} & \multicolumn{1}{@{}c}{\textbf{\infuser{}}} & \multicolumn{1}{c|}{\textbf{\ripples{}}} & \multicolumn{1}{c}{\textbf{\ours{1}}} & \multicolumn{1}{c}{\textbf{\ours{0.1}}} & \multicolumn{1}{@{}c}{\textbf{\infuser}} & \multicolumn{1}{c|}{\textbf{\ripples{}}} & \multicolumn{1}{c}{\textbf{CSR}} & \multicolumn{1}{c}{\textbf{\ours{1}}} & \multicolumn{1}{c}{\textbf{\ours{0.1}}} & \multicolumn{1}{@{}c}{\textbf{\infuser{}}} & \multicolumn{1}{c}{\textbf{\ripples{}}} \\
    \midrule
    \multirow{9}[1]{*}{\begin{sideways}\textbf{Social \& Web}\end{sideways}} 
          & \textbf{EP} & \textbf{615} & 605   & 598   & 0.30  & 0.41  & \textbf{0.28} & 1.19  & 0.01  & 0.13  & \textbf{0.04} & 0.16  & \underline{0.09} \\
          & \textbf{SLDT} & \textbf{645} & 628   & 629   & 0.29  & 0.35  & \textbf{0.28} & 0.77  & 0.01  & 0.13  & \textbf{0.04} & 0.17  & \underline{0.09} \\
          & \textbf{DBLP} & \textbf{738} & 636   & 713   & \textbf{0.37} & 0.42  & 1.04  & 3.18  & 0.02  & 0.47  & \textbf{0.08} & 0.67  & \underline{0.31} \\
          & \textbf{YT} & \textbf{678} & 652   & 656   & \textbf{1.00} & 1.08  & 3.34  & 13.5  & 0.05  & 1.68  & \textbf{0.27} & 2.37  & \underline{1.22} \\
          & \textbf{OK} & 1747  & 1622  & \textbf{1810} & 6.81  & \textbf{6.77} & 80.3  & 65.5  & 1.77  & 6.21  & \textbf{2.43} & 8.10  & \underline{5.62} \\
          & \textbf{LJ} & 1199  & 1054  & \textbf{1211} & \textbf{5.23} & 5.23  & 32.9  & 75.6  & 0.68  & 6.04  & \textbf{1.63} & 10.5  & \underline{4.98} \\
          & \textbf{TW} & \textbf{990} & 509   & -     & \textbf{73.2} & 75.6  & 416   & -     & 18.2  & \underline{62.8}  & \textbf{25.6} & 103   & - \\
          & \textbf{FT} & \textbf{1425} & 1324  & -     & 95.7  & \textbf{93.2} & 1352  & -     & 27.4  & \underline{97.6}  & \textbf{39.2} & 161   & - \\
          & \textbf{SD} & \textbf{3762} & 3675  & -     & \textbf{117} & 121   & 1924  & -     & 29.6  & \underline{125}   & \textbf{44.4} & 211   & - \\
    \midrule
    \multirow{2}[2]{*}{\begin{sideways}\textbf{Road}\end{sideways}} 
        & \textbf{GER} & \textbf{430} & 385   & 408     & 10.9  & \textbf{9.87} & 35.3  & 334 & 0.33  & \underline{13.26} & \textbf{2.62} & 25.2  & 21.0 \\
        & \textbf{USA} & \textbf{422} & 405   & 411    & 13.9  & \textbf{12.9} & 86.0  & 6283  & 0.61  & \underline{25.80} & \textbf{4.99} & 49.1  & 43.0 \\
         
    \midrule
    \multirow{5}[2]{*}{\begin{sideways}\textbf{k-NN}\end{sideways}} 
          & \textbf{HT5} & \textbf{440} & 397   & 404   & 0.73  & \textbf{0.68} & 2.50  & 15.8  & 0.11  & 1.35  & \textbf{0.23} & 1.94  & \underline{1.24} \\
          & \textbf{HH5} & \textbf{518} & 453   & 491   & 1.97  & \textbf{1.83} & 7.55  & 45.4  & 0.11  & 3.05  & \textbf{0.49} & 4.27  & \underline{2.69} \\
          & \textbf{CH5} & \textbf{755} & 596   & 733   & 3.03  & \textbf{2.95} & 16.4  & 88.1  & 0.25  & \underline{4.83}  & \textbf{1.05} & 8.80  & 4.87 \\
          & \textbf{GL5} & \textbf{489} & 406   & 459   & 18.2  & \textbf{17.0} & 74.1  & 3493 & 1.36  & \underline{27.6}  & \textbf{6.05} & 51.7  & 45.3 \\
          & \textbf{COS5} & \textbf{519} & 453   & -     & 294   & \textbf{267} & 1729  & -     & 17    & \underline{355}   & \textbf{66.0} & 666   & - \\
    \bottomrule
    \end{tabular}%
    \caption{\textbf{Under vertex-related edge distribution: running time, memory usage, and the influence spread of all tested systems on a machine with 96 cores (192 hyperthreads).}
    For an undirected-edge $(u,v)$, the probability it is sampled is $\frac{2}{d_u+d_v}$, where $d_u$ and $d_v$ are the degrees of vertex $u$ and $v$. 
    ``-'': out of memory (1.5 TB) or time limit (3 hours).
    \ours{1} is our implementation with \wintree{} without compression. 
    {\ours{0.1}} is our implementation with $\rate=0.1$ ($10\times$ compression for sketches).
    {\infuser{}}~\cite{gokturk2020boosting} and {\ripples{}}~\cite{minutoli2019fast, minutoli2020curipples} are baselines.
    We report the \emph{best time} of \infuser{} and \ripples{} by varied core counts (the scalability issue of \infuser{} and \ripples{} are shown in \cref{fig:scale}).
    {CSR} is the memory used to store the graph in CSR format (see more in \cref{sec:overall_performance}).
    The bold numbers are the highest influence spread/fastest time/smallest memory among all implementations on each graph.
    %The bold numbers in memory usage are the smallest memory of graphs among all implementations.
    The underlined numbers in memory usage are the smallest memory among systems that do not use compression (\ours{1}, \ripples{} and \infuser{}).
    }
  \label{tab:baselines_wic}%
\end{table*}%

\ifconference{
    (This section is added to the appendix of our \textbf{full paper} to address Question 2 in \cref{sec:exp-cover} of the cover page (Improve the Experimental Analysis). We also provide it here for the reviewers for easier access.)
}

To study how different edge probability distributions affect our algorithm and baseline algorithms, we also tested two more edge probability distributions. The first one uses $U(0,0.1)$ for each edge in social and web graphs, and $U(0.1,0.3)$ for road and $k$-NN graphs, where $U(x,y)$ means to draw a uniform random number from $x$ to~$y$. 
This edge probability assignment is also commonly used in previous work~\cite{gokturk2020boosting, minutoli2019fast}. 
The other probability distribution is taking $p_{uv}=\frac{2}{d_u+d_v}$, where $p_{uv}$ is the probability to sample edge $(u,v)$, and $d_u$ and $d_v$ are degrees of vertex $u$ and $v$. It is similar to the Weighted IC (WIC) model on the directed graph, where $p_{uv} = \frac{1}{d_{in}(v)}$, $p_{uv}$ is the probability that the directed edge $u\to v$ is sampled, and $d_{in}(v)$ is the in-degree of vertex $v$. We will use \textit{Uniform} and \textit{WIC} to refer to these two edge probability assignments, respectively, and use \textit{Consistent} to refer to the assignment mentioned in the main context of the paper.

\cref{tab:baselines_uniform} and \cref{tab:baselines_wic} show the influence scores, running time and memory usage of all systems with \textit{Uniform} and \textit{WIC} edge probability assignments respectively. \ours{1} and \ours{0.1} are our \oursystem{} with $\rate=1$ (no compression) and $\rate=0.1$ ($10\times$ sketch compression), respectively. As mentioned in the main paper, \infuser{} and \ripples{} have scalability issues.  We tested them with both 192 hyperthreads and the same core counts to get their best performance in \cref{tab:baselines} for each graph. In \cref{tab:baselines_uniform} and \cref{tab:baselines_wic}, we report the smaller one as the running time.

% \subsection{Comparison between systems}
\myparagraph{Influence Score.}
With $R=256$ for \oursystem{} and \infuser{}, and $\epsilon=0.5$ for \ripples{}, \oursystem{} has the largest influence score on all 16 tested graphs with \textit{Uniform} edge probability assignment, and is smaller than \ripples{} by $3.5\%$ and $1.0\%$ on graph OK and LJ with \textit{WIC} edge probability assignment. The influence scores of \ripples{} and \ours{} differ less than $8\%$, which indicates $R=256$ and $\epsilon=0.5$ is a fair setting for baselines to compare their running time and memory. We observe that the influence score of \infuser{} is only 41--89\% of the best-achieved score by \ripples{} and \oursystem{} on sparse graphs with \textit{Uniform} edge assignment and $51\%$ of the best influence on TW with \textit{WIC} edge assignment. The observation is the same with the \textit{Consistent} assignment in the main body of the paper.

\myparagraph{Running Time.}
\oursystem{} is significantly faster than two baseline algorithms on almost all graphs. 
%As mentioned, we report the running time with the best core counts setting for \infuser{} and \ripples{}, because they may not scale to 192 threads. 
With both \textit{Uniform} and \textit{WIC} assignments, \ours{1} is faster than all baselines on 14 graphs. \ours{1} and \ours{0.1} are just slower than \infuser{} on the two smallest graphs.  On average, with the \textit{Uniform} assignment, \ours{1} is $4.6\times$ faster than \infuser{} and $21\times$ faster than \ripples{}, and \ours{0.1} is $2.9\times$ faster than \infuser{} and $15\times$ faster than \ripples. With \textit{WIC} assignment, \ours{1} is $4.5\times$ faster than \infuser{} and $21\times$ faster than \ripples{}, and \ours{0.1} is $4.4\times$ faster than \infuser{} and $20\times$ faster than \ripples{}.

\myparagraph{Memory Usage.}
\oursystem{} and \infuser{} have almost the same memory usage under different edge probability assignments. Recall that \oursystem{} and \infuser{} are forward-reachable sketch-based IM algorithms. Their memory usage for the same graph only depends on the number of sketches $R$. Therefore, with the same $R=256$, they will have almost the same memory usage. 
Different from \oursystem{} and \infuser{}, \ripples{} is a reverse-reachable sketch-based IM algorithm, which dynamically samples sketches according to the $\epsilon$ and graphs. \ripples{} memory usage varies on different edge probability assignments. One observation is that on social graphs YT, OK and LJ, the memory usage of \ripples{} with WIC is much smaller than that with \textit{Uniform} and \textit{Consistent} edge probability assignments ($10\times$ fewer than \textit{Uniform} and $4.1\times$ fewer than \textit{Consistent} on average on these three graphs). 
With $\alpha=0.1$, \oursystem{} uses the least memory on all graphs with all different assignments.

\myparagraph{Summary.}
Overall, the performance comparison between \oursystem{} and the baseline algorithms are similar on \textit{Consistent} and \textit{Uniform} assignments. 
\ripples{} exhibits slightly different performance on \textit{WIC} to \textit{Consistent}, because it  samples reverse-reachable sketches dynamically. %sketches while \oursystem and \infuser sample fixed $R$ sketches. 
Both \ours{1} and \ours{0.1} have shorter running time than the baselines in time on \emph{all but the two smallest graphs} for each tested edge probability assignment. In those exceptions, \ours{1} has a running time very close to \infuser{}, which is much shorter than \ripples{}. 
\ours{0.1} has lower space usage on \emph{all} tested edge probability assignments and graphs. 
In summary, the relative performance is fairly consistent on these three edge probability distributions. Therefore, in the main body of the paper, we simply use the \textit{Consistent} assignment (fixing $p$ for all edges) to demonstrate the performance. 

}

\end{document}
\endinput